\pdfoutput=1
\documentclass[11pt,letterpaper]{article}

\usepackage[margin=1in]{geometry}
\usepackage{amsmath,amssymb,amsthm,mathtools}
\usepackage{booktabs,tabularx,array}
\newcolumntype{Y}{>{\raggedright\arraybackslash}X}
\newcolumntype{P}[1]{>{\raggedright\arraybackslash}p{#1}}
\usepackage{enumitem}
\usepackage{needspace}
\usepackage{flafter}
\usepackage[section]{placeins}
\usepackage{microtype}
\usepackage{xcolor}
\usepackage{xurl}
\usepackage{graphicx}
\usepackage{wrapfig}
\usepackage{tikz}
\usetikzlibrary{arrows.meta,calc,positioning,angles,quotes,decorations.pathreplacing}
\definecolor{figInk}{HTML}{202832}
\definecolor{figBlue}{HTML}{00548A}
\definecolor{figGreen}{HTML}{1A6B28}
\definecolor{figBlueText}{HTML}{003D6E}
\definecolor{figGreenText}{HTML}{125420}
\definecolor{figOrangeText}{HTML}{7A4100}
\definecolor{figOrange}{HTML}{A86800}
\definecolor{figGray}{HTML}{414B55}
\definecolor{figLight}{HTML}{596774}
\definecolor{figFill}{HTML}{EEF2F5}
\tikzset{
  paperfig/.style={x=1cm,y=1cm,font=\normalsize,text=black,line cap=round,line join=round},
  fpt/.style={circle,fill=figInk,inner sep=1.4pt},
  fpto/.style={circle,draw=figInk,fill=white,line width=.8pt,inner sep=1.3pt},
  fdisk/.style={draw=figGray,line width=1pt},
  frail/.style={draw=figBlue,line width=1.5pt},
  fgate/.style={draw=figOrange,line width=1.8pt},
  fdashgate/.style={draw=figOrange,line width=1.8pt,dash pattern=on 4pt off 2.5pt},
  ftri/.style={draw=figGray,line width=.85pt},
  faxis/.style={draw=black,line width=.8pt},
  fguide/.style={draw=figLight,line width=.85pt,dash pattern=on 4pt off 2.5pt},
  fpencil/.style={draw=figLight,line width=.85pt,dash pattern=on 4pt off 2.5pt},
  fold/.style={draw=figBlue,line width=1.45pt},
  fnew/.style={draw=figGreen,line width=1.45pt},
  fthick/.style={draw=figInk,line width=.9pt},
  fangle/.style={draw=figInk,line width=.8pt},
  fvec/.style={draw=figInk,line width=.9pt},
  fbrace/.style={draw=figGray,line width=.85pt,decorate,decoration={brace,amplitude=4pt}},
}
\usepackage{aliascnt}
\usepackage{hyperref}
\usepackage[nameinlink,noabbrev,capitalise]{cleveref}

\hypersetup{
  colorlinks=true,
  linkcolor=blue!45!black,
  citecolor=blue!45!black,
  urlcolor=blue!45!black,
  pdftitle={The Stretch Factor of Planar Delaunay Triangulations Is Nearly Determined},
  pdfauthor={Guanlin Mo, Kangke Cheng, Hu Ding}
}
\setlist{leftmargin=2em,itemsep=0.25em,topsep=0.4em}
\allowdisplaybreaks
\newif\iftracked
\trackedfalse
\usepackage[normalem]{ulem}
\definecolor{revisionDeleted}{HTML}{555555}
\iftracked

  \DeclareRobustCommand{\deleted}[1]{%
    {\color{revisionDeleted}\hypersetup{citecolor=revisionDeleted,linkcolor=revisionDeleted,urlcolor=revisionDeleted}%
     \let\svcref\cref
     \let\svCref\Cref
     \let\sveqref\eqref
     \let\svcrefrange\crefrange
     \let\svcite\cite
     \def\cref##1{\mbox{\svcref{##1}}}%
     \def\Cref##1{\mbox{\svCref{##1}}}%
     \def\eqref##1{\mbox{\sveqref{##1}}}%
     \def\crefrange##1##2{\mbox{\svcrefrange{##1}{##2}}}%
     \def\cite##1{\mbox{\svcite{##1}}}%
     \sout{#1}}}
\else

  \DeclareRobustCommand{\deleted}[1]{}
\fi

\DeclareRobustCommand{\Hcond}[1]{\texorpdfstring{\hyperlink{cond:H#1}{\textup{(H#1)}}}{(H#1)}}

\newtheorem{maintheorem}{Theorem}
\crefname{maintheorem}{Theorem}{Theorems}
\newtheorem{theorem}{Theorem}[section]
\newcommand{\newaliastheorem}[3]{%
  \newaliascnt{#1}{theorem}%
  \newtheorem{#1}[#1]{#2}%
  \aliascntresetthe{#1}%
  \crefname{#1}{#2}{#3}%
  \Crefname{#1}{#2}{#3}%
}
\newaliastheorem{lemma}{Lemma}{Lemmas}
\newaliastheorem{proposition}{Proposition}{Propositions}
\newaliastheorem{corollary}{Corollary}{Corollaries}
\theoremstyle{definition}
\newaliastheorem{definition}{Definition}{Definitions}
\newaliastheorem{remark}{Remark}{Remarks}
\numberwithin{equation}{section}

\newcommand{\R}{\mathbb R}
\newcommand{\norm}[1]{\lVert #1\rVert}
\newcommand{\abs}[1]{\lvert #1\rvert}
\newcommand{\clip}{\operatorname{clip}}
\newcommand{\pow}{\operatorname{pow}}
\newcommand{\sinc}{\operatorname{sinc}}
\newcommand{\cD}{\mathcal D}
\newcommand{\cE}{\mathcal E}
\newcommand{\cA}{\mathcal A}
\newcommand{\cB}{\mathcal B}
\newcommand{\cH}{\mathcal H}
\newcommand{\Hs}{H_\star}
\newcommand{\doi}[1]{\href{https://doi.org/#1}{\mbox{\nolinkurl{https://doi.org/}}\allowbreak\mbox{\nolinkurl{#1}}}}

\title{The Stretch Factor of Planar Delaunay Triangulations\\ Is Nearly Determined}
\author{%
\begin{tabular}{@{}l@{\hspace{2em}}r@{}}
\normalsize\bfseries Guanlin Mo\textsuperscript{1,\(\dagger\)}
  & \small\scshape moguanlin@mail.ustc.edu.cn\\
\normalsize\bfseries Kangke Cheng\textsuperscript{1,\(\dagger\)}
  & \small\scshape ke314159@mail.ustc.edu.cn\\
\normalsize\bfseries Hu Ding\textsuperscript{1,*}
  & \small\scshape huding@ustc.edu.cn\\[0.4em]
\multicolumn{2}{@{}l@{}}{\small\itshape
  \textsuperscript{1}University of Science and Technology of China, Hefei, China}
\end{tabular}}
\date{}

\begin{document}
\maketitle

\begin{center}
\small \textsuperscript{\(\dagger\)}Guanlin Mo and Kangke Cheng contributed equally.\qquad
\textsuperscript{*}Corresponding author.
\end{center}

\begin{abstract}
The Delaunay triangulation is a classical geometric structure on a finite set of points in the plane.  Its stretch factor is the largest ratio of the shortest-path distance along its edges to the Euclidean distance between the same two points.  Determining the worst-case stretch factor of planar Delaunay triangulations is a long-standing open problem.  We prove a uniform upper bound of \(1.5949\), improving the previous bound of \(1.998\) and narrowing the gap to the best known lower bound of \(1.5932\) from \(0.4048\) to less than \(0.0017\), nearly pinpointing the worst-case stretch factor.  A central difficulty is that the line segment between two points may cross arbitrarily many triangles.  Our \emph{progress accounting} framework addresses this by converting the global bound into a family of local inequalities, each involving only a fixed number of real parameters regardless of the number of triangles crossed.  Two further geometric lemmas make computer-assisted verification tractable: the \emph{pencil lemma} decomposes each local transition into two geometrically independent stages, reducing verification over infinitely many possible next disks to a single differential condition; and the \emph{corner lemmas} handle degenerate boundary cases where direct verification would fail.  With these reductions, the proof ends with a finite verification of an explicit potential in exact rational arithmetic and outward-rounded interval arithmetic.  Numerous geometric network constructions use the Delaunay stretch factor as a black box, so the tighter bound directly improves stretch-factor guarantees for a range of downstream applications, including bounded-degree plane spanners and low-weight spanners.
\end{abstract}

\clearpage
\setcounter{tocdepth}{1}
\begingroup

\tableofcontents
\endgroup
\clearpage

\section{Introduction}\label{sec:intro}

The Delaunay triangulation is a fundamental structure in computational geometry: it is the dual graph of the Voronoi diagram and appears throughout geometric algorithms~\cite{deberg2008}.  To define it, consider a finite set \(S\) of points in the plane.  A triangulation of \(S\) partitions the convex hull into nonoverlapping triangles whose vertex set is exactly \(S\).  A triangulation \(\cD\) of \(S\) is called \emph{Delaunay} if the open circumdisk of every triangle contains no point of \(S\) (\cref{fig:delaunay})~\cite{delaunay1934}; such a triangulation exists whenever \(S\) is not collinear~\cite{deberg2008}.

Given a Delaunay triangulation \(\cD\) of \(S\) and two points \(u,v\in S\), denote their Euclidean distance by \(\norm{u-v}\) and the length of a shortest path from \(u\) to \(v\) along the edges of \(\cD\) by \(d_{\cD}(u,v)\).  The \textbf{\emph{stretch factor}} of \(\cD\) is
\begin{equation}\label{eq:stretch-factor}
\max_{u,v\in S,\ u\ne v}\ \frac{d_{\cD}(u,v)}{\norm{u-v}}.
\end{equation}

\begin{figure}[!htbp]
\centering
\begin{tikzpicture}[paperfig]
\path[use as bounding box] (-0.10,0.10) rectangle (6.85,4.02);
\fill[figFill] (0.350,2.000) -- (1.904,0.876) -- (2.730,2.859) -- cycle;
\fill[figFill] (2.730,2.859) -- (1.904,0.876) -- (4.222,1.090) -- cycle;
\fill[figFill] (2.730,2.859) -- (4.222,1.090) -- (4.676,3.029) -- cycle;
\fill[figFill] (4.676,3.029) -- (4.222,1.090) -- (6.650,2.000) -- cycle;
\draw[ftri] (0.408,3.179) -- (2.730,2.859) -- (3.310,3.677) -- cycle;
\draw[ftri] (0.421,2.652) -- (0.408,3.179) -- (0.090,2.484) -- cycle;
\draw[ftri] (2.730,2.859) -- (0.408,3.179) -- (0.421,2.652) -- cycle;
\draw[ftri] (2.730,2.859) -- (1.904,0.876) -- (4.222,1.090) -- cycle;
\draw[ftri] (3.310,3.677) -- (2.730,2.859) -- (4.676,3.029) -- cycle;
\draw[ftri] (4.676,3.029) -- (4.383,3.623) -- (3.310,3.677) -- cycle;
\draw[ftri] (2.730,2.859) -- (4.222,1.090) -- (4.676,3.029) -- cycle;
\draw[ftri] (6.650,2.000) -- (4.383,3.623) -- (4.676,3.029) -- cycle;
\draw[ftri] (4.676,3.029) -- (4.222,1.090) -- (6.650,2.000) -- cycle;
\draw[ftri] (0.350,2.000) -- (1.904,0.876) -- (2.730,2.859) -- cycle;
\draw[ftri] (2.730,2.859) -- (0.421,2.652) -- (0.350,2.000) -- cycle;
\draw[ftri] (0.350,2.000) -- (0.421,2.652) -- (0.090,2.484) -- cycle;
\draw[ftri] (1.950,0.334) -- (1.904,0.876) -- (0.350,2.000) -- cycle;
\draw[ftri] (3.898,0.601) -- (1.904,0.876) -- (1.950,0.334) -- cycle;
\draw[ftri] (3.898,0.601) -- (4.222,1.090) -- (1.904,0.876) -- cycle;
\draw[ftri] (6.650,2.000) -- (4.222,1.090) -- (3.898,0.601) -- cycle;
\draw[fdisk,draw=figBlue,dash pattern=on 3pt off 2pt] (1.641,2.149) circle[radius=1.300];
\draw[fdisk,draw=figBlue,dash pattern=on 3pt off 2pt] (3.008,1.580) circle[radius=1.309];
\draw[fdisk,draw=figBlue,dash pattern=on 3pt off 2pt] (3.766,2.220) circle[radius=1.218];
\draw[fdisk,draw=figBlue,dash pattern=on 3pt off 2pt] (5.320,1.856) circle[radius=1.338];
\draw[fgate] (1.904,0.876) -- (2.730,2.859);
\draw[fgate] (2.730,2.859) -- (4.222,1.090);
\draw[fgate] (4.222,1.090) -- (4.676,3.029);
\draw[fthick] (0.350,2.000) -- (6.650,2.000);
\foreach \x/\y in {0.350/2.000,6.650/2.000,1.904/0.876,0.408/3.179,3.898/0.601,1.950/0.334,0.421/2.652,4.383/3.623,0.090/2.484,3.310/3.677,4.222/1.090,2.730/2.859,4.676/3.029}
  \node[fpt] at (\x,\y) {};
\node[anchor=east] at (0.270,1.980) {$u$};
\node[anchor=west] at (6.730,1.980) {$v$};
\end{tikzpicture}
\caption{A Delaunay triangulation of \(13\) points.  The segment \([u,v]\) crosses the four shaded triangles.  The dashed circles are their circumcircles, and no input point lies inside any of them.  Consecutive shaded triangles share the three thick edges.}
\label{fig:delaunay}
\end{figure}

Determining the exact worst-case stretch factor of planar Delaunay triangulations remains a long-standing open problem~\cite{xia2013,bosesmid2013}.  The problem has attracted sustained interest for nearly four decades, yet progress on the upper bound has been slow: it decreased from roughly \(5.08\) in 1987~\cite{dfs1990} to \(2.42\) in 1989~\cite{keilgutwin1989,keilgutwin1992}, and then to \(1.998\) in 2011~\cite{xia2013}, where it stood for over a decade.  The best known lower bound is \(1.5932\), due to Xia and Zhang~\cite{xiazhang2011}.  \Cref{tab:bounds} summarizes the history of these bounds.

\begin{table}[!htbp]
\centering
\begin{tabularx}{\textwidth}{@{}lllY@{}}
\toprule
& Year & Bound & Reference\\
\midrule
 Lower bounds & 1989 & \(\pi/2\approx1.5708\) & Chew~\cite{chew1989}\\
& 2009 & \(1.5846\) & \cite{bose2011}\\
& 2011 & \(1.5932\) & Xia and Zhang~\cite{xiazhang2011}\\
\midrule
 Upper bounds & 1987 & \((1+\sqrt5)\pi/2\approx5.08\) & Dobkin, Friedman and Supowit~\cite{dfs1990}\\
& 1989 & \(2\pi/(3\cos(\pi/6))\approx2.42\) & Keil and Gutwin~\cite{keilgutwin1989,keilgutwin1992}\\
& 2011 & \(1.998\) & Xia~\cite{xia2013}\\
& 2026 & \(1.5949\) & this paper (\cref{thm:main})\\
\bottomrule
\end{tabularx}
\caption{Bounds on the stretch factor of Delaunay triangulations of finite planar point sets, listed in the order in which they were proved.  The year is that of first publication.}
\label{tab:bounds}
\end{table}

\Needspace{8\baselineskip}
We substantially improve the upper bound:

\begin{maintheorem}\label{thm:main}
For every finite point set \(S\subset\R^2\) and every Delaunay triangulation \(\cD\) of \(S\),
\[
d_{\cD}(u,v)\le 1.5949\;\norm{u-v}\qquad\text{for all }u,v\in S.
\]
\end{maintheorem}

The proof establishes the slightly stronger bound
\(U_0=1.59485994\). Together with the lower bound \(1.5932\),
this reduces the gap between the known bounds from \(0.4048\)
to less than \(0.0017\). The upper and lower bounds now agree
to within \(0.11\%\), nearly pinpointing the exact worst-case
value. We obtain the bound by reducing a global
path estimate to local inequalities, proved by geometric
arguments and finite computer-assisted verification~\cite{moore2009}.

\paragraph{Consequences of a Uniform Stretch Bound.}
Networks in applications such as wireless sensing or geographic routing must connect scattered points with short paths while limiting the number of connections at each node.  A \textbf{\emph{plane \(t\)-spanner}} formalizes this requirement: it is a straight-line graph with no crossing edges in which the distance along the graph between every pair of vertices is at most \(t\) times their Euclidean distance~\cite{narasimhan2007,bosesmid2013}.  We denote by \(t_{\cD}\) the worst-case stretch factor of Delaunay triangulations, that is, the supremum of~\eqref{eq:stretch-factor} over all finite sets \(S\subset\R^2\) and all Delaunay triangulations \(\cD\) of \(S\).  Every Delaunay triangulation is thus a plane \(t_{\cD}\)-spanner, and many constructions build sparser subgraphs on top of it~\cite{bose2005,bsx2009,kanj2010,bcc2012,li2003,kp2008}, often inheriting the factor \(t_{\cD}\) in their stretch guarantees.  \Cref{thm:main} tightens every such guarantee at once.

An independently studied open problem asks: what is the smallest \(t\) such that every finite planar point set admits a plane \(t\)-spanner~\cite[p.~471]{narasimhan2007}?  The best upper bounds have all come from Delaunay triangulations. Ours brings this bound from \(1.998\) to \(1.5949\), while the best known lower bound is \(1.44116\)~\cite{fkp2025}.  Bose, Gudmundsson and Smid~\cite{bose2005} construct a plane spanner with maximum degree~\(27\) and stretch factor at most \((\pi+1)t_{\cD}\), whose total edge length is at most a constant times the weight of a minimum spanning tree; our bound reduces the stretch guarantee from~\(8.27\) to~\(6.61\).  Kanj and Perković~\cite{kp2008} construct, for every integer \(k\ge14\), a plane spanner with maximum degree \(k\) and stretch at most \(\bigl(1+\frac{2\pi}{k\cos(\pi/k)}\bigr)t_{\cD}\).  For unit-disk graphs, which model wireless networks where two nodes can communicate whenever they are within a given range, they also give a distributed, strictly localized construction suited to wireless sensor networks~\cite{bmsu2001,li2003}.  Its stretch bound depends on the stretch factor of unit Delaunay graphs rather than on \(t_{\cD}\).  \Cref{tab:consequences} summarizes these guarantees.

\begin{table}[!htb]
\centering\small
{
\begin{tabularx}{\textwidth}{@{}P{3.8cm}Yll@{}}
\toprule
 Question & Guarantee in terms of \(t_{\cD}\) & \(t_{\cD}\!=\!1.998\) & \(t_{\cD}\!=\!1.5949\)\\
\midrule
 Minimum stretch of plane spanners~\cite{narasimhan2007,fkp2025}
  & Upper bound \(t_{\cD}\); lower bound \(1.44116\)
  & \(\le1.998\) & \(\le1.5949\)\\[0.35em]
Bounded-degree, low-weight plane spanners~\cite{bose2005}
  & Max.\ degree \(27\), stretch \(\le(\pi\!+\!1)t_{\cD}\)
  & \(\le8.27\) & \(\le6.61\)\\[0.35em]
Bounded-degree plane spanners~\cite{kp2008}
  & Max.\ degree \(k\ge14\), stretch \(\le\bigl(1\!+\!\frac{2\pi}{k\cos(\pi/k)}\bigr)t_{\cD}\)
  & \(\le2.92\)\rlap{\({}^{\dagger}\)} & \(\le2.33\)\rlap{\({}^{\dagger}\)}\\
\bottomrule
\multicolumn{4}{@{}l}{\footnotesize\({}^{\dagger}\)Evaluated at \(k=14\); the bound improves as \(k\) increases.}
\end{tabularx}}
\caption{Representative guarantees that improve when the Delaunay stretch bound \(t_{\cD}\) is tightened from \(1.998\) to \(1.5949\).}
\label{tab:consequences}
\end{table}

The proof of \cref{thm:main} begins with a reformulation introduced by Xia~\cite{xia2013}: the global triangulation is replaced by the ordered circumdisks of the triangles crossed by a segment.  We adopt this reformulation but analyze the resulting chains with a different framework.  To describe it, fix two points \(u,v\in S\).  The required bound is \(d_{\cD}(u,v)\le U_0\norm{u-v}\).  The segment \([u,v]\) can cross arbitrarily many triangles, so a uniform estimate on each crossed triangle does not suffice.  Consider the circumdisks of the triangles whose interiors meet \([u,v]\), ordered along the segment.  These disks form a \textbf{\emph{disk chain}}: two consecutive circles meet at the endpoints of their shared triangulation edge, which we call a \textbf{\emph{gate}}.  We denote these disks by \(O_1,\dots,O_n\) and the gate between \(O_i\) and \(O_{i+1}\) by \(G_i\).  \Cref{fig:intro-chain} illustrates this construction.

\begin{figure}[!htb]
\centering
\begin{tikzpicture}[paperfig]
\path[use as bounding box] (8.2,-0.12) rectangle (15.4,4.02);
\draw[fdisk] (10.09,2.11) circle[radius=1.672];
\draw[fdisk] (11.769,1.889) circle[radius=1.675];
\draw[fdisk] (13.421,1.73) circle[radius=1.73];
\draw[frail] (8.419,2.056) arc[start angle=-178.164,end angle=-307.806,radius=1.672];
\draw[frail] (11.114,3.431) arc[start angle=113,end angle=57,radius=1.675];
\draw[frail] (12.681,3.294) arc[start angle=115.299,end angle=10.881,radius=1.73];
\draw[frail,draw=figGreen] (8.419,2.056) arc[start angle=181.836,end angle=292.806,radius=1.672];
\draw[frail,draw=figGreen] (10.738,0.569) arc[start angle=232,end angle=292,radius=1.675];
\draw[frail,draw=figGreen] (12.397,0.336) arc[start angle=233.701,end angle=370.881,radius=1.73];
\draw[fthick] (8.419,2.056) -- (15.119,2.056);
\draw[fgate] (11.114,3.431) -- (10.738,0.569);
\draw[fgate] (12.681,3.294) -- (12.397,0.336);
\node[fpt] at (8.419,2.056) {};
\node[fpt] at (15.119,2.056) {};
\node[fpt] at (11.114,3.431) {};
\node[fpt] at (10.738,0.569) {};
\node[fpt] at (12.681,3.294) {};
\node[fpt] at (12.397,0.336) {};
\node[anchor=east] at (8.349,2.056) {$u$};
\node[anchor=west] at (15.189,2.056) {$v$};
\node[anchor=east,text=figOrangeText] at (10.781,1.427) {$G_1$};
\node[anchor=west,text=figOrangeText] at (12.552,1.223) {$G_2$};
\node[text=black] at (8.72,3.82) {$O_1$};
\node[text=black] at (11.769,3.784) {$O_2$};
\node[text=black] at (14.799,3.109) {$O_3$};
\end{tikzpicture}
\caption{Three disks form a chain from \(u\) to \(v\).  Consecutive circles meet at the endpoints of a gate \(G_i\).  A path from \(u\) to \(v\) along the chain can follow the blue or green circular arcs and switch between them along a gate.  The shortest such path is the chain path.}
\label{fig:intro-chain}
\end{figure}

A path from \(u\) to \(v\) can travel along the circular arcs of the chain and switch between arcs by crossing a gate.  We call the shortest such path the \textbf{\emph{chain path}} and denote its length by \(P(u,v)\).  Every arc can be replaced by a path in the triangulation that is no longer, so \(d_{\cD}(u,v)\le P(u,v)\).

For a real constant \(U\), we say the \textbf{\emph{chain inequality}} holds at \textbf{\emph{level}} \(U\) if
\begin{equation}\label{eq:chain-ineq-intro}
P(u,v)\le U\norm{u-v}
\end{equation}
for every disk chain.  Proving the stretch bound \(U_0\) for every Delaunay triangulation \(\cD\) thus reduces to establishing the chain inequality at level \(U_0\).  \Cref{def:chain,def:ladder} in \cref{sec:chains} define disk chains and the chain inequality formally, and \cref{thm:reduction} states the reduction, which we prove in \cref{app:reduction}.

\paragraph{Obstacles to a Direct Proof.}  First, specifying the centers and radii of a chain of \(n\) disks already requires \(3n\) real parameters, and \(n\) can be arbitrarily large.  Second, the chain inequality does not admit a separate bound for each disk.  As \cref{fig:intro-obstacle}(a) illustrates, the first gate may cross \([u,v]\) very close to \(u\); reaching that gate then requires a long circular arc despite little advance along the segment.  The resulting arc-length cost can exceed \(U\) times the progress along \([u,v]\), as \cref{fig:intro-obstacle}(b) shows.  The proof therefore needs to allow an excess on one disk to be offset by savings from other disks.

\begin{figure}[!htbp]
\centering
\resizebox{0.98\linewidth}{!}{%
\begin{tikzpicture}[paperfig]
\path[use as bounding box] (-2.55,-1.55) rectangle (14.25,3.60);
\node[font=\normalsize] at (-2.30,3.30) {\textbf{(a)}};
\draw[fdisk] (2,1) circle[radius=2.5];
\draw[fthick] (-0.291288,0) -- (5.047641,0);
\draw[frail] (-0.291288,0) arc[start angle=203.557,end angle=143.130,radius=2.5];
\draw[fgate] (0,2.5) -- (0,-0.5);
\node[fpt] at (-0.291288,0) {};
\node[fpt] at (5.047641,0) {};
\node[fpt] at (0,2.5) {};
\node[fpt] at (0,-0.5) {};
\node[fpto] at (0,0) {};
\node[anchor=east] at (-0.391288,0) {$u$};
\node[anchor=west] at (5.147641,0) {$v$};

\node[anchor=west,text=figOrangeText] at (0.16,1.05) {$G_1$};
\node[text=black] at (3.75,3.25) {$O_1$};
\node[text=figBlueText,anchor=east] at (-0.60,2.35) {arc to \(G_1\)};
\draw[fguide] (-0.291288,0) -- (-0.291288,-0.86);
\draw[fguide] (0,-0.5) -- (0,-0.86);
\draw[frail,draw=figGreen,dash pattern=on 3pt off 2pt] (-0.291288,-0.86) -- (0,-0.86);
\node[anchor=north,text=figGreenText] at (-0.145,-1.00) {\(\text{progress}\)};
\begin{scope}[xshift=8cm]
\node[font=\normalsize] at (-0.75,3.22) {\textbf{(b)}};
\node[anchor=east,text=figBlueText] at (-0.20,2.10) {arc-length cost};
\draw[frail,line width=3.2pt] (0,2.10) -- (5.60,2.10);
\node[anchor=east,text=figGreenText] at (-0.20,0.80) {\(U\times\text{progress}\)};
\draw[frail,draw=figGreen,line width=3.2pt] (0,0.80) -- (1.00,0.80);
\draw[fguide] (1.00,0.55) -- (1.00,2.35);
\draw[{Latex[length=1.5mm]}-{Latex[length=1.5mm]},draw=figGray,line width=.85pt] (1.18,1.45) -- (5.42,1.45);
\node[anchor=south,text=black] at (3.30,1.52) {positive excess};
\node[anchor=north] at (2.80,-0.05) {\(\text{cost}>U\times\text{progress}\)};
\end{scope}
\end{tikzpicture}}
\caption{The first disk can cost more than \(U\) times its progress.  (a)~When the first gate crosses \([u,v]\) close to \(u\), a path from \(u\) along the chain still needs a long circular arc to reach that gate.  (b)~The bars compare the quantities schematically.  The blue bar shows the arc-length cost, and the green bar shows \(U\) times the progress.  The arrow marks the positive excess that other disks must offset.}
\label{fig:intro-obstacle}
\end{figure}

\Needspace{8\baselineskip}
Xia~\cite{xia2013} addresses these difficulties through an inductive estimate with a terminal correction.  More precisely, the analysis introduces a \emph{nonpositive} correction term \(\Phi\) and obtains
\begin{equation}\label{eq:xia-bound}
P(u,v)<1.8\norm{u-v}-\Phi.
\end{equation}
Because \(\Phi\le0\), the term \(-\Phi\) cannot simply be discarded.  Xia supplies a separate estimate relating \(\Phi\) to \(P(u,v)\), and incorporating it changes the constant from \(1.8\) to \(1.998\).  The gap between \(1.8\) and \(1.998\) reflects a structural limitation: the global correction couples the total path length to the bound, so tightening one requires tightening the other.

Our approach eliminates this coupling.  We distribute the adjustment among the individual local steps, rather than leaving a global correction that must be bounded separately.  The local adjustments cancel when summed along the path, so no separate correction remains.  With the global correction eliminated, the proof reduces entirely to verifying a family of local inequalities, each in a fixed number of parameters.  This decoupling is what allows a much tighter constant.

\subsection[Overview of Our Method]{Overview of Our Method}\label{subsec:our-method}

Our proof introduces three main ingredients:
\begin{enumerate}
\item A \emph{progress accounting} framework that converts the chain inequality into a family of local conditions, each in a fixed number of parameters regardless of how many disks the chain contains.
\item The \emph{pencil lemma}, which decomposes each local transition into a gate slide that determines cost and progress, and a pencil slide through the one-parameter family of circles sharing the new gate.  The potential's monotonicity along the pencil is governed by a single differential inequality, so one analytic condition at each gate state controls all possible continuations at once.
\item The \emph{corner lemmas}, which handle degenerate boundary configurations where the gate lies on the line through \(u\) and \(v\); near these configurations the condition produced by the pencil lemma has zero margin, and interval subdivision alone cannot certify it.
\end{enumerate}
Together, they yield a proof that the chain inequality holds at level \(U_0\).  The complete chain of reductions, from \cref{thm:main} to the final computer-assisted verification, is traced in \cref{sec:proof-overview}.

\paragraph{Progress Accounting.}
Our goal is to prove the chain inequality at level \(U = U_0\).  The key step is to reduce this global statement about arbitrarily long chains to a family of local inequalities, each in a fixed number of real parameters that does not depend on the chain length.  Once the problem takes this form, a finite computer-assisted verification becomes possible in principle.  Our \textbf{\emph{progress accounting}} framework achieves this reduction.  It decomposes the chain inequality into one local comparison per disk, and introduces a potential that transfers savings between disks.  The basic idea is as follows.  Each disk contributes a \textbf{\emph{cost}} (the path distance it requires) and some \textbf{\emph{progress}} (the advance it provides along the segment).  If every disk's cost were at most \(U\) times its progress, the chain inequality would follow at once.  However, some disks overspend, and their excess must be offset by savings from other disks.  A potential on the boundary between consecutive disks carries these savings forward; when the potentials are summed, intermediate terms cancel.  We now describe the decomposition precisely.  Denote the disks by \(O_1,\dots,O_n\), ordered along \([u,v]\).  Let \(G_i=[a_i,b_i]\), where \(a_i\) and \(b_i\) are the upper and lower endpoints of the gate, and let \(c_i\) be the point where \(G_i\) crosses \([u,v]\), with \(c_0=u\) and \(c_n=v\).  The \emph{progress} of \(O_i\) is \(\norm{c_i-c_{i-1}}\).  \Cref{fig:intro-accounting} illustrates this decomposition.

A path from \(u\) along the chain may reach a gate at either endpoint, so we record the shortest distances from \(u\) to the two endpoints.  The mean and difference capture two distinct aspects: the mean tracks the average distance accumulated so far, while the difference tracks the imbalance between the upper and lower routes.  Formally, let \(d_i^+\) and \(d_i^-\) be the shortest distances from \(u\) to the upper and lower endpoints of \(G_i\) within the chain so far, and define
\begin{equation}\label{eq:intro-mean-diff}
\mu_i=\tfrac12(d_i^++d_i^-),\qquad \delta_i=d_i^+-d_i^-;
\end{equation}
see \cref{fig:intro-gate-state}.  The \emph{cost} of \(O_i\) is \(\mu_i-\mu_{i-1}\).  The mean therefore gives a cost that telescopes along the chain, while the difference retains the information needed for the later distance update.  Since \(\mu_0=0\) and \(\mu_n=P(u,v)\), the costs sum to \(P(u,v)\), while the progress contributions sum to \(\norm{u-v}\).  Hence
\begin{equation}\label{eq:intro-split}
P(u,v)-U\norm{u-v}
=\sum_{i=1}^{n}\bigl(\mathrm{cost}_i-U\cdot\mathrm{progress}_i\bigr).
\end{equation}
To prove the chain inequality, it suffices to show that the entire sum on the right-hand side of~\eqref{eq:intro-split} is at most \(0\).  We cannot do so by proving that each summand is at most \(0\): the first summand, for example, may be positive, as \cref{fig:intro-obstacle} illustrates.  A positive summand represents an excess, because that disk costs more than \(U\) times its progress.  It must therefore be balanced by negative summands, which represent savings on other disks.

\begin{figure}[!htb]
\centering
\begin{tikzpicture}[paperfig]
\path[use as bounding box] (-0.45,-1.75) rectangle (9.65,3.10);
\coordinate (aleft) at (2,1.60);
\coordinate (bleft) at (2,-0.40);
\coordinate (aright) at (7.2,1.60);
\coordinate (bright) at (7.2,-0.40);
\coordinate (cleft) at (2,0);
\coordinate (cright) at (7.2,0);
\draw[frail] (aleft) -- (aright);
\draw[frail,draw=figGreen] (bleft) -- (bright);
\draw[fgate] (aleft) -- (bleft);
\draw[fgate] (aright) -- (bright);
\draw[fthick] (0.45,0) -- (8.85,0);
\node[fpt] at (aleft) {};
\node[fpt] at (bleft) {};
\node[fpt] at (aright) {};
\node[fpt] at (bright) {};
\node[fpto] at (cleft) {};
\node[fpto] at (cright) {};
\node[anchor=east] at (1.84,1.66) {$a_{i-1}$};
\node[anchor=west] at (7.36,1.66) {$a_i$};
\node[anchor=east] at (1.84,-0.46) {$b_{i-1}$};
\node[anchor=west] at (7.36,-0.46) {$b_i$};
\node[anchor=west,text=figOrangeText] at (2.14,0.78) {$G_{i-1}$};
\node[anchor=west,text=figOrangeText] at (7.34,0.78) {$G_i$};
\node[anchor=south east] at (1.88,0.10) {$c_{i-1}$};
\node[anchor=south west] at (7.32,0.10) {$c_i$};
\node at (4.6,2.78) {disk \(O_i\)};
\node[text=figBlueText] at (4.6,2.12) {\(\mathrm{cost}_i=\mu_i-\mu_{i-1}\)};
\draw[fguide] (cleft) -- (2,-0.96);
\draw[fguide] (cright) -- (7.2,-0.96);
\draw[frail,draw=figGreen,dash pattern=on 3pt off 2pt] (2,-0.96) -- (7.2,-0.96);
\node[anchor=north,text=figGreenText] at (4.6,-1.08) {\(\mathrm{progress}_i=\norm{c_i-c_{i-1}}\)};
\end{tikzpicture}
\caption{The disk contributes a cost and a progress along \([u,v]\).  The rest of the chain is omitted.  The blue and green lines schematically represent the circular arcs that paths along the chain use on this disk.  The cost of the disk is the increase of \(\mu\), and its progress is the distance along \([u,v]\) between the crossings of its two gates.}
\label{fig:intro-accounting}
\end{figure}

\begin{figure}[!htb]
\centering
\resizebox{0.94\linewidth}{!}{%
\begin{tikzpicture}[paperfig]
\path[use as bounding box] (-0.35,-0.45) rectangle (14.55,3.65);
\node at (0,3.40) {\textbf{(a)}};
\coordinate (u) at (0.45,1.65);
\coordinate (a) at (3.20,2.75);
\coordinate (b) at (3.20,0.55);
\draw[frail] (u) -- (1.30,2.75) -- (a);
\draw[frail,draw=figGreen] (u) -- (1.30,0.55) -- (b);
\draw[fgate] (a) -- (b);
\node[fpt] at (u) {};
\node[fpt] at (a) {};
\node[fpt] at (b) {};
\node[anchor=east] at (0.30,1.65) {$u$};
\node[anchor=west] at (3.36,2.83) {$a_i$};
\node[anchor=west] at (3.36,0.47) {$b_i$};
\node[anchor=west,text=figOrangeText] at (3.36,1.65) {$G_i$};
\node[text=figBlueText,anchor=south] at (1.95,2.83) {$d_i^+$};
\node[text=figGreenText,anchor=north] at (1.95,0.47) {$d_i^-$};
\begin{scope}[xshift=7.0cm]
\node at (-0.20,3.40) {\textbf{(b)}};
\draw[frail,line width=2.2pt] (0.80,2.45) -- (5.80,2.45);
\draw[frail,draw=figGreen,line width=2.2pt] (0.80,1.10) -- (3.80,1.10);
\node[fpt] at (0.80,2.45) {};
\node[fpt] at (5.80,2.45) {};
\node[fpt] at (0.80,1.10) {};
\node[fpt] at (3.80,1.10) {};
\node[anchor=east,text=figBlueText] at (0.58,2.45) {\(d_i^+\)};
\node[anchor=east,text=figGreenText] at (0.58,1.10) {\(d_i^-\)};
\draw[fguide] (4.80,0.85) -- (4.80,2.72);
\node[anchor=south] at (4.80,2.76) {\(\mu_i\)};
\draw[fguide] (3.80,1.10) -- (3.80,0.35);
\draw[fguide] (5.80,2.45) -- (5.80,0.35);
\draw[draw=figGray,line width=.9pt] (3.80,0.35) -- (5.80,0.35);
\node[anchor=north] at (4.80,0.23) {\(\delta_i=d_i^+-d_i^-\)};
\end{scope}
\end{tikzpicture}}
\caption{Two shortest-path distances are recorded at each gate.  (a)~The distances \(d_i^+\) and \(d_i^-\) are the lengths of shortest paths from \(u\) to the upper and lower endpoints of \(G_i\) within the chain so far.  The paths are schematic.  (b)~The two distances are drawn on a common scale.  Their mean is \(\mu_i\), and their difference is \(\delta_i\).  The increase of \(\mu\) across a disk gives its cost, while \(\delta\) carries the information needed for the next distance update.}
\label{fig:intro-gate-state}
\end{figure}

\Needspace{6\baselineskip}
We need to transfer these savings between disks without examining the whole chain at every step.  A \textbf{\emph{gate state}} \(s_i\) records the information that a step inherits from earlier disks: the current gate and the difference~\(\delta_i\).  It also records the disk \(O_{i+1}\) beyond the gate, whose cost is still to be paid.  The mean \(\mu_i\) itself need not be stored, because its increments have already been recorded as costs.  A step can therefore be described without referring to the disks before its current gate.  To pass savings from one disk to the next, we define a potential on these local states.  Assign a value \(W(s)\) to each gate state \(s\), and require the decrease in \(W\) across each disk to cover that disk's summand in~\eqref{eq:intro-split}:
\begin{equation}\label{eq:intro-step}
\mathrm{cost}_i-U\cdot\mathrm{progress}_i
\le W(s_{i-1})-W(s_i)\qquad\text{for }1\le i\le n.
\end{equation}
Thus, if a disk has a positive summand, a sufficient decrease in the potential pays for its excess.  If it has a negative summand, the potential may increase, carrying the saving forward to later disks.  When we add these inequalities, the potentials at intermediate gates cancel.  The potential is zero at the two ends, so the entire sum in~\eqref{eq:intro-split} is at most \(0\).

The chain inequality now reduces to constructing a single function \(W(s)\) for which~\eqref{eq:intro-step} holds for every possible local step at level \(U=U_0\).  \Cref{thm:criterion,prop:exact} show that, for \(U\ge\pi/2\), this local criterion is equivalent to the chain inequality. In the converse
direction, a potential is obtained by taking, at each gate state,
the supremum of the remaining excess cost over all finite ways
of completing the chain. The proof in \cref{app:exactness} shows
that this quantity is finite whenever the chain inequality holds
and gives the smallest feasible potential. Each local inequality
involves at most nine real parameters, independently of the
chain length.

\paragraph{Pencil Lemma.}
The local inequality must hold for every possible next disk, of which there are infinitely many.  The pencil lemma resolves this by identifying a two-stage structure in each local transition.  In the first stage, the gate slides along the current circle to its new position.  This determines the step's cost and progress entirely: whichever circle passes through the two endpoints of the new gate, the arc lengths and the advance along the segment are the same.  The possible next circles through these two endpoints form a \textbf{\emph{pencil}}, a classical one-parameter family.  In the second stage, the next disk sweeps through this pencil while the gate remains fixed.  The \textbf{\emph{pencil lemma}} computes the exact rate at which the potential changes during this sweep: a single expression, the pencil residual of \cref{def:residual}.  When this residual is nonpositive, the potential does not increase along the pencil, so one sign condition at each gate state controls all possible next disks at once.  The lemma holds for every potential, so it applies to whichever candidate potential the numerical search later produces.  \Cref{fig:pencil} illustrates the construction.

The pencil lemma also has a substantial computational benefit.  A local inequality involves at most nine real parameters: three for the current disk, one for each end of its gate on the current circle, one for the difference~\(\delta\), and three for the next disk.  Two symmetries of the configuration, translation along the segment and scaling, remove two of them, reducing the count to seven.  The pencil lemma removes the three parameters of the next disk, leaving the sufficient conditions of \cref{sec:pencil}, which involve four angular coordinates.  The final verification at four parameters, in \cref{sec:verification}, uses approximately \(1.04\) million boxes.  At seven parameters, a comparable resolution per axis would require on the order of \(10^{10}\) boxes, and at the original nine, roughly \(10^{13}\).

\begin{figure}[!htbp]
\centering
\begin{tikzpicture}[paperfig]
\path[use as bounding box] (-0.1,0.45) rectangle (15.5,5.45);
\draw[faxis,-{Latex[length=1.8mm]}] (0,2.6) -- (4.65,2.6);
\draw[fdisk] (2.3,3) circle[radius=1.75];
\draw[frail,-{Latex[length=1.8mm]}] (0.784,3.875) arc[start angle=150,end angle=58,radius=1.75];
\draw[frail,-{Latex[length=1.8mm]}] (0.832,2.047) arc[start angle=213,end angle=292,radius=1.75];
\draw[fgate] (0.784,3.875) -- (0.832,2.047);
\draw[fgate] (3.227,4.484) -- (2.956,1.377);
\node[fpt] at (0.784,3.875) {};
\node[fpt] at (0.832,2.047) {};
\node[fpt] at (3.227,4.484) {};
\node[fpt] at (2.956,1.377) {};
\node[fpt] at (2.3,3) {};
\node[anchor=south] at (2.3,3.1) {$o$};
\node[anchor=south east] at (0.704,3.925) {$a_G$};
\node[anchor=north east] at (0.752,1.997) {$b_G$};
\node[anchor=south west] at (3.287,4.544) {$a'$};
\node[anchor=north west] at (3.016,1.317) {$b'$};
\node[anchor=west,text=figOrangeText] at (0.908,3.311) {$G$};
\node[anchor=east,text=figOrangeText] at (2.937,2.309) {$G'$};
\node[text=black] at (0.87,4.52) {$O$};
\draw[faxis,-{Latex[length=1.8mm]}] (6.7,2.6) -- (13.65,2.6);
\draw[fpencil] (9.628,2.945) circle[radius=1.568];
\draw[fpencil] (10.265,2.889) circle[radius=1.63];
\draw[fold] (9,3) circle[radius=1.75];
\draw[fnew] (10.893,2.834) circle[radius=1.911];
\draw[fguide] (8.303,3.061) -- (11.64,2.769);
\draw[fvec,-{Latex[length=1.8mm]}] (9.12,2.99) -- (10.773,2.845);
\draw[fgate] (9.927,4.484) -- (9.656,1.377);
\node[fpt] at (9.927,4.484) {};
\node[fpt] at (9.656,1.377) {};
\node[fpt] at (9,3) {};
\node[fpt] at (10.893,2.834) {};
\node[anchor=south east] at (8.95,3.1) {$o$};
\node[anchor=north] at (10.893,2.5) {$o'$};
\node[anchor=south east] at (9.867,4.564) {$a'$};
\node[anchor=north east] at (9.596,1.297) {$b'$};
\node[anchor=west,text=figOrangeText] at (9.837,2.309) {$G'$};
\node[text=figBlueText] at (7.65,4.45) {$O$};
\node[text=figGreenText] at (12.5,4.44) {$O'$};
\node[anchor=west] at (0,5.25) {(a)};
\node[anchor=west] at (6.3,5.25) {(b)};
\end{tikzpicture}
\caption{A step is performed in two stages.  (a)~The gate moves along the circle of the current disk to the chord shared with the next disk.  (b)~The disk then moves through circles that pass through the two endpoints of the new gate.  Its center follows the perpendicular bisector of that gate.  The blue and green circles show the current and next disks, and the dashed circles show two intermediate positions.}
\label{fig:pencil}
\end{figure}

\paragraph{Corner Lemmas.}
For a fixed candidate potential, interval arithmetic certifies an inequality over a small region when the inequality holds with a positive \emph{margin}.  Near certain degenerate configurations, however, the margin of the condition produced by the pencil lemma tends to zero, so subdivision alone may never certify it.  The \textbf{\emph{corner lemmas}} of \cref{sec:verification} address this obstruction by replacing the vanishing-margin condition near these configurations with bounds on first derivatives that hold with a positive margin.

With these tools in hand, we use numerical optimization to search for a candidate potential \(W(s)\) with rational coefficients.  The correctness of the proof does not depend on the numerical optimization: once a candidate potential with rational coefficients has been found, its validity at level \(U_0\) is established by a finite, self-contained verification in exact rational arithmetic and outward-rounded interval arithmetic~\cite{moore2009}.  The data, programs, and software versions used in this verification are publicly archived at \doi{10.5281/zenodo.23835635}~\cite{mcd2026package}.  \Cref{tab:proofmap} summarizes the main reductions and their formal statements; \cref{tab:chain} in \cref{sec:proof-overview} extends this summary to a finer chain of reductions that tracks the number of parameters at each stage.  \Cref{sec:chains} presents the reduction to disk chains, \cref{sec:accounting} develops progress accounting, and \crefrange{sec:pencil}{sec:verification} construct and certify the potential.

\begin{table}[!htbp]
\centering\small
\begin{tabularx}{\textwidth}{@{}P{3.25cm}YP{4.35cm}@{}}
\toprule
 Form of the problem & Purpose of the reduction & Location and key statement\\
\midrule
\textbf{Disk chain}\newline \(3n\) parameters
  & Replace paths in the triangulation by shortest paths along circular arcs and gates.
  & \Cref{sec:chains}, \cref{thm:reduction}\\
\addlinespace[1.2em]
\textbf{One local step}\newline \(9\) parameters
  & Replace the chain inequality by a local potential criterion that is exact.
  & \Cref{sec:accounting}, \cref{thm:criterion}\\
\addlinespace[1.2em]
\textbf{Angular form}\newline \(7\) parameters
  & Remove translation and scaling, and describe each gate state by four angles.
  & \Cref{subsec:angles}, \cref{lem:range}\\
\addlinespace[1.2em]
\textbf{\allowbreak Sufficient conditions on four angles}\newline \(4\) parameters
  & Use the pencil lemma to eliminate the next disk from the verification conditions.
  & \Cref{sec:pencil}, \cref{thm:sufficient}\\
\addlinespace[1.2em]
\textbf{Explicit potential}\newline certified
  & Construct the potential and verify every condition, including the corner cases.
  & \Cref{sec:family,sec:verification}, \cref{thm:explicit}\\
\bottomrule
\end{tabularx}
\caption{Overview of the proof.  Each row describes a reduction and gives its location and key formal statement.}
\label{tab:proofmap}
\end{table}

\FloatBarrier

\subsection{Related Work}\label{subsec:related}

\paragraph{Bounds on the Stretch Factor.}
Dobkin, Friedman and Supowit~\cite{dfs1990} first proved that Delaunay triangulations are spanners, with the upper bound \((1+\sqrt5)\pi/2\approx5.08\).  Das and Joseph~\cite{dasjoseph1989} extended the result to other natural triangulations, including minimum-weight and greedy triangulations.  Keil and Gutwin~\cite{keilgutwin1989,keilgutwin1992} improved this bound to \(2\pi/(3\cos(\pi/6))\approx2.42\).  On the lower-bound side, Chew~\cite{chew1989} proved the bound \(\pi/2\).  These bounds stood for about twenty years, and Chew's lower bound was widely believed to be tight.  The lower bound \(1.5846\) of~\cite{bose2011} disproved this conjecture, and Xia and Zhang~\cite{xiazhang2011} improved it to \(1.5932\).  In 2011, Xia~\cite{xia2013} proved the upper bound \(1.998\), the first below \(2\), using the amortized analysis described above.

For points in convex position, meaning that every input point is a vertex of the convex hull, \cref{thm:main} also improves the best known upper bound.  In this special case, the best known lower bound is \(1.5810\)~\cite{bose2011}.  The upper bound \(2.33\) of Cui, Kanj and Xia~\cite{cui2011} was successively improved to \(1.84\)~\cite{tan2023} and \(1.77\)~\cite{tan2024}. Our bound \(1.5949\) improves it further.

\paragraph{Delaunay Triangulations for Other Distance Functions.}
Replacing the Euclidean disk in the definition of a Delaunay triangulation by another convex shape gives other plane spanners, several of which have exactly known stretch factors.  Chew initiated this line of work: for the \(L_1\)-metric he proved the bound \(\sqrt{10}\)~\cite{chew1986}, and for the equilateral-triangle distance he showed that the stretch factor is exactly \(2\)~\cite{chew1989}.  The \(L_1\)- and \(L_\infty\)-Delaunay triangulations have stretch factor exactly \(\sqrt{4+2\sqrt2}\approx2.61\)~\cite{bghp2015}, and hexagon-Delaunay triangulations have stretch factor exactly \(2\)~\cite{pdt2021}.  Exact values for rectangles and parallelograms were obtained more recently~\cite{vrssw2023,bdn2025}.  Bose, Carmi, Collette and Smid~\cite{bccs2010} proved that the Delaunay graph of every convex distance function is a spanner.  For the Euclidean disk, however, the exact value remains unknown.

\paragraph{Plane and Other Geometric Spanners.}
A related open problem asks for the smallest \(t\) such that every finite planar point set has a plane spanner with stretch factor at most \(t\)~\cite[p.~471]{narasimhan2007}.  The best known upper bounds for this problem have all come from Delaunay triangulations: \(2\) from Chew's triangular-distance construction~\cite{chew1989}, \(1.998\) from Xia's bound~\cite{xia2013}, and now \(1.5949\) from \cref{thm:main}.  The best known lower bound is \(1.44116\)~\cite{fkp2025}.  Earlier lower bounds were obtained from regular polygons: every plane spanner of the regular \(23\)-gon has stretch factor at least \(1.4308\)~\cite{dg2016}.  The bound \(1.44116\) for the regular \(84\)-gon was found by an exact algorithm for minimum-dilation triangulations~\cite{fkp2025}.  Computing a plane graph of minimum stretch factor is NP-hard when the number of edges is bounded~\cite{gkkkm2010}. Whether a minimum-dilation triangulation can be computed in polynomial time remains open~\cite{fkp2025}.

Spanners that need not be plane can have stretch arbitrarily close to \(1\), and work on them focuses on other parameters, such as size, degree, weight and hop diameter.  Greedy spanners~\cite{addjs1993} and spanners based on well-separated pair decompositions~\cite{ck1995} achieve stretch \(1+\varepsilon\) with \(O(n)\) edges for any fixed \(\varepsilon>0\), and serve as starting points for constructions with additional degree and weight guarantees~\cite{narasimhan2007}.  Arya, Mount and Smid~\cite{ams1994} developed geometric spanners with small hop diameter, and~\cite{admss1995} studied simultaneous bounds on degree, weight, and hop diameter.  Yao graphs~\cite{yao1982} and \(\Theta\)-graphs use sufficiently fine local angular partitions to obtain constant stretch with bounded outdegree~\cite{keilgutwin1992,narasimhan2007}.  These constructions offer different trade-offs among stretch, sparsity, degree, weight, and construction cost, without requiring the graph to be plane~\cite{narasimhan2007}.  In higher dimensions, the classical \((1+\varepsilon)\)-spanner constructions based on well-separated pair decompositions have size exponential in the ambient dimension~\(d\).  Har-Peled, Indyk and Sidiropoulos~\cite{his2013} sidestepped this barrier: for any \(n\) points in \(\R^d\), they constructed near-linear-size spanners with stretch \(O(\sqrt{\log n})\), independent of~\(d\).  In a separate line of work, Elkin and Solomon~\cite{es2015} constructed Euclidean \((1+\varepsilon)\)-spanners that simultaneously achieve near-optimal bounds on size, weight, and degree, and Le and Solomon~\cite{ls2019} further tightened these to optimal lightness.

\section{Disk Chains and the Reduction}\label{sec:chains}

In this section, we define disk chains and their ladder graphs, and we reduce \cref{thm:main} to an inequality for disk chains (\cref{thm:reduction}).  The reduction to disk chains was introduced by Xia~\cite{xia2013}.  We give a self-contained presentation with an independent formulation suited to progress accounting.  As recalled in \cref{sec:intro}, the circumdisks of the triangles that a segment \([u,v]\) crosses form such a chain (\cref{fig:delaunay}), and a path along their circles can be replaced by a path in the triangulation that is not longer.  We define disk chains without reference to a triangulation.  This reduction is the first step of \cref{tab:proofmap}.  Once it is in place, the triangulation plays no further role: \crefrange{sec:accounting}{sec:verification} work entirely with disk chains.

A disk chain is an abstraction of the circumdisks described above: it captures the geometric properties that the proof needs, without requiring a triangulation (\cref{rem:reduction}).  Informally, a disk chain advances along the \(x\)-axis from \(u\) to \(v\): consecutive circles cross once above and once below the axis, and the chord between the two crossings is the gate between the two disks (\cref{fig:chain}(a)).  The last condition below says that the two gates of a disk appear in order along both arcs of its circle; this ordering ensures that the arcs between consecutive gates, which form the edges of the ladder graph (\cref{def:ladder}), are well defined.  We formulate disk chains independently of a triangulation, which simplifies the conditions and makes the accounting identity possible; \cref{rem:xia} compares the two definitions in detail.

We begin with notation used throughout this section.  For a point \(z\in\R^2\), we denote its coordinates by \(x(z)\) and \(y(z)\).  We denote the Euclidean norm by \(\norm{\cdot}\) and the boundary circle of a closed disk~\(O\) by \(\partial O\).  When \(\partial O\) meets the \(x\)-axis in two points, we denote them by \(w_-(O)\) and \(w_+(O)\), where \(x(w_-(O))<x(w_+(O))\).  The \emph{upper arc} of \(O\) is the arc of \(\partial O\) from \(w_-(O)\) to \(w_+(O)\) that lies in the half-plane \(\{y\ge0\}\), and the \emph{lower arc} is the arc from \(w_-(O)\) to \(w_+(O)\) that lies in \(\{y\le0\}\).  We orient both arcs from \(w_-(O)\) to \(w_+(O)\); this gives a linear order on the points of each arc.

\Needspace{12\baselineskip}
\begin{definition}[Disk chains]\label{def:chain}
Let \(n\ge1\).  A \textbf{\emph{disk chain}} \((O_1,\dots,O_n;u,v)\) consists of closed disks \(O_1,\dots,O_n\), with centers \(o_i\) and radii \(r_i>0\), together with two points \(u,v\) on the \(x\)-axis satisfying \(u\in\partial O_1\), \(v\in\partial O_n\), and \(x(u)<x(v)\).  The following conditions must hold:
\begin{enumerate}[label=(C\arabic*),leftmargin=2.7em]
\item\label{c:gate} for every \(1\le i\le n-1\), the circles \(\partial O_i\) and \(\partial O_{i+1}\) meet in exactly two points \(a_i\) and \(b_i\), where \(y(a_i)>0>y(b_i)\); the segment \(G_i=[a_i,b_i]\) is the \(i\)-th \textbf{\emph{gate}}, and \(c_i\) denotes the point where the \(x\)-axis crosses \(G_i\);
\item\label{c:forward} \(x(o_1)<x(o_2)<\dots<x(o_n)\);
\item\label{c:order} \(x(u)<x(c_1)<\dots<x(c_{n-1})<x(v)\);
\item\label{c:rails} for every \(2\le i\le n-1\), on the upper arc of \(O_i\), the point \(a_{i-1}\) does not come after \(a_i\), and on the lower arc of \(O_i\), the point \(b_{i-1}\) does not come after \(b_i\).
\end{enumerate}
\end{definition}

Each \(c_i\) lies in the interior of both \(O_i\) and \(O_{i+1}\), so every circle \(\partial O_i\) meets the \(x\)-axis in two points.  By~\ref{c:order}, the point \(u\in\partial O_1\) on the \(x\)-axis lies to the left of \(c_1\), which is interior to \(O_1\); similarly, \(v\in\partial O_n\) lies to the right of \(c_{n-1}\), which is interior to \(O_n\).  (When \(n=1\) there are no gates, but \(x(u)<x(v)\) with \(u,v\in\partial O_1\) gives the same conclusion.)  Hence
\begin{equation}\label{eq:ends}
u=w_-(O_1)\qquad\text{and}\qquad v=w_+(O_n).
\end{equation}

A path from \(u\) to \(v\) along the circles of the chain can follow arcs between consecutive gates and cross gates.  The ladder graph records exactly these pieces and their lengths (\cref{fig:chain}(b)).

\begin{definition}[Ladder graph and chain inequality]\label{def:ladder}
Let \((O_1,\dots,O_n;u,v)\) be a disk chain, and set \(a_0=b_0=u\) and \(a_n=b_n=v\).  For \(1\le i\le n\), the \textbf{\emph{upper rail}} of \(O_i\) is the part of the upper arc of \(O_i\) from \(a_{i-1}\) to \(a_i\), and the \textbf{\emph{lower rail}} of \(O_i\) is the part of the lower arc from \(b_{i-1}\) to \(b_i\).  The \textbf{\emph{ladder graph}} of the chain has the vertices \(u\), \(v\), and \(a_i,b_i\) for \(1\le i\le n-1\).\footnote{Different labels among \(a_1,\dots,a_{n-1},b_1,\dots,b_{n-1}\) are different vertices even if they coincide as points; \(a_0\) and \(b_0\) both denote \(u\), and \(a_n\) and \(b_n\) both denote \(v\).}  Its edges are the \(2n\) rails, weighted by their arc lengths, and the \(n-1\) gates, weighted by their lengths \(\abs{G_i}=\norm{a_i-b_i}\).  We denote by \(P(u,v)\) the length of a shortest path from \(u\) to \(v\) in the ladder graph; it depends on the whole chain, not only on \(u\) and \(v\).  For \(U\in\R\), we say that the \textbf{\emph{chain inequality}} holds at \textbf{\emph{level}} \(U\) if \(P(u,v)\le U\norm{u-v}\) for every disk chain \((O_1,\dots,O_n;u,v)\).
\end{definition}

\begin{figure}[!htbp]
\centering
\begin{tikzpicture}[paperfig]
\path[use as bounding box] (-0.2,0.95) rectangle (15.6,5.6);
\draw[fdisk] (2.106,3.2) circle[radius=1.557];
\draw[fdisk] (3.67,2.994) circle[radius=1.56];
\draw[fdisk] (5.208,2.846) circle[radius=1.611];
\draw[faxis,-{Latex[length=1.8mm]}] (0.121,3.15) -- (7.297,3.15);
\draw[frail] (0.55,3.15) arc[start angle=-178.164,end angle=-307.806,radius=1.557];
\draw[frail,draw=figGreen] (0.55,3.15) arc[start angle=181.836,end angle=292.806,radius=1.557];
\draw[frail] (3.06,4.43) arc[start angle=113,end angle=57,radius=1.56];
\draw[frail,draw=figGreen] (2.71,1.765) arc[start angle=232,end angle=292,radius=1.56];
\draw[frail] (4.52,4.302) arc[start angle=115.299,end angle=10.881,radius=1.611];
\draw[frail,draw=figGreen] (4.254,1.548) arc[start angle=233.701,end angle=370.881,radius=1.611];
\draw[fgate] (3.06,4.43) -- (2.71,1.765);
\draw[fgate] (4.52,4.302) -- (4.254,1.548);
\node[fpt] at (0.55,3.15) {};
\node[fpt] at (6.79,3.15) {};
\node[fpt] at (3.06,4.43) {};
\node[fpt] at (2.71,1.765) {};
\node[fpt] at (4.52,4.302) {};
\node[fpt] at (4.254,1.548) {};
\node[fpto] at (2.892,3.15) {};
\node[fpto] at (4.409,3.15) {};
\node[anchor=south east] at (0.49,3.21) {$u$};
\node[anchor=west] at (6.87,3.37) {$v$};
\node[anchor=south east] at (3.02,4.51) {$a_1$};
\node[anchor=south west] at (4.56,4.382) {$a_2$};
\node[anchor=north east] at (2.65,1.715) {$b_1$};
\node[anchor=north west] at (4.314,1.498) {$b_2$};
\node[anchor=south west] at (2.952,3.19) {$c_1$};
\node[anchor=south east] at (4.349,3.19) {$c_2$};
\node[anchor=west,text=figOrangeText] at (2.877,2.431) {$G_1$};
\node[anchor=west,text=figOrangeText] at (4.48,2.65) {$G_2$};
\node[text=black] at (0.765,4.541) {$O_1$};
\node[text=black] at (3.67,4.804) {$O_2$};
\node[text=black] at (6.565,4.203) {$O_3$};
\draw[frail] (9.6,3.15) -- (11.35,4.4);
\draw[frail,draw=figGreen] (9.6,3.15) -- (11.35,1.9);
\draw[frail] (11.35,4.4) -- (13.1,4.4);
\draw[frail,draw=figGreen] (11.35,1.9) -- (13.1,1.9);
\draw[frail] (13.1,4.4) -- (14.85,3.15);
\draw[frail,draw=figGreen] (13.1,1.9) -- (14.85,3.15);
\draw[fgate] (11.35,4.4) -- (11.35,1.9);
\draw[fgate] (13.1,4.4) -- (13.1,1.9);
\node[fpt] at (9.6,3.15) {};
\node[fpt] at (11.35,4.4) {};
\node[fpt] at (11.35,1.9) {};
\node[fpt] at (13.1,4.4) {};
\node[fpt] at (13.1,1.9) {};
\node[fpt] at (14.85,3.15) {};
\node[anchor=east] at (9.5,3.15) {$u$};
\node[anchor=west] at (14.95,3.15) {$v$};
\node[anchor=south] at (11.35,4.48) {$a_1$};
\node[anchor=south] at (13.1,4.48) {$a_2$};
\node[anchor=north] at (11.35,1.82) {$b_1$};
\node[anchor=north] at (13.1,1.82) {$b_2$};
\node[anchor=west,text=figOrangeText] at (11.43,3.15) {$\abs{G_1}$};
\node[anchor=west,text=figOrangeText] at (13.18,3.15) {$\abs{G_2}$};
\node[anchor=west] at (-0.1,5.3) {(a)};
\node[anchor=west] at (9,5.3) {(b)};
\end{tikzpicture}
\caption{A disk chain determines a ladder graph.  (a)~The three disks \(O_1,O_2,O_3\) form a chain from \(u\) to \(v\).  The gates \(G_1\) and \(G_2\) cross \([u,v]\) at \(c_1\) and \(c_2\).  The upper rails are blue, and the lower rails are green.  (b)~The ladder graph has one edge for each rail and one edge for each gate.  Each rail is weighted by its arc length, and each gate is weighted by its Euclidean length.}
\label{fig:chain}
\end{figure}

A chain with one disk has no gate, so the path from \(u\) to \(v\) simply follows one of the two arcs of \(\partial O_1\).  We handle this base case directly.  This case also reveals why the level must satisfy \(U\ge\pi/2\): when the center of \(O_1\) lies on the \(x\)-axis, the shorter arc is a semicircle of length \(\frac\pi2\norm{u-v}\), so no smaller level can work.

\begin{lemma}[One disk]\label{lem:one-disk}
Let \(u\ne v\) be two points on a circle of radius \(r\), and let \(\eta\in(0,\pi]\) be the smaller central angle between them.  Then \(r\eta\le\frac\pi2\norm{u-v}\), with equality if and only if \(\eta=\pi\).  Consequently, every disk chain with \(n=1\) satisfies \(P(u,v)\le\frac\pi2\norm{u-v}\), with equality if and only if the center of \(O_1\) lies on the \(x\)-axis.  In particular, the chain inequality can hold only at levels \(U\ge\pi/2\), and at every such level it holds for all disk chains with one disk.
\end{lemma}

\begin{proof}
Since \(\norm{u-v}=2r\sin(\eta/2)\), we need to show that \(\sin t\ge 2t/\pi\) for \(t=\eta/2\in(0,\pi/2]\) (Jordan's inequality).  The sine function is strictly concave on \([0,\pi/2]\) and agrees with the linear function \(2t/\pi\) at \(t=0\) and \(t=\pi/2\); hence \(\sin t\ge2t/\pi\) on this interval, with equality only at the two endpoints.  For a chain with \(n=1\), the ladder graph consists of the two arcs of \(\partial O_1\) between \(u\) and \(v\), so \(P(u,v)\) is the length \(r_1\eta\) of the shorter arc.  Since \(u\) and \(v\) lie on the \(x\)-axis, they are the ends of a diameter if and only if the center of \(O_1\) lies on the \(x\)-axis; so equality holds exactly in this case.  For such a chain, \(P(u,v)=\frac\pi2\norm{u-v}>U\norm{u-v}\) for every \(U<\pi/2\), so the chain inequality fails at these levels; and for \(U\ge\pi/2\), every chain with one disk satisfies \(P(u,v)\le\frac\pi2\norm{u-v}\le U\norm{u-v}\).
\end{proof}

\begin{theorem}[Reduction to disk chains]\label{thm:reduction}
Let \(U\in\R\), and suppose that the chain inequality holds at level \(U\).  Then, for every finite set \(S\subset\R^2\), every Delaunay triangulation \(\cD\) of \(S\), and all \(u,v\in S\), we have \(d_{\cD}(u,v)\le U\norm{u-v}\).
\end{theorem}

We prove \cref{thm:reduction} in \cref{app:reduction}; the argument adapts Xia's reduction~\cite{xia2013} to our definition of disk chains.  Here we only sketch the idea.  After a rigid motion that puts \(u\) and \(v\) on the \(x\)-axis, the circumdisks of the triangles that \([u,v]\) crosses form a disk chain whose gates are edges of \(\cD\) (\cref{fig:delaunay}).  Each rail is an arc of a triangle's circumcircle, connecting two vertices of that triangle.  The triangulation edge between these two vertices is a chord of the circle, and therefore no longer than the arc.  So every rail can be replaced by a path in \(\cD\) that is not longer, as \cref{lem:rails-to-paths} shows, and \(d_{\cD}(u,v)\le P(u,v)\le U\norm{u-v}\).  The appendix also handles two degenerate cases: segments \([u,v]\) that pass through other points of \(S\), and cocircular configurations in which several triangles share a circumcircle.

\begin{remark}\label{rem:reduction}
The reduction in \cref{thm:reduction} need not be tight.  Two distinct sources of slack are present: the class of disk chains is strictly larger than the class of chains arising from Delaunay triangulations, and replacing circular arcs by chords (triangulation edges) can only shorten the path.  Both act in the same direction, so \(d_{\cD}(u,v)\le P(u,v)\) with potentially strict inequality: the framework can only overestimate the stretch factor.  This slack is intrinsic to the disk-chain model.  Whether the reduction is tight, that is, whether the smallest level at which the chain inequality holds equals the worst-case stretch factor of Delaunay triangulations, is the main open question left by this work; we discuss it in \cref{sec:conclusion}.
\end{remark}

\section{Progress Accounting}\label{sec:accounting}

In this section, we develop the progress accounting framework introduced in \cref{subsec:our-method}.  \Cref{sec:chains} reduced \cref{thm:main} to the chain inequality: \(P(u,v)\le U\norm{u-v}\) for every disk chain.  Checking it for a chain of \(n\) disks involves \(3n\) real parameters, where \(n\) is arbitrary.  We replace it by the inequalities of a feasible potential, which concern a single step and involve at most nine real parameters, whatever \(n\) is (\cref{thm:criterion}), and nothing is lost in this change (\cref{prop:exact}).  This is the second reduction of \cref{tab:proofmap}, from the whole chain to one step.  We want to check the chain inequality one disk at a time, and this requires two ingredients.

The first is to split both sides into one term per disk (a cost from the left side and a progress from the right, as in~\eqref{eq:intro-split}), so that each term depends only on information near its disk.  The right-hand side splits easily: the gates cross the segment \([u,v]\) one after another and cut it into \(n\) parts, one for each disk (\cref{fig:accounting}(a)); the length of the part that belongs to a disk is its progress.  The left-hand side is harder to split: \(P(u,v)\) is the length of a shortest path in the whole ladder graph, and where this path changes rails seems to depend on the whole chain.  To split \(P(u,v)\), we record at each gate the shortest distances from \(u\) to its two ends.  Shortest paths never go back, so these distances can be computed gate by gate, and the cost of a disk depends only on the disk, its two gates and the difference of the two distances at its left gate (\cref{subsec:recursion}).  \Cref{subsec:identity} collects this information in the gate state.

The second is to let the disks pay for each other.  Even after the splitting, we cannot require that the cost of each disk be at most \(U\) times its progress, because a single disk may cost more: when the first gate crosses \([u,v]\) very close to \(u\), as in the example of \cref{sec:intro}, the first disk makes almost no progress, while every path from \(u\) to \(v\) in the ladder graph still has to follow an arc from \(u\) to an end of this gate.  In \cref{subsec:potentials}, we introduce a potential that pays for the excess, as in~\eqref{eq:intro-step}.  Taking the potential to be a function of the gate state makes each requirement a condition on one gate state and one next disk.

\subsection{Splitting \texorpdfstring{\(P(u,v)\)}{P(u,v)} Disk by Disk}\label{subsec:recursion}

This subsection carries out the splitting of \(P(u,v)\), the first ingredient, and arrives at the accounting identity~\eqref{eq:accounting}, which writes \(P(u,v)-U\norm{u-v}\) as a sum of one term per disk.  If the parts are only required to add up to \(P(u,v)\), any splitting will do.  What we need is a splitting in which the part of each disk can be computed from this disk, its two gates and a small amount of information carried over from the left, by the same formula for all chains.

When a shortest path reaches a gate, it may arrive at the upper end or at the lower end.  So we record two numbers: the shortest distances from \(u\) to the upper and to the lower end of the gate.  We measure them in the part of the ladder graph to the left of the gate (\cref{def:prefix}), so that the distances at a gate can be computed from those at the previous gate (\cref{lem:update}); this does not change them, since a detour beyond the gate, from one of its ends to the other, is not shorter than the gate itself.  As in \cref{sec:intro}, we take the part of a disk to be the increase of the mean of these two distances (the mean, for symmetry between the two rails) from its left gate to its right gate, and we call it the \emph{cost} of the disk.

Fix a disk chain \((O_1,\dots,O_n;u,v)\) with \(n\ge2\); chains with one disk are covered by \cref{lem:one-disk}.  To cut its ladder graph at a gate (\cref{def:prefix}), we arrange the vertices in \(n+1\) columns: column \(0\) is \(\{u\}\), column \(i\) is \(\{a_i,b_i\}\) for \(1\le i\le n-1\), and column \(n\) is \(\{v\}\).  Each gate joins the two vertices of one column, and the two rails of \(O_i\) join column \(i-1\) to column \(i\) (\cref{fig:chain}(b)).  To treat the first and the last disk like the others, we regard \(u\) and \(v\) as gates of length \(0\): recall from \cref{def:ladder} that \(a_0=b_0=u\) and \(a_n=b_n=v\), and set \(G_0=[a_0,b_0]=\{u\}\), \(G_n=[a_n,b_n]=\{v\}\), \(\abs{G_0}=\abs{G_n}=0\), \(c_0=u\) and \(c_n=v\).

\begin{definition}[Prefix graphs, costs and progresses]\label{def:prefix}
For \(0\le i\le n\), the \(i\)-th \textbf{\emph{prefix graph}} is the subgraph of the ladder graph induced by the vertices of columns \(0,\dots,i\); so the \(0\)-th prefix graph is the single vertex \(u\), and the \(n\)-th prefix graph is the ladder graph.  Let \(d_i^+\) and \(d_i^-\) be the lengths of shortest paths from \(u\) to \(a_i\) and to \(b_i\) in the \(i\)-th prefix graph, and let
\begin{equation}\label{eq:mean-diff}
\mu_i=\tfrac12\bigl(d_i^++d_i^-\bigr),\qquad \delta_i=d_i^+-d_i^- .
\end{equation}
For \(1\le i\le n\), the \textbf{\emph{cost}} of the disk \(O_i\) is \(\mu_i-\mu_{i-1}\), and its \textbf{\emph{progress}} is \(x(c_i)-x(c_{i-1})\).
\end{definition}

Since column \(0\) is \(\{u\}\) and \(a_n=b_n=v\), we have \(\mu_0=\delta_0=0\), \(\mu_n=P(u,v)\) and \(\delta_n=0\).  So the costs of the \(n\) disks add up to \(P(u,v)\) and their progresses to \(x(v)-x(u)=\norm{u-v}\); both sums telescope by definition.

The next two lemmas show that the cost has exactly the locality we need.  \Cref{lem:forward} says that shortest paths in the ladder graph never go back; so the shortest distances to the two ends of a gate can be computed from those to the two ends of the previous gate, each new distance being the smaller of two candidates of the form ``an old distance plus a fixed length'' (\cref{fig:accounting}(b)).  Consequently, adding the same number to both old distances adds this number to both new distances and leaves the increase of the mean unchanged; so the cost depends on the two old distances only through their difference, not through their mean.

\begin{lemma}[Shortest paths go forward]\label{lem:forward}
Let \(1\le i\le n\), and consider the \(i\)-th prefix graph with target \(a_i\) or \(b_i\); for \(i=n\), this is the ladder graph with target \(v\).  Then some shortest path from \(u\) to the target is simple, and every simple path from \(u\) to the target traverses exactly one rail of each disk \(O_j\) with \(1\le j\le i\), from its end in column \(j-1\) to its end in column \(j\).
\end{lemma}

\begin{proof}
The graph is finite and all weights are nonnegative, so a shortest path from \(u\) to the target exists and, after the removal of its closed parts, is simple.

For the second claim, fix a disk \(O_j\) with \(1\le j\le i\) and remove its two rails.  Every remaining edge is either a gate, which joins two vertices of the same column, or a rail of some \(O_{j'}\) with \(j'\ne j\), which joins columns \(j'-1\) and \(j'\).  So no remaining edge joins a column before \(j\) to a column from \(j\) on.  A simple path from \(u\) (in column \(0\)) to the target (in column \(i\ge j\)) therefore passes between the two groups of columns an odd number of times, each time along a rail of \(O_j\).  Since \(O_j\) has only two rails and a simple path uses every edge at most once, this odd number is \(1\).  So the path traverses exactly one rail of \(O_j\), exactly once, and it goes from the group of columns that contains \(u\) to the group that contains the target, that is, from column \(j-1\) to column \(j\).
\end{proof}

The update formulas use two notations.  For \(h\ge0\) and \(t\in\R\), let \(\clip_h(t)=\max(-h,\min(h,t))\) be the clipping of \(t\) to \([-h,h]\), and let \(t_+=\max(t,0)\).

\begin{lemma}[Exact update]\label{lem:update}
For \(1\le i\le n\), let \(\alpha_i\) and \(\beta_i\) be the lengths of the upper and lower rails of \(O_i\).  Then, for \(0\le i\le n-1\),
\begin{equation}\label{eq:update}
\delta_{i+1}=\clip_{\abs{G_{i+1}}}\bigl(\delta_i+\alpha_{i+1}-\beta_{i+1}\bigr),\qquad
\mu_{i+1}=\mu_i+\tfrac12\bigl(\alpha_{i+1}+\beta_{i+1}-\omega_{i+1}\bigr),
\end{equation}
where \(\omega_{i+1}=\bigl(\abs{\delta_i+\alpha_{i+1}-\beta_{i+1}}-\abs{G_{i+1}}\bigr)_+\).  In particular, for \(i=n-1\), since \(\abs{G_n}=0\),
\begin{equation}\label{eq:Puv-last}
P(u,v)=\mu_n=\mu_{n-1}+\min\bigl(\tfrac12\delta_{n-1}+\alpha_n,\,-\tfrac12\delta_{n-1}+\beta_n\bigr).
\end{equation}
\end{lemma}

\begin{proof}
We use two identities, valid for all \(A,B\in\R\) and \(h\ge0\):
\begin{equation}\label{eq:minclip}
\begin{gathered}
\min(A,B+h)-\min(B,A+h)=\clip_h(A-B),\\
\min(A,B+h)+\min(B,A+h)=A+B-\bigl(\abs{A-B}-h\bigr)_+ .
\end{gathered}
\end{equation}
If \(\abs{A-B}\le h\), the two minima are \(A\) and \(B\).  If \(A-B>h\), they are \(B+h\) and \(B\), with difference \(h\) and sum \(A+B-(A-B-h)\).  The case \(A-B<-h\) is symmetric.

Let \(0\le i\le n-1\).  The \((i+1)\)-th prefix graph consists of the \(i\)-th prefix graph, the upper rail of \(O_{i+1}\) from \(a_i\) to \(a_{i+1}\), the lower rail from \(b_i\) to \(b_{i+1}\), and, if \(i+1\le n-1\), the gate \(G_{i+1}\); for \(i=0\), both rails start at \(a_0=b_0=u\), and for \(i+1=n\), both rails end at \(a_n=b_n=v\), and the gate \(G_n\) is the single point \(v\), of length \(0\).  By \cref{lem:forward}, some shortest path from \(u\) to \(a_{i+1}\) is simple and traverses exactly one rail of \(O_{i+1}\).  Before this rail, the path stays in the \(i\)-th prefix graph; after it, the path can only use the gate \(G_{i+1}\).  If the rail is the upper one, the path has length at least \(d_i^++\alpha_{i+1}\); if it is the lower one, the path reaches \(b_{i+1}\) and then crosses \(G_{i+1}\), so it has length at least \(d_i^-+\beta_{i+1}+\abs{G_{i+1}}\).  Both values are attained by appending one or two edges to a shortest path in the \(i\)-th prefix graph.  Hence
\begin{equation}\label{eq:dist-update}
d_{i+1}^+=\min\bigl(d_i^++\alpha_{i+1},\,d_i^-+\beta_{i+1}+\abs{G_{i+1}}\bigr),
\qquad
 d_{i+1}^-=\min\bigl(d_i^-+\beta_{i+1},\,d_i^++\alpha_{i+1}+\abs{G_{i+1}}\bigr),
\end{equation}
where the second formula follows in the same way.  We apply~\eqref{eq:minclip} with \(A=d_i^++\alpha_{i+1}\), \(B=d_i^-+\beta_{i+1}\) and \(h=\abs{G_{i+1}}\).  Since \(A-B=\delta_i+\alpha_{i+1}-\beta_{i+1}\) and \(A+B=2\mu_i+\alpha_{i+1}+\beta_{i+1}\), the two identities give~\eqref{eq:update}.  For \(i=n-1\), we have \(\abs{G_n}=0\) and \(\omega_n=\abs{\delta_{n-1}+\alpha_n-\beta_n}\), so \(\mu_n-\mu_{n-1}=\frac12(A'+B'-\abs{A'-B'})=\min(A',B')\) with \(A'=\frac12\delta_{n-1}+\alpha_n\) and \(B'=-\frac12\delta_{n-1}+\beta_n\); and \(\mu_n=d_n^\pm=P(u,v)\) because \(a_n=b_n=v\).
\end{proof}

In words, the rails of \(O_{i+1}\) extend the shortest paths to \(a_i\) and to \(b_i\) by \(\alpha_{i+1}\) and \(\beta_{i+1}\).  If one extended path exceeds the other by more than \(\abs{G_{i+1}}\), crossing the gate is shorter: the longer path is replaced by the shorter one followed by \(G_{i+1}\), saving the length \(\omega_{i+1}\) and clipping the difference to \([-\abs{G_{i+1}},\abs{G_{i+1}}]\) (\cref{fig:accounting}(b)).  \Cref{lem:update} thus shows that the cost \(\mu_{i+1}-\mu_i\) of the disk \(O_{i+1}\) is computed from three things: the lengths of its two rails, the length of its right gate \(G_{i+1}\), and the difference \(\delta_i\) at its left gate; and the new difference \(\delta_{i+1}\) is computed from the same three.  Along the chain from left to right, a single number \(\delta\) is carried from gate to gate, and the cost of each disk is computed on the spot.  The telescoping sums noted after \cref{def:prefix} give the \emph{accounting identity} that drives progress accounting.  By \cref{lem:update}, each term is local: for every \(U\in\R\),
\begin{equation}\label{eq:accounting}
P(u,v)-U\norm{u-v}=\sum_{i=1}^{n}\Bigl[(\mu_i-\mu_{i-1})-U\bigl(x(c_i)-x(c_{i-1})\bigr)\Bigr].
\end{equation}
This is~\eqref{eq:intro-split}, with the cost and the progress of \cref{def:prefix}.  So it remains to show that, for every disk chain with \(n\ge2\), the right-hand side of~\eqref{eq:accounting} is at most \(0\).

\begin{figure}[!htbp]
\centering
\begin{tikzpicture}[paperfig]
\path[use as bounding box] (-0.2,-0.6) rectangle (15.6,5.0);
\coordinate (u) at (0.8,2.5);
\coordinate (a1) at (2.9,3.7); \coordinate (b1) at (2.9,1.3);
\coordinate (a2) at (4.7,3.7); \coordinate (b2) at (4.7,1.3);
\coordinate (v) at (7.1,2.5);
\coordinate (c1) at (2.9,2.5); \coordinate (c2) at (4.7,2.5);
\draw[fguide] (u) -- (0.8,0.6);
\draw[fguide] (c1) -- (2.9,0.6);
\draw[fguide] (c2) -- (4.7,0.6);
\draw[fguide] (v) -- (7.1,0.6);
\draw[faxis,-{Latex[length=1.8mm]}] (0.3,2.5) -- (7.75,2.5);
\draw[frail] (u) -- (a1); \draw[frail,draw=figGreen] (u) -- (b1);
\draw[frail] (a1) -- (a2); \draw[frail,draw=figGreen] (b1) -- (b2);
\draw[frail] (a2) -- (v); \draw[frail,draw=figGreen] (b2) -- (v);
\draw[fgate] (a1) -- (b1); \draw[fgate] (a2) -- (b2);
\node[fpt] at (u) {}; \node[fpt] at (v) {};
\node[fpt] at (a1) {}; \node[fpt] at (b1) {}; \node[fpt] at (a2) {}; \node[fpt] at (b2) {};
\node[fpto] at (c1) {}; \node[fpto] at (c2) {};
\node[anchor=north east] at (0.78,2.42) {$u=c_0$};
\node[anchor=north west] at (7.12,2.42) {$v=c_3$};
\node[anchor=south] at (2.9,3.78) {$a_1$}; \node[anchor=south] at (4.7,3.78) {$a_2$};
\node[anchor=north east] at (2.86,1.26) {$b_1$}; \node[anchor=north west] at (4.74,1.26) {$b_2$};
\node[anchor=south west] at (3.0,2.56) {$c_1$}; \node[anchor=south west] at (4.8,2.56) {$c_2$};
\node[anchor=east,text=figOrangeText] at (2.8,3.15) {$G_1$};
\node[anchor=west,text=figOrangeText] at (4.8,3.15) {$G_2$};
\node[text=figBlueText,anchor=south] at (1.85,4.0) {$\mu_1-\mu_0$};
\node[text=figBlueText,anchor=south] at (3.8,4.0) {$\mu_2-\mu_1$};
\node[text=figBlueText,anchor=south] at (5.9,4.0) {$\mu_3-\mu_2$};
\draw[fbrace] (2.9,0.6) -- node[below=5pt] {\shortstack{$x(c_1)$\\$-x(c_0)$}} (0.8,0.6);
\draw[fbrace] (4.7,0.6) -- node[below=5pt] {\shortstack{$x(c_2)$\\$-x(c_1)$}} (2.9,0.6);
\draw[fbrace] (7.1,0.6) -- node[below=5pt] {\shortstack{$x(c_3)$\\$-x(c_2)$}} (4.7,0.6);
\coordinate (ai) at (10.0,3.7); \coordinate (bi) at (10.0,1.3);
\coordinate (aj) at (13.0,3.7); \coordinate (bj) at (13.0,1.3);
\fill[figFill] (8.6,0.95) rectangle (10.0,4.05);
\node[text=black,rotate=90] at (9.15,2.5) {prefix graph};
\draw[fthick,-{Latex[length=1.4mm]}] (8.75,3.7) -- (9.88,3.7);
\draw[fthick,-{Latex[length=1.4mm]}] (8.75,1.3) -- (9.88,1.3);
\node[text=black,anchor=north] at (9.3,1.26) {$d_i^-$};
\draw[frail] (ai) -- (aj); \draw[frail,draw=figGreen] (bi) -- (bj);
\draw[fgate] (ai) -- (bi); \draw[fgate] (aj) -- (bj);
\draw[fold,-{Latex[length=1.6mm]}] (10.15,3.52) -- (12.85,3.52);
\draw[fnew,dash pattern=on 3pt off 2pt,-{Latex[length=1.6mm]}] (10.15,1.48) -- (12.82,1.48) -- (12.82,3.5);
\node[fpt] at (ai) {}; \node[fpt] at (bi) {}; \node[fpt] at (aj) {}; \node[fpt] at (bj) {};
\node[anchor=south] at (10.0,3.78) {$a_i$}; \node[anchor=south] at (13.0,3.78) {$a_{i+1}$};
\node[anchor=north] at (10.0,1.22) {$b_i$}; \node[anchor=north] at (13.0,1.22) {$b_{i+1}$};
\node[anchor=south,text=figBlueText] at (11.5,3.78) {$\alpha_{i+1}$};
\node[anchor=north,text=figGreenText] at (11.5,1.22) {$\beta_{i+1}$};
\node[anchor=west,text=figOrangeText] at (10.07,2.5) {$G_i$};
\node[anchor=west,text=figOrangeText] at (13.06,2.5) {$\abs{G_{i+1}}$};
\node[anchor=north] at (11.5,0.6) {Two routes to $a_{i+1}$};
\node[anchor=west] at (-0.1,4.85) {(a)};
\node[anchor=west] at (8.4,4.85) {(b)};
\end{tikzpicture}
\caption{The accounting identity compares the cost and progress of each disk.  (a)~The gates divide \([u,v]\) into the progresses of three disks, shown by the braces.  Each label above the rails gives the cost of that disk as an increase of \(\mu\).  The accounting identity compares each cost with \(U\) times the corresponding progress.  (b)~A shortest path to \(a_{i+1}\) either follows the upper rail from \(a_i\), shown in blue, or follows the lower rail from \(b_i\) and then crosses \(G_{i+1}\), shown by the dashed green route.  The distance to \(b_{i+1}\) is computed in the same way.  Crossing the gate when it shortens the longer candidate path clips the new difference to \([-\abs{G_{i+1}},\abs{G_{i+1}}]\).}
\label{fig:accounting}
\end{figure}

\subsection{Gate States and Moves}\label{subsec:identity}

By \cref{lem:update}, once a chain has passed its \(i\)-th gate, the remaining terms of~\eqref{eq:accounting} are determined by the disk \(O_{i+1}\), the gate \(G_i\), the number \(\delta_i\) and the disks \(O_{i+2},\dots,O_n\); how the chain arrived there plays no further role.  We call the triple \((O_{i+1},G_i,\delta_i)\) the \emph{gate state} of the chain at its \(i\)-th gate.  Note the shift of indices: the disk of a gate state is the disk to the right of the gate, whose cost is still to be paid.  The potential of \cref{subsec:potentials} will be defined on gate states, and its inequalities will be required at every gate state and for every possible next disk, not only along one chain.  So we define gate states, and the moves from one gate state to the next, without reference to a chain, and we check afterwards that every disk chain is a sequence of moves (\cref{lem:moves}).

\begin{definition}[Gate states]\label{def:gate-state}
A \textbf{\emph{gate state}} is a triple \(s=(O,G,\delta)\), where \(O\) is a closed disk with center \(o\) and radius \(r>0\), \(G=[a_G,b_G]\) is a chord of \(\partial O\) with \(y(a_G)>0>y(b_G)\), and \(\delta\in[-\abs{G},\abs{G}]\).  We denote by \(c(s)\) the point where the \(x\)-axis crosses \(G\).
\end{definition}

A move accounts for one disk of a chain: it goes from the left gate of the disk to its right gate, and its cost and its progress are given by the formulas of \cref{lem:update}, with the ends \(u\) and \(v\) of a chain regarded as gates of length \(0\).  The moves of a chain are of three kinds, depending on the position in the chain: an initial pair of disks for the first disk, a step from a gate state to a next disk for each intermediate disk, and a final step at a gate state for the last disk.

\begin{definition}[Moves]\label{def:moves}
Let \(O\) be a closed disk whose circle meets the \(x\)-axis in two points, and let \(G=[a_G,b_G]\) and \(G'=[a',b']\) each be either a chord of \(\partial O\) with one end above and the other end below the \(x\)-axis, the upper end being listed first, or a point of \(\partial O\) on the \(x\)-axis, regarded as a chord of length \(0\) whose two ends coincide.  Suppose that, on the upper arc of \(O\), \(a_G\) does not come after \(a'\), and on the lower arc, \(b_G\) does not come after \(b'\); let \(\alpha\) and \(\beta\) be the lengths of the parts of the upper arc from \(a_G\) to \(a'\) and of the lower arc from \(b_G\) to \(b'\), and let \(c\) and \(c'\) be the points where the \(x\)-axis crosses \(G\) and \(G'\).  For \(\delta\in[-\abs{G},\abs{G}]\), the \emph{move on \(O\) from \(G\) to \(G'\) with difference \(\delta\)} has the \emph{new difference} \(\delta'\), the \emph{cost} \(\Delta\mu\) and the \emph{progress} \(\Delta x\) given by
\begin{equation}\label{eq:move-formulas}
\delta'=\clip_{\abs{G'}}(\delta+\alpha-\beta),\qquad
\Delta\mu=\tfrac12(\alpha+\beta-\omega),\qquad
\Delta x=x(c')-x(c),
\end{equation}
where \(\omega=\bigl(\abs{\delta+\alpha-\beta}-\abs{G'}\bigr)_+\).
\begin{enumerate}[label=(\arabic*)]
\item\label{mv:initial} An \emph{initial pair} is a pair \((O_1,O_2)\) of closed disks with centers \(o_1,o_2\) such that \(\partial O_1\) and \(\partial O_2\) meet in exactly two points \(a_1,b_1\) with \(y(a_1)>0>y(b_1)\), and \(x(o_2)>x(o_1)\).  Let \(G_1=[a_1,b_1]\), let \(c_1\) be the point where the \(x\)-axis crosses \(G_1\), and let \(u=w_-(O_1)\).  The initial pair accounts for \(O_1\) by the move on \(O_1\) from \(u\) to \(G_1\) with difference \(0\); we denote the new difference and the cost of this move by \(\delta_1\) and \(\mu_1\), and we say that the initial pair \emph{leads to} the gate state \(s_1=(O_2,G_1,\delta_1)\).  Its progress is \(x(c_1)-x(u)\).
\item\label{mv:step} A \emph{next disk} for a gate state \(s=(O,G,\delta)\) is a closed disk \(O'\) with center \(o'\) such that the following hold:
\begin{enumerate}[label=(N\arabic*)]
\item\label{n:gate} \(\partial O\) and \(\partial O'\) meet in exactly two points \(a',b'\) with \(y(a')>0>y(b')\);
\item\label{n:forward} \(x(o')>x(o)\);
\item\label{n:order} the point \(c'\) where the \(x\)-axis crosses \(G'=[a',b']\) satisfies \(x(c')>x(c(s))\);
\item\label{n:rails} on the upper arc of \(O\), \(a_G\) does not come after \(a'\), and on the lower arc of \(O\), \(b_G\) does not come after \(b'\).
\end{enumerate}
The \emph{step} from \(s\) to \(O'\) accounts for \(O\) by the move on \(O\) from \(G\) to \(G'\) with difference \(\delta\); it \emph{leads to} the gate state \(s'=(O',G',\delta')\), where \(\delta'\) is the new difference of this move, and its cost and progress are \(\Delta\mu\) and \(\Delta x=x(c')-x(c(s))\).
\item\label{mv:final} The \emph{final step} at a gate state \(s=(O,G,\delta)\) accounts for \(O\) by the move on \(O\) from \(G\) to \(w_+(O)\) with difference \(\delta\); it \emph{ends} at \(w_+(O)\).  Here we denote by \(\alpha_{\rm end}\) and \(\beta_{\rm end}\) the lengths \(\alpha\) and \(\beta\), that is, the lengths of the parts of the upper arc of \(O\) from \(a_G\) to \(w_+(O)\) and of the lower arc from \(b_G\) to \(w_+(O)\).  Since the gate \(w_+(O)\) has length \(0\), the cost of the final step is
\begin{equation}\label{eq:final-cost}
\kappa=\min\bigl(\tfrac12\delta+\alpha_{\rm end},\,-\tfrac12\delta+\beta_{\rm end}\bigr),
\end{equation}
and its progress is \(x(w_+(O))-x(c(s))\).
\end{enumerate}
\end{definition}

Conditions~\ref{n:gate}--\ref{n:rails} mirror~\ref{c:gate}--\ref{c:rails} for two consecutive disks.

A step is determined by nine real parameters: the center and the radius of \(O\), the positions of \(a_G\) and \(b_G\) on \(\partial O\), the difference \(\delta\), and the center and the radius of the next disk \(O'\).  The new gate \(G'\), the new difference \(\delta'\), the cost and the progress of the step are computed from them by \cref{def:moves}.  An initial pair is determined by the six parameters of its two disks, and a final step by the six parameters of its gate state.

\begin{lemma}[Chains are sequences of moves]\label{lem:moves}
Let \((O_1,\dots,O_n;u,v)\) be a disk chain with \(n\ge2\), let \(\mu_i\) and \(\delta_i\) be as in \cref{def:prefix}, and let \(s_i=(O_{i+1},G_i,\delta_i)\) for \(1\le i\le n-1\).  Then:
\begin{enumerate}[label=(\arabic*)]
\item\label{acc:moves} each \(s_i\) is a gate state; \((O_1,O_2)\) is an initial pair leading to \(s_1\); for \(1\le i\le n-2\), \(O_{i+2}\) is a next disk for \(s_i\), and the step from \(s_i\) to \(O_{i+2}\) leads to \(s_{i+1}\); and the final step at \(s_{n-1}\) ends at \(v\);
\item\label{acc:sums} for \(1\le i\le n\), the move that accounts for \(O_i\) has cost \(\mu_i-\mu_{i-1}\) and progress \(x(c_i)-x(c_{i-1})\), the cost and the progress of \(O_i\) in \cref{def:prefix}.
\end{enumerate}
Consequently, for every \(U\in\R\), \(P(u,v)-U\norm{u-v}\) is the sum of \(\text{cost}-U\cdot\text{progress}\) over these \(n\) moves.
\end{lemma}

\begin{proof}
Each \(s_i\) is a gate state: \(G_i\) is a chord of \(\partial O_{i+1}\), its ends satisfy \(y(a_i)>0>y(b_i)\) by~\ref{c:gate}, and \(\abs{\delta_i}\le\abs{G_i}\) by~\eqref{eq:update}.

\ref{acc:moves} The pair \((O_1,O_2)\) satisfies the conditions of an initial pair by~\ref{c:gate} and~\ref{c:forward}, and \(u=w_-(O_1)\) by~\eqref{eq:ends}.  For \(1\le i\le n-2\), the disk \(O_{i+2}\) is a next disk for \(s_i\): conditions~\ref{n:gate} and~\ref{n:forward} hold by~\ref{c:gate} and~\ref{c:forward}, condition~\ref{n:order} by~\ref{c:order}, since \(c(s_i)=c_i\), and condition~\ref{n:rails} by~\ref{c:rails} for the disk \(O_{i+1}\).  The final step at \(s_{n-1}\) ends at \(w_+(O_n)=v\) by~\eqref{eq:ends}.  It remains to see that each move leads to the next gate state, which follows from~\ref{acc:sums}: the new difference of the move that accounts for \(O_i\) is \(\delta_i\).

\ref{acc:sums} For \(1\le i\le n\), the move that accounts for \(O_i\) goes on \(O_i\) from \(G_{i-1}\) to \(G_i\) with difference \(\delta_{i-1}\), where \(G_0=\{u\}\) and \(G_n=\{v\}\); here we use \(\delta_0=0\) for the initial pair.  The lengths \(\alpha,\beta\) of this move are the lengths \(\alpha_i,\beta_i\) of the rails of \(O_i\), by \cref{def:ladder} and~\eqref{eq:ends}, and \(\abs{G'}=\abs{G_i}\).  So the formulas of \cref{def:moves} for the new difference and the cost are the formulas~\eqref{eq:update} of \cref{lem:update}, with \(i-1\) in place of \(i\), and give \(\delta_i\) and \(\mu_i-\mu_{i-1}\).  The progress is \(x(c_i)-x(c_{i-1})\), since \(c(s_{i-1})=c_{i-1}\), \(c_0=u\) and \(c_n=v\).  The last claim follows from~\eqref{eq:accounting}.
\end{proof}

\subsection{Feasible Potentials}\label{subsec:potentials}

By the accounting identity~\eqref{eq:accounting}, the chain inequality for a chain with \(n\ge2\) disks says that the \(n\) terms of~\eqref{eq:accounting} add up to at most \(0\), and this is what we now have to show for every disk chain.  This subsection reduces it to the existence of a feasible potential (\cref{thm:criterion}), the second ingredient; once such a potential is in hand, the chain inequality follows.  We cannot do this term by term: a single term can be positive, as in the example of \cref{sec:intro}, where the first gate is very close to \(u\).  So, as in~\eqref{eq:intro-step}, we let a potential pay for the excess.

To this end, we attach a number to every gate of the chain, \(0\) to the gates \(u\) and \(v\), and we require that each term of~\eqref{eq:accounting} be at most the number at the left gate of its disk minus the number at its right gate; then the \(n\) requirements telescope, and the sum is at most \(0\).

We take the number at the \(i\)-th gate to be \(W(s_i)\), where \(s_i\) is the gate state of the chain at this gate and \(W\) is a fixed function on gate states, the same for all chains.  This choice is what makes each requirement a condition on one gate state and one next disk, which can be checked without reference to any chain.  The three kinds of moves give three inequalities; the first and the last involve the number \(0\) at an end of the chain.

\begin{definition}[Feasible potential]\label{def:feasible}
Let \(U\in\R\).  A \textbf{\emph{feasible potential}} of level \(U\) is a function \(W\) from gate states to \(\R\) that satisfies the following three inequalities, in the notation of \cref{def:moves}:
\begin{itemize}
\item[(I)] \(\mu_1-U\bigl(x(c_1)-x(u)\bigr)\le-W(s_1)\) for every initial pair, where \(s_1\) is the gate state it leads to;
\item[(S)] \(\Delta\mu-U\,\Delta x\le W(s)-W(s')\) for every gate state \(s\) and every next disk for \(s\), where \(s'\) is the gate state the step leads to;
\item[(F)] \(\kappa-U\bigl(x(w_+(O))-x(c(s))\bigr)\le W(s)\) for every gate state \(s=(O,G,\delta)\), where \(\kappa\) is the cost of the final step at \(s\).
\end{itemize}
\end{definition}

In words, for every move, \emph{the cost minus \(U\) times the progress is at most the decrease of \(W\)}, where \(W\) is taken to be \(0\) before an initial pair and after a final step.  Along a chain, (I), (S) and (F) are the inequalities~\eqref{eq:intro-step} for \(i=1\), for \(2\le i\le n-1\) and for \(i=n\).

This restriction to functions on gate states loses no generality.  The largest possible sum of the remaining terms of~\eqref{eq:accounting}, over all ways of continuing a chain from a gate state, depends on the gate state alone.  If the chain inequality holds, this largest sum satisfies (I), (S) and (F), and it is the smallest function that does (\cref{prop:exact} below).

\begin{proposition}[The criterion is exact]\label{prop:exact}
Let \(U\ge\pi/2\).  For a gate state \(s\), let \(W^*_U(s)\in(-\infty,+\infty]\) be the supremum of the sums of \(\text{cost}-U\cdot\text{progress}\) over all finite, possibly empty, sequences of steps from \(s\) followed by the final step, the \emph{continuations} from \(s\), which \cref{def:continuation} defines formally.  We call \(W^*_U\) the \emph{value function} of level \(U\).  If the chain inequality holds at level \(U\), then \(W^*_U\) is finite at every gate state and is a feasible potential of level \(U\), and \(W^*_U\le W\) for every feasible potential \(W\) of level \(U\).
\end{proposition}

\Cref{app:exactness} proves \cref{prop:exact} as part of \cref{thm:exactness}.  The proposition is not used in the proof of \cref{thm:main}.

\begin{theorem}[Potential criterion]\label{thm:criterion}
Let \(U\ge\pi/2\).  If there exists a feasible potential of level \(U\), then the chain inequality holds at level \(U\).
\end{theorem}

\begin{proof}
Chains with one disk satisfy the chain inequality by \cref{lem:one-disk}, since \(U\ge\pi/2\).  Consider a chain with \(n\ge2\) disks, and let \(s_1,\dots,s_{n-1}\) be as in \cref{lem:moves}.  We add (I) for its initial pair, (S) for the steps from \(s_1,\dots,s_{n-2}\), and (F) for the final step at \(s_{n-1}\).  The right-hand sides add up to
\[
-W(s_1)+\sum_{i=1}^{n-2}\bigl(W(s_i)-W(s_{i+1})\bigr)+W(s_{n-1})=0,
\]
and, by \cref{lem:moves}, the left-hand sides add up to \(P(u,v)-U\norm{u-v}\).
\end{proof}

\begin{remark}[Comparison with Xia's proof]\label{rem:xia}
Our disk chains correspond to Xia's chains: condition~\ref{c:rails} of \cref{def:chain} plays the role of his requirement that, on each disk, the two arcs of the circle lying in the previous and in the next disk do not overlap~\cite[Definition~1]{xia2013}, the ladder graph is the graph that he associates with a chain, and \(P(u,v)\) is the length of the path that he bounds.  Xia's potential accumulates along the chain and is estimated at the end; ours looks ahead and vanishes at both ends.  Precisely, Xia's inequality~\cite[Lemma~2]{xia2013} reads \(P(u,v)<1.8\,L(u,v)-\Phi\) for every chain and all end points \(u,v\) allowed in~\cite[Definitions~1 and~2]{xia2013}, where \(L(u,v)\) is the length of the shortest polyline from \(u\) to \(v\) that meets the gates in order, and \(\Phi\) is a potential, a sum of terms for pairs of consecutive disks, which vanishes for a single disk and does not increase when a disk is added~\cite[Lemma~1]{xia2013}; so \(\Phi\le0\).  For the chain of a segment \([u,v]\), the polyline is the segment itself; but the induction also applies the inequality to sub-chains whose end points are ends of gates, and for them \(L(u,v)\) depends on all their disks at once, so the inequality cannot be split into one term per disk, as~\eqref{eq:accounting} splits \(P(u,v)-U\norm{u-v}\).  Since \(\Phi\le0\), the inequality does not bound \(P(u,v)\) by \(1.8\norm{u-v}\), and Xia bounds \(\Phi\) from below in terms of \(P(u,v)\) by a separate estimate for a chain with certain extremal properties~\cite[Lemma~3]{xia2013}; this estimate turns \(1.8\) into \(1.998\).  Our accounting identity~\eqref{eq:accounting} decomposes \(P(u,v)-U\norm{u-v}\) into exactly one term per disk, and our potential absorbs savings at each gate locally.  Xia's inequality, by contrast, compares \(P(u,v)\) with \(L(u,v)\), which depends on the whole chain, and it needs a separate global estimate of \(\Phi\) to recover the stretch bound.  This local structure is what makes computer-assisted verification in a bounded number of parameters feasible, and it accounts for the improvement from \(1.998\) to \(1.5949\).
\end{remark}

\Cref{thm:criterion,prop:exact} take the proof from the whole chain to one step.  The chain inequality holds if and only if a feasible potential exists, and the inequalities of a feasible potential involve at most nine real parameters, whatever the number of disks.  It remains to construct a feasible potential of level \(U_0\); \crefrange{sec:pencil}{sec:verification} do so with the help of a computer, and \cref{sec:pencil} first brings the nine parameters down to four.

\section[Reduction of the Verification Space]{\texorpdfstring{\raggedright Reduction of the Verification Space}{Reduction of the Verification Space}}\label{sec:pencil}

By \cref{thm:criterion}, it suffices to construct a feasible potential.  As counted in \cref{subsec:identity}, its inequality~(S) involves nine real parameters.  Since the number of boxes that interval arithmetic needs can grow exponentially with the dimension, reducing the parameter count is essential for a practical verification.  In this section, we bring the nine parameters down to four: we reduce the existence of a feasible potential to conditions \Hcond{0}--\Hcond{4} on a single function \(H\) of the four angles that describe a gate state (\cref{thm:sufficient}).  This takes two reductions, the third and fourth rows of \cref{tab:proofmap}, based on two ideas: symmetry and continuity.

The first idea is symmetry.  Translations along the \(x\)-axis leave all costs and progresses unchanged, and scalings centered on the axis multiply them by the same factor, so we describe a gate state by four angles and take the potential to be the radius of its disk times a function of these angles; this reduction is tight and loses no feasible potential.  It leaves seven parameters: the four angular coordinates \((\gamma,\varphi,\theta,\psi)\) of the gate state, two coordinates for the new gate, and one for the next disk along the pencil (\cref{subsec:angles}).

The second idea is continuity, which removes the three parameters of the next disk; a gate state can have infinitely many next disks.  The key observation is that the cost and the progress of a step are determined by the current gate state and the new gate, whichever circle through the two ends of the new gate bounds the next disk.  So we reach the next gate state continuously, in two stages (\cref{fig:pencil}): first the gate slides along the current circle to the new gate, and then the disk slides along the pencil of the new gate, that is, through the circles that pass through its two ends.  We charge the cost and the progress to the first stage; the second stage then only has to keep the potential from increasing, and the pencil lemma (\cref{lem:pencil}) computes the rate at which the potential changes along it.  So conditions at the current gate state handle all next disks at once.  Unlike the first reduction, the second may not be tight: paying for the two stages of a step separately can overestimate the true cost, so the sufficient conditions may exclude some feasible potentials.  \Cref{subsec:angles} introduces the four angles, the angular coordinates of a gate state; \cref{subsec:pencil} writes the potential in terms of a function \(H\) of the four angles, which loses nothing, and proves the pencil lemma; \cref{subsec:conditions} states conditions \Hcond{0}--\Hcond{4} and \cref{thm:sufficient}; and \cref{subsec:proof-sufficient} proves the theorem.

\subsection{Symmetries and Angular Coordinates}\label{subsec:angles}

We first remove the two parameters that describe the position of a gate state along the \(x\)-axis and its size.  Costs and progresses are computed from lengths of arcs and chords and from distances along the \(x\)-axis, so we describe a gate state by four angles (\cref{lem:range}), in which each of these lengths is the radius of the disk times a function of the angles.  These angles are the variables of the function \(H\) of \cref{thm:sufficient} and of the rest of the paper.  Roughly speaking (\cref{fig:angles}), \(\gamma\) measures the height of the center above the \(x\)-axis, \(\varphi\) is the direction of the gate's normal, \(2\theta\) is the central angle of the arc of \(\partial O\) from \(a_G\) to \(b_G\) through \(w_-(O)\), and \(\psi-\varphi\) is \(\delta\) in units of the diameter \(2r\).  For \(t\in\R\), let \(e(t)=(\cos t,\sin t)\).

\begin{definition}[Angular coordinates; the sets \(\Omega\) and \(\Sigma\)]\label{def:angles}
Let \(s=(O,G,\delta)\) be a gate state, where \(O\) has center \(o\) and radius \(r\), and \(G=[a_G,b_G]\) has midpoint \(M\).  Since \(y(a_G)>0>y(b_G)\), the segment \(G\) is not horizontal, so it has a unique unit normal vector \(\nu=e(\varphi)\) with positive \(x\)-component, where \(\varphi\in(-\frac\pi2,\frac\pi2)\).  The center \(o\) lies on the perpendicular bisector of \(G\), so \(o=M+X\nu\) for a unique \(X\in\R\).  The \textbf{\emph{angular coordinates}} of \(s\) are \((\gamma,\varphi,\theta,\psi)\), where \(\gamma\in(-\frac\pi2,\frac\pi2)\) is given by \(\sin\gamma=y(o)/r\), the angle \(\theta\in(0,\pi)\) is given by \(\cos\theta=X/r\) and \(\sin\theta=\abs{G}/(2r)\), and \(\psi=\varphi+\delta/(2r)\).  Let
\begin{gather*}
\Omega=\bigl\{(\gamma,\varphi,\theta)\in\R^3:\ \abs{\gamma-\varphi}<\theta<\pi-\abs{\gamma+\varphi}\bigr\},\\
\Sigma=\bigl\{(\gamma,\varphi,\theta,\psi)\in\R^4:\ (\gamma,\varphi,\theta)\in\Omega,\ \abs{\psi-\varphi}\le\sin\theta\bigr\}.
\end{gather*}
\end{definition}

The inequalities that define \(\Omega\) say that \(a_G\) lies above and \(b_G\) below the \(x\)-axis, and the inequality \(\abs{\psi-\varphi}\le\sin\theta\) says that \(\abs\delta\le\abs G\) (\cref{lem:range} below).  We add \(\varphi\) so that \(\psi\) does not change when the gate slides forward along \(\partial O\) in a move and \(\delta\) is not clipped: such a motion changes \(\varphi\) and \(\delta/(2r)\) by opposite amounts (see the proof of \cref{lem:gate-move}).

\begin{figure}[!htbp]
\centering
\begin{tikzpicture}[paperfig,scale=1.5]
\path[use as bounding box] (-3.64,-1.726) rectangle (3.64,3.324);
\draw[fdisk] (0,0.824) circle[radius=2.2];
\draw[faxis,-{Latex[length=1.8mm]}] (-2.64,0) -- (2.74,0);
\draw[fguide] (0,0.824) -- (-1.845,2.022);
\draw[fguide] (0,0.824) -- (-0.999,-1.136);
\draw[fguide] (0,0.824) -- (2.04,0);
\draw[fguide] (0,0.824) -- (-1.422,0.443);
\draw[fguide] (0,0.824) -- (1.25,0.824);
\draw[fguide] (-1.303,0) -- (-1.303,1.05);
\draw[fgate] (-1.845,2.022) -- (-0.999,-1.136);
\draw[fvec,-{Latex[length=1.8mm]}] (-1.422,0.443) -- (-0.823,0.604);
\draw[fangle] (0.575,0.592) arc[start angle=-22,end angle=0,radius=0.62];
\node[anchor=west] at (0.78,0.694) {$\gamma$};
\draw[fangle] (-1.303,0.8) arc[start angle=90,end angle=105,radius=0.8];
\node[anchor=south] at (-1.434,1.041) {$\varphi$};
\draw[fangle] (-0.483,0.695) arc[start angle=195,end angle=243,radius=0.5];
\draw[fangle] (-0.352,1.053) arc[start angle=147,end angle=195,radius=0.42];
\node[] at (-0.56,0.371) {$\theta$};
\node[] at (-0.652,0.927) {$\theta$};
\node[fpt] at (-1.845,2.022) {};
\node[fpt] at (-0.999,-1.136) {};
\node[fpt] at (-2.04,0) {};
\node[fpt] at (2.04,0) {};
\node[fpt] at (0,0.824) {};
\node[fpt] at (-1.422,0.443) {};
\node[fpto] at (-1.303,0) {};
\node[anchor=south east] at (-1.925,2.082) {$a_G$};
\node[anchor=north east] at (-1.079,-1.196) {$b_G$};
\node[anchor=north east] at (-2.08,-0.08) {$w_-(O)$};
\node[anchor=north west] at (2.08,-0.08) {$w_+(O)$};
\node[anchor=south] at (0.05,0.944) {$o$};
\node[anchor=east] at (-1.522,0.463) {$M$};
\node[anchor=north west] at (-1.17,-0.08) {$c(s)$};
\node[anchor=south] at (-1.036,0.652) {$\nu$};
\node[anchor=east,text=figOrangeText] at (-1.33,-0.346) {$G$};
\node[text=black] at (1.5,2.7) {$O$};
\end{tikzpicture}
\caption{Four angular coordinates describe the gate state \(s=(O,G,\delta)\).  The angle \(\gamma\) satisfies \(y(o)=r\sin\gamma\).  The unit normal \(\nu=e(\varphi)\) of \(G\) makes the angle \(\varphi\) with the \(x\)-axis, so the gate makes the same angle with the vertical direction.  The center is \(o=M+X\nu\), where \(X=r\cos\theta\).  The two angles marked \(\theta\) together span the arc from \(a_G\) to \(b_G\) through \(w_-(O)\).  Here \(X>0\) and \(\theta<\pi/2\).  If the center lies on the other side of the gate, then \(X<0\) and \(\theta>\pi/2\).  The fourth coordinate is \(\psi=\varphi+\delta/(2r)\).}
\label{fig:angles}
\end{figure}

The next lemma identifies gate states, up to translations along the \(x\)-axis and scalings centered on it, with the points of the four-dimensional set \(\Sigma\), and expresses in angular coordinates the arc lengths and \(x\)-coordinates from which the costs and progresses of \cref{def:moves} are computed.

\begin{lemma}[Range of the coordinates]\label{lem:range}
\leavevmode
\begin{enumerate}[label=(\arabic*)]
\item\label{range:sigma} The angular coordinates of every gate state lie in \(\Sigma\).  Conversely, for every \(z\in\Sigma\), there is a gate state with angular coordinates \(z\), and it is unique up to translations along the \(x\)-axis and scalings with positive ratio centered at points of the \(x\)-axis.
\item\label{range:arcs} Let \(s=(O,G,\delta)\) be a gate state, where \(O\) has center \(o\) and radius \(r\), with angular coordinates \((\gamma,\varphi,\theta,\psi)\).  Then the upper and the lower arc of \(O\) have lengths \(r(\pi+2\gamma)\) and \(r(\pi-2\gamma)\); the part of the upper arc from \(w_-(O)\) to \(a_G\) has length \(r(\theta+\gamma-\varphi)\), and the part of the lower arc from \(w_-(O)\) to \(b_G\) has length \(r(\theta-\gamma+\varphi)\), so that the arc of \(\partial O\) from \(a_G\) to \(b_G\) through \(w_-(O)\) has central angle \(2\theta\); \(x(w_\pm(O))=x(o)\pm r\cos\gamma\); and
\begin{equation}\label{eq:crossing}
x(c(s))-x(o)=r\,\frac{\sin\gamma\sin\varphi-\cos\theta}{\cos\varphi}.
\end{equation}
\end{enumerate}
\end{lemma}

\begin{proof}
The angular coordinates of a gate state are well defined: \(\abs{y(o)}<r\) because \(\partial O\) meets the \(x\)-axis in two points, and \(X^2+\abs{G}^2/4=\norm{o-a_G}^2=r^2\).  Every \((\gamma,\varphi,\theta)\in\Omega\) satisfies \(\abs\gamma<\frac\pi2\), \(\abs\varphi<\frac\pi2\) and \(0<\theta<\pi\), since \(\abs{\gamma-\varphi}+\abs{\gamma+\varphi}=2\max(\abs\gamma,\abs\varphi)\).  We prove~\ref{range:arcs} first, since~\ref{range:sigma} uses the formulas~\eqref{eq:polar}.

\ref{range:arcs} Since \(\nu=e(\varphi)\) and \(y(a_G)>y(b_G)\), the vector \(a_G-b_G\) is a positive multiple of \(e(\varphi+\frac\pi2)=(-\sin\varphi,\cos\varphi)\).  Using \(M=o-r\cos\theta\,e(\varphi)\) and \(\abs G/2=r\sin\theta\), we obtain
\begin{equation}\label{eq:polar}
a_G=o+r\,e(\pi-\theta+\varphi),\qquad b_G=o+r\,e(\pi+\theta+\varphi).
\end{equation}
Since \(y(o)=r\sin\gamma\), the two points of \(\partial O\) on the \(x\)-axis are \(w_-(O)=o+r\,e(\pi+\gamma)\) and \(w_+(O)=o+r\,e(-\gamma)\); in particular, \(x(w_\pm(O))=x(o)\pm r\cos\gamma\).  So the upper arc consists of the points \(o+r\,e(t)\) with \(-\gamma\le t\le\pi+\gamma\), traversed from \(t=\pi+\gamma\) down to \(t=-\gamma\), and the lower arc consists of the points with \(\pi+\gamma\le t\le2\pi-\gamma\), traversed from \(t=\pi+\gamma\) up to \(t=2\pi-\gamma\).  This gives the lengths of the two arcs, and by~\eqref{eq:polar}, the parts from \(w_-(O)\) to \(a_G\) and to \(b_G\) have central angles \((\pi+\gamma)-(\pi-\theta+\varphi)\) and \((\pi+\theta+\varphi)-(\pi+\gamma)\).  Finally, the points of the line through \(G\) are \(M+t\,e(\varphi+\frac\pi2)\) with \(t\in\R\); the one on the \(x\)-axis has \(t=r(\cos\theta\sin\varphi-\sin\gamma)/\cos\varphi\), and its \(x\)-coordinate gives~\eqref{eq:crossing}.

\ref{range:sigma} Both directions rest on one computation: for arbitrary values of the angles, we determine when the points~\eqref{eq:polar} lie on opposite sides of the \(x\)-axis.  Let \(\gamma,\varphi\in(-\frac\pi2,\frac\pi2)\) and \(\theta\in(0,\pi)\), let \(O\) be a closed disk with center \(o\) and radius \(r>0\) such that \(y(o)=r\sin\gamma\), and define \(a_G\), \(b_G\) by~\eqref{eq:polar}.  The angle \(\pi-\theta+\varphi\) lies in the interval \((-\frac\pi2,\frac{3\pi}2)\), which has length \(2\pi\) and contains \([-\gamma,\pi+\gamma]\); so, by the description of the upper arc above, \(y(a_G)>0\) if and only if \(-\gamma<\pi-\theta+\varphi<\pi+\gamma\).  Likewise, \(\pi+\theta+\varphi\in(\frac\pi2,\frac{5\pi}2)\), so \(y(b_G)<0\) if and only if \(\pi+\gamma<\pi+\theta+\varphi<2\pi-\gamma\).  Together, these four inequalities say that \(\abs{\gamma-\varphi}<\theta<\pi-\abs{\gamma+\varphi}\).  Moreover, \(\abs\delta\le\abs G\) if and only if \(\abs{\psi-\varphi}\le\sin\theta\).

Hence the angular coordinates of a gate state lie in \(\Sigma\).  Conversely, let \(z=(\gamma,\varphi,\theta,\psi)\in\Sigma\), let \(r>0\) and \(x_0\in\R\), let \(O\) be the closed disk with center \(o=(x_0,r\sin\gamma)\) and radius \(r\), define \(a_G,b_G\) by~\eqref{eq:polar}, and let \(\delta=2r(\psi-\varphi)\).  By the previous paragraph and the first paragraph of this proof, \((O,[a_G,b_G],\delta)\) is a gate state, and its angular coordinates are \(z\).  The computation above shows that every gate state with angular coordinates \(z\) arises in this way, so it is determined by \(r\) and \(x_0\).  Changing \(x_0\) is a translation along the \(x\)-axis, and changing \(r\) is a scaling centered at \((x_0,0)\).
\end{proof}

We look for a potential with the same symmetries as the costs and the progresses; this restriction loses no generality.  Translations along the \(x\)-axis and scalings with positive ratio centered at points of the \(x\)-axis map every finite sequence of steps followed by the final step to such a sequence and multiply all costs and progresses by the ratio of the scaling.  So the value function \(W^*_U\) of \cref{prop:exact} is invariant under these translations and positively homogeneous of degree one under these scalings, and by \cref{thm:exactness} it is a feasible potential whenever the chain inequality holds at a level \(U\ge\pi/2\).  By \cref{lem:range}\ref{range:sigma}, a function on gate states with these two properties is \(r\) times a function of the angular coordinates.  For such a potential, (S) is an inequality in seven parameters: the four angular coordinates of the gate state, two for the new gate and one for the next disk among the circles through its two ends.

\subsection{The Pencil Lemma}\label{subsec:pencil}

In this subsection, we fix the form in which we look for a feasible potential (\cref{def:residual}), and we prove the pencil lemma (\cref{lem:pencil}), which computes the rate at which such a potential changes when the disk slides along the pencil of its gate.  The lemma removes the next disk from the conditions: once the new gate is fixed, it replaces the infinitely many next disks whose circles pass through the two ends of this gate by one inequality at each gate state, condition \Hcond{4} of \cref{subsec:conditions}.

We now make precise the two stages of a step described at the beginning of this section (\cref{fig:pencil}(a) and~(b)); the two stages separate the choice of the new gate from the choice of the next disk among the circles through its two ends.  Let \(s=(O,G,\delta)\) be a gate state, let \(O'\) be a next disk for \(s\), and let the step from \(s\) to \(O'\) lead to \(s'=(O',G',\delta')\), where \(G'=[a',b']\) is the new gate.  In the first stage, the gate slides along the current circle \(\partial O\) from \(G\) to \(G'\) and the difference changes from \(\delta\) to \(\delta'\), while the disk stays.  In the second stage, the disk slides from \(O\) to \(O'\) through the circles that pass through \(a'\) and \(b'\), the \emph{pencil} of \(G'\), while the gate and the difference stay.  The cost \(\Delta\mu\) and the progress \(\Delta x\) of the step are determined by \(s\) and \(G'\) alone (\cref{def:moves}\ref{mv:step}), so we charge both to the first stage; the second stage then only has to keep the potential from increasing.

Accordingly, \cref{def:residual} below writes the potential as explicit terms minus the radius of the disk times a function \(H\) of the angular coordinates, where the explicit terms absorb the cost and the progress of the first stage.  The first stage then only constrains how \(H\) varies with the gate and the difference (\cref{subsec:conditions}).  Along the second stage, the pencil lemma computes exactly how fast such a potential changes, as an explicit expression in \(H\) and its derivatives at the current angular coordinates, the \emph{pencil residual} of \cref{def:residual}.

\begin{definition}[The potential of \(H\) and the pencil residual]\label{def:residual}
Let \(U\in\R\), and let \(H\) be a real function on \(\Omega\times\R\).
\begin{enumerate}[label=(\arabic*)]
\item The \emph{potential of \(H\)} at level \(U\) is the function on gate states
\begin{equation}\label{eq:potential}
W_{H,U}(s)=r\Bigl[\frac\pi2-\theta-H(\gamma,\varphi,\theta,\psi)\Bigr]+U\bigl(x(c(s))-x(o)\bigr),
\end{equation}
where \(o\) and \(r\) are the center and the radius of the disk of \(s\), and \((\gamma,\varphi,\theta,\psi)\) are the angular coordinates of \(s\).
\item For \(z=(\gamma,\varphi,\theta,\psi)\in\Omega\times\R\) and \(\nabla=(\nabla_\gamma,\nabla_\varphi,\nabla_\theta,\nabla_\psi)\in\R^4\), the \textbf{\emph{pencil residual}} is
\begin{equation}\label{eq:pencil-residual}
\cE_U[H;\nabla](z)=\cos\theta\Bigl(\frac\pi2-\theta\Bigr)+\sin\theta-U\cos\varphi-\cos\theta\,H(z)+\cos\theta\,(\psi-\varphi)\,\nabla_\psi-\cA\,\nabla_\gamma+\sin\theta\,\nabla_\theta,
\end{equation}
where \(\cA=(\sin\varphi-\sin\gamma\cos\theta)/\cos\gamma\).  If \(H\) is differentiable at \(z\), we set \(\cE_U[H](z)=\cE_U[H;\nabla H(z)](z)\), where \(\nabla H\) is the gradient with respect to \((\gamma,\varphi,\theta,\psi)\).
\end{enumerate}
\end{definition}

\paragraph{The Potential.}  The explicit terms \(r(\frac\pi2-\theta)\) and \(U(x(c(s))-x(o))\) of \(W_{H,U}\) absorb the cost and the progress of the first stage of a step.  Indeed, by \cref{lem:range}\ref{range:arcs}, \(r\theta\) is the average of the lengths of the arcs from \(w_-(O)\) to \(a_G\) and to \(b_G\), so the change of \(-r\theta\) cancels the term \(\frac12(\alpha+\beta)\) of the cost of the step, and the change of \(U\,x(c(s))\) cancels \(U\) times its progress; what remains is \(-\frac\omega2\) minus the change of \(rH\) (see~\eqref{eq:gate-move} in the proof of \cref{lem:gate-move}).  So \(H\) measures how the potential deviates from these explicit terms; the constant \(\frac\pi2\) and the domain \(\Omega\times\R\) of \(H\) are explained after \cref{def:conditions}.  By \cref{lem:range}\ref{range:arcs}, \(W_{H,U}(s)\) is \(r\) times a function of the angular coordinates of \(s\).  Conversely, by~\eqref{eq:crossing}, the ratio \((x(c(s))-x(o))/r\) also depends only on the angular coordinates, so every real function \(W\) on gate states with the two symmetries of \cref{subsec:angles} is \(W_{H,U}\) for exactly one function \(H\) on \(\Sigma\), namely
\[
H=\frac\pi2-\theta+U\,\frac{\sin\gamma\sin\varphi-\cos\theta}{\cos\varphi}-\frac{W}{r}.
\]
So this form loses no generality.  \Cref{thm:sufficient} below gives sufficient conditions on \(H\) for \(W_{H,U}(s)\) to be a feasible potential.  These conditions, \Hcond{0}--\Hcond{4}, are stated in \cref{subsec:conditions} after the pencil lemma.

\paragraph{The Pencil Residual.}  Recall that the center of the disk of a gate state is \(o=M+X\nu\), where \(M\) is the midpoint of the gate (\cref{def:angles}); along the pencil of the gate, \(X\) parametrizes the disk, and it increases as the center moves forward (\cref{lem:pencil}\ref{pencil:family} below).  By \cref{lem:pencil}\ref{pencil:rate}, \(\cE_U[H]\) is the derivative of \(W_{H,U}\) with respect to \(X\) along the pencil; so condition \Hcond{4} below, that the pencil residual be nonpositive, says that the potential does not increase as the disk slides forward along the pencil.  The expression \(\cE_U[H;\nabla]\) is the same formula, with a vector \(\nabla\) in place of the gradient of \(H\), because \(H\) may have several gradients at a point (\cref{def:piecewise}).  The number \(\cA\) is the rate \(r\,d\gamma/dX\) at which \(\gamma\) changes along the pencil (\cref{lem:pencil}\ref{pencil:derivative}), and \(\nabla_\varphi\) does not occur because \(\varphi\) is constant along a pencil.  The term \(-U\cos\varphi\) comes from the second term of \(W_{H,U}\): along a pencil, \(c(s)\) is fixed and \(x(o)\) increases at the rate \(\cos\varphi\).

\begin{lemma}[Pencil lemma]\label{lem:pencil}
Let \(G=[a_G,b_G]\) be a segment with \(y(a_G)>0>y(b_G)\), midpoint \(M\), and unit normal vector \(\nu\) with positive \(x\)-component.  For \(X\in\R\), let \(O(X)\) be the closed disk with center \(o(X)=M+X\nu\) and radius \(r(X)=\sqrt{X^2+\abs G^2/4}\).
\begin{enumerate}[label=(\arabic*)]
\item\label{pencil:family} The closed disks whose boundary passes through \(a_G\) and \(b_G\) are exactly the disks \(O(X)\), \(X\in\R\).  For \(X_1\ne X_2\), the circles \(\partial O(X_1)\) and \(\partial O(X_2)\) meet exactly in \(a_G\) and \(b_G\), and \(x(o(X_2))>x(o(X_1))\) if and only if \(X_2>X_1\).
\item\label{pencil:derivative} Let \(\delta\in[-\abs G,\abs G]\).  For every \(X\in\R\), the triple \(s(X)=(O(X),G,\delta)\) is a gate state.  Its angular coordinates \(z(X)=(\gamma,\varphi,\theta,\psi)\) are smooth functions of \(X\), \(\varphi\) does not depend on \(X\), and
\begin{equation}\label{eq:pencil-derivatives}
r'=\cos\theta,\qquad r\gamma'=\cA,\qquad r\theta'=-\sin\theta,\qquad r\psi'=-(\psi-\varphi)\cos\theta,
\end{equation}
where \(r=r(X)\), the prime denotes \(d/dX\), and \(\cA\) is as in \cref{def:residual}.
\item\label{pencil:rate} Let \(\delta\), \(s(X)\) and \(z(X)\) be as in~\ref{pencil:derivative}, let \(U\in\R\), let \(H\) be a real function on \(\Omega\times\R\), and let \(X_0\in\R\) and \(\nabla\in\R^4\) be such that \(H\circ z\) is differentiable at \(X_0\) with \((H\circ z)'(X_0)=\nabla\cdot z'(X_0)\).  Then \(X\mapsto W_{H,U}(s(X))\) is differentiable at \(X_0\), and its derivative there is \(\cE_U[H;\nabla](z(X_0))\).  In particular, if \(H\) is differentiable at \(z(X_0)\), the derivative is \(\cE_U[H](z(X_0))\).
\end{enumerate}
\end{lemma}

\begin{proof}
\ref{pencil:family} A circle passes through \(a_G\) and \(b_G\) if and only if its center is equidistant from them, that is, lies on the perpendicular bisector \(\{M+X\nu:X\in\R\}\) of \(G\); its radius is then \(\norm{M+X\nu-a_G}=r(X)\).  Two distinct circles have at most two common points.  Finally, \(x(o(X))=x(M)+X\,x(\nu)\) increases with \(X\), since \(x(\nu)>0\).

\ref{pencil:derivative} The triple \(s(X)\) is a gate state, since \(G\) is a chord of \(\partial O(X)\) and \(\abs\delta\le\abs G\).  In \cref{def:angles}, the vector \(\nu=e(\varphi)\) and the point \(c(s(X))\) depend only on \(G\), and the remaining coordinates are given by \(\cos\theta=X/r\), \(\sin\gamma=y(o(X))/r\) and \(\psi=\varphi+\delta/(2r)\), which are smooth in \(X\) because \(r>0\), \(\theta\in(0,\pi)\) and \(\gamma\in(-\frac\pi2,\frac\pi2)\).  Differentiating \(r^2=X^2+\abs G^2/4\) gives \(r'=X/r=\cos\theta\).  Differentiating \(\cos\theta=X/r\) gives \(-\sin\theta\,\theta'=(r-Xr')/r^2=\sin^2\theta/r\), so \(r\theta'=-\sin\theta\).  Differentiating \(r\sin\gamma=y(M)+X\sin\varphi\) gives \(r'\sin\gamma+r\cos\gamma\,\gamma'=\sin\varphi\), so \(r\gamma'=\cA\).  Differentiating \(\psi=\varphi+\delta/(2r)\) gives \(r\psi'=-\delta r'/(2r)=-(\psi-\varphi)\cos\theta\).  This proves~\eqref{eq:pencil-derivatives}.

\ref{pencil:rate} Suppose that \((H\circ z)'(X_0)=\nabla\cdot z'(X_0)\).  Since \(x(c(s(X)))\) does not depend on \(X\) and \(x(o(X))\) has derivative \(\cos\varphi\), the product rule gives, at \(X=X_0\),
\[
\frac{d}{dX}W_{H,U}(s(X))=r'\Bigl(\frac\pi2-\theta-H(z)\Bigr)-r\theta'-r\bigl(\nabla_\gamma\gamma'+\nabla_\theta\theta'+\nabla_\psi\psi'\bigr)-U\cos\varphi,
\]
where we used \(\varphi'=0\).  Substituting~\eqref{eq:pencil-derivatives}, the right-hand side becomes \(\cE_U[H;\nabla](z(X_0))\).
\end{proof}

By the pencil lemma, the potential does not increase along the second stage of a step if the pencil residual is nonpositive on \(\Sigma\) and \(H\) is continuous and locally Lipschitz (\cref{cor:pencil-move}); these two requirements are conditions \Hcond{4} and \Hcond{0} of \cref{subsec:conditions}.  Indeed, by \cref{lem:pencil}\ref{pencil:family} and condition~\ref{n:forward}, the second stage runs through the pencil of \(G'\) in the direction of increasing \(X\), along which the derivative of the potential is the pencil residual (\cref{lem:pencil}\ref{pencil:rate}).  Xia moves the last disk of a chain along such a pencil in the proof of~\cite[Proposition~7]{xia2013}, for one specific function; the pencil lemma gives the exact rate of change for every \(H\).

\subsection{Sufficient Conditions}\label{subsec:conditions}

In this subsection, we state conditions \Hcond{0}--\Hcond{4} on \(H\) and show that they make \(W_{H,U}\) a feasible potential (\cref{thm:sufficient}).  Together with \cref{thm:criterion}, this reduces the chain inequality at level \(U\) to the existence of one function \(H\) that satisfies these conditions; \cref{sec:family,sec:verification} provide such a function for \(U=U_0\).  Each condition pays for one part of a move.  In the first stage of a step, the two ends of the gate move forward along the circle, and the explicit terms of \(W_{H,U}\) already pay for the lengths of the arcs that they travel and for the progress.  So it suffices that \(H\) does not decrease when an end of the gate moves forward, which is \Hcond{2}, and that, when a shortest path switches rails along the new gate, the potential changes by no more than the resulting saving in the cost, which is the Lipschitz condition \Hcond{1}.  In the second stage, the disk slides along the pencil of the new gate, and by the pencil lemma, it suffices that the pencil residual be nonpositive, which is \Hcond{4}.  The ends of a chain need, in addition, bounds on \(H\) at two values of \(\psi\); these bounds are condition~\Hcond{3}.  A regularity condition, \Hcond{0}, makes the rates along the two stages meaningful.

Two of the conditions, \Hcond{2} and \Hcond{4}, are inequalities for the gradient of \(H\).  We allow \(H\) to be only piecewise differentiable, because functions with a clipped variable are convenient to construct, and the function of \cref{sec:family} is of this kind.  We therefore fix a finite cover of \(\Omega\times\R\) by closed pieces on each of which \(H\) agrees with a \(C^1\) function, and we require these inequalities for the gradient of every piece that contains the point.

\begin{definition}[Piecewise \(C^1\) structures]\label{def:piecewise}
Let \(H\) be a real function on \(\Omega\times\R\).  A \textbf{\emph{piecewise \(C^1\) structure}} for \(H\) consists of finitely many closed subsets \(Z_1,\dots,Z_m\) of \(\Omega\times\R\) that cover \(\Omega\times\R\), and functions \(H_1,\dots,H_m\), where \(H_i\) is continuously differentiable on an open neighborhood of \(Z_i\), such that \(H=H_i\) on \(Z_i\) for every \(i\).  For \(z\in\Omega\times\R\), the \textbf{\emph{incident gradients}} of \(H\) at \(z\) with respect to this structure are the vectors \(\nabla H_i(z)\) for the indices \(i\) with \(z\in Z_i\).
\end{definition}

In the angular coordinates, the first two payments read as follows; \cref{lem:gate-move}, \cref{cor:pencil-move} and \cref{lem:ends} in \cref{subsec:proof-sufficient} carry out all of them.  The first stage of a step, in which the gate slides along the current circle, increases \(\theta\) by at least as much as it changes \(\varphi\); so \Hcond{2} below is a monotonicity condition in these two variables.  The clipping of \(\delta\) in initial pairs and steps, and the minimum in the cost of the final step, lower the cost by \(\frac\omega2\) and move \(\psi\) by \(\omega/(2r)\); so \Hcond{1} is a Lipschitz condition in \(\psi\), which says that \(rH\) changes by at most \(\frac\omega2\) under such a move of \(\psi\).

\begin{definition}[Conditions \Hcond{0}--\Hcond{4}]\label{def:conditions}
Let \(U\in\R\), and let \(H\) be a real function on \(\Omega\times\R\).  We say that \(H\) \emph{satisfies \Hcond{0}--\Hcond{4} at level \(U\)}, or some of these conditions, if there is a piecewise \(C^1\) structure for \(H\), which we fix, such that the following hold, where incident gradients are taken with respect to this structure.
\begin{itemize}[leftmargin=2.7em]
\item[(H0)]\hypertarget{cond:H0}{} \(H\) is continuous and locally Lipschitz.
\item[(H1)]\hypertarget{cond:H1}{} \(\abs{H(\gamma,\varphi,\theta,\psi_1)-H(\gamma,\varphi,\theta,\psi_2)}\le\abs{\psi_1-\psi_2}\) for all \((\gamma,\varphi,\theta)\in\Omega\) and all \(\psi_1,\psi_2\in\R\).
\item[(H2)]\hypertarget{cond:H2}{} \(\nabla_\theta\ge\abs{\nabla_\varphi}\) for every \(z\in\Omega\times\R\) and every incident gradient \(\nabla\) of \(H\) at \(z\).
\item[(H3)]\hypertarget{cond:H3}{} \(H(\gamma,\varphi,\theta,\gamma)\ge\frac\pi2-U\cos\gamma\) and \(H(\gamma,\varphi,\theta,-\gamma)\le-\frac\pi2+U\cos\gamma\) for all \((\gamma,\varphi,\theta)\in\Omega\).
\item[(H4)]\hypertarget{cond:H4}{} \(\cE_U[H;\nabla](z)\le0\) for every \(z\in\Sigma\) and every incident gradient \(\nabla\) of \(H\) at \(z\).
\end{itemize}
\end{definition}

The two values of \(\psi\) and the right-hand sides in \Hcond{3} come from the ends of a chain (\cref{lem:ends}): \(\psi=\gamma\) is the value at the first disk of a chain when \(\delta_1\) is not clipped, and \(\psi=-\gamma\) is the value at which the two terms of the minimum in the cost of the final step are equal.  The term \(U\cos\gamma\) is the part of the progress of an initial pair or a final step that the term \(U(x(c(s))-x(o))\) of \(W_{H,U}\) does not absorb.  The constant \(\frac\pi2\) in \cref{def:residual} normalizes \(H\) so that the two inequalities of \Hcond{3} are exchanged when \(H\) and \(\psi\) are replaced by \(-H\) and \(-\psi\).  Note that \Hcond{1} and \Hcond{2} are required for all real \(\psi\), not only on \(\Sigma\), because the proof of \cref{thm:sufficient} moves \(\psi\) outside the range \(\abs{\psi-\varphi}\le\sin\theta\) (the points \(\hat z\) and \(z_t\) in the proof of \cref{lem:gate-move}); \Hcond{4} is required only on \(\Sigma\).

\begin{theorem}[Sufficient conditions]\label{thm:sufficient}
Let \(U\ge\pi/2\), and let \(H\) be a real function on \(\Omega\times\R\) that satisfies \Hcond{0}--\Hcond{4} at level \(U\).  Then \(W_{H,U}\) is a feasible potential of level \(U\).  Consequently, the chain inequality holds at level \(U\).
\end{theorem}

These conditions are sufficient, but they may be stronger than necessary: the two stages of a step are paid for separately, and \Hcond{1}--\Hcond{3} ask for more than the first stage and the ends of a chain strictly need.  Consequently, the reduction from seven to four parameters may not be tight, whereas the symmetry reduction of \cref{subsec:angles} is tight.

\paragraph{Proof Synopsis.}  Instead of verifying (S) for every next disk directly, we split each step into the two stages of \cref{subsec:pencil}.  The first stage, in which the gate slides along the current circle, is paid for by \Hcond{1} and \Hcond{2}, with \Hcond{0} for the chain rule (\cref{lem:gate-move}); in the second stage, in which the disk slides along the pencil, the potential does not increase, by \Hcond{4} and the pencil lemma (\cref{cor:pencil-move}).  Inequalities (I) and (F) follow from \Hcond{1} and \Hcond{3} (\cref{lem:ends}), together with \cref{cor:pencil-move} for (I).

\subsection{Proof of \texorpdfstring{\cref{thm:sufficient}}{Theorem 4.7}}\label{subsec:proof-sufficient}

The synopsis after \cref{thm:sufficient} names the three results that pay for the moves; they rest on \cref{lem:curves}, the chain rule for \(H\) along a curve, with an incident gradient in place of the gradient.

\begin{lemma}[Derivatives along curves]\label{lem:curves}
Let \(H\) be a real function on \(\Omega\times\R\) that is continuous and locally Lipschitz, that is, satisfies \Hcond{0}, with a piecewise \(C^1\) structure, and let \(z:[t_0,t_1]\to\Omega\times\R\) be a continuously differentiable curve.  Then \(H\circ z\) is absolutely continuous, and for almost every \(t\in[t_0,t_1]\), there is an incident gradient \(\nabla\) of \(H\) at \(z(t)\) with respect to this structure such that \((H\circ z)'(t)=\nabla\cdot z'(t)\).
\end{lemma}

\begin{proof}
The image of \(z\) is compact and \(H\) is locally Lipschitz, so \(H\) is Lipschitz on this image; since \(z\) is continuously differentiable, \(H\circ z\) is Lipschitz.  Hence \(H\circ z\) is absolutely continuous and differentiable almost everywhere.  For each \(i\), let \(T_i=\{t:z(t)\in Z_i\}\), a closed set.  Let \(t\in T_i\) be a point at which \(H\circ z\) is differentiable and which is not an isolated point of \(T_i\).  Choose \(t_n\in T_i\setminus\{t\}\) with \(t_n\to t\).  Since \(H=H_i\) on \(Z_i\) and \(H_i\circ z\) is differentiable at \(t\),
\[
(H\circ z)'(t)=\lim_{n\to\infty}\frac{H_i(z(t_n))-H_i(z(t))}{t_n-t}=\nabla H_i(z(t))\cdot z'(t),
\]
and \(\nabla H_i(z(t))\) is an incident gradient of \(H\) at \(z(t)\).  Every \(t\in[t_0,t_1]\) lies in some \(T_i\), and each \(T_i\) has only countably many isolated points.  So the formula holds for all \(t\) outside a set of measure zero.
\end{proof}

The next lemma handles the first stage of a step.  In it, \(\tilde s=(O,G',\delta')\) is the intermediate state after this stage: the current disk with the new gate and the new difference.

\begin{lemma}[Sliding the gate along the current circle]\label{lem:gate-move}
Let \(U\in\R\), and let \(H\) satisfy \Hcond{0}, \Hcond{1} and \Hcond{2}.  Let \(s=(O,G,\delta)\) be a gate state, let \(O'\) be a next disk for \(s\), let \(G'=[a',b']\) be the common chord of \(\partial O\) and \(\partial O'\), and let \(\delta'\), \(\Delta\mu\) and \(\Delta x\) be the new difference, the cost and the progress of the step from \(s\) to \(O'\) (\cref{def:moves}\ref{mv:step}).  Then \(\tilde s=(O,G',\delta')\) is a gate state, and
\begin{equation}\label{eq:gate-move-bound}
\Delta\mu-U\,\Delta x+W_{H,U}(\tilde s)-W_{H,U}(s)\le0 .
\end{equation}
\end{lemma}

\begin{proof}
The triple \(\tilde s\) is a gate state, since \(G'=[a',b']\) is a chord of \(\partial O\) with \(y(a')>0>y(b')\) and \(\abs{\delta'}\le\abs{G'}\).  Let \(r\) be the radius of \(O\), and let \(z=(\gamma,\varphi,\theta,\psi)\) and \(\tilde z=(\gamma,\tilde\varphi,\tilde\theta,\tilde\psi)\) be the angular coordinates of \(s\) and \(\tilde s\); the first coordinates agree because the disk is the same.  Recall that \(\alpha\) and \(\beta\) are the lengths of the parts of the upper arc of \(O\) from \(a_G\) to \(a'\) and of the lower arc from \(b_G\) to \(b'\), where \(a_G\) does not come after \(a'\) and \(b_G\) does not come after \(b'\).  By \cref{lem:range}\ref{range:arcs},
\[
r(\tilde\theta+\gamma-\tilde\varphi)=r(\theta+\gamma-\varphi)+\alpha,\qquad
 r(\tilde\theta-\gamma+\tilde\varphi)=r(\theta-\gamma+\varphi)+\beta,
\]
so \(\tilde\theta-\theta=(\alpha+\beta)/(2r)\) and \(\tilde\varphi-\varphi=(\beta-\alpha)/(2r)\).  In particular, \(\tilde\theta-\theta\ge\abs{\tilde\varphi-\varphi}\).  Since \(\delta'=\clip_{\abs{G'}}(\delta+\alpha-\beta)\), we have \(\abs{\delta'-(\delta+\alpha-\beta)}=\omega\), where \(\omega\) is as in \cref{def:moves}\ref{mv:step}, and since \(\tilde\varphi+(\delta+\alpha-\beta)/(2r)=\psi\), we obtain \(\abs{\tilde\psi-\psi}=\omega/(2r)\).

By \cref{def:residual} and \cref{def:moves}\ref{mv:step}, \(W_{H,U}(\tilde s)-W_{H,U}(s)=-r(\tilde\theta-\theta)-r[H(\tilde z)-H(z)]+U\,\Delta x\) and \(\Delta\mu=\frac12(\alpha+\beta-\omega)=r(\tilde\theta-\theta)-\frac\omega2\).  Hence
\begin{equation}\label{eq:gate-move}
\Delta\mu-U\,\Delta x+W_{H,U}(\tilde s)-W_{H,U}(s)=-\frac\omega2-r\bigl[H(\tilde z)-H(z)\bigr].
\end{equation}
Let \(\hat z=(\gamma,\tilde\varphi,\tilde\theta,\psi)\), which lies in \(\Omega\times\R\) because \(\tilde s\) is a gate state.  By \Hcond{1}, \(H(\tilde z)\ge H(\hat z)-\abs{\tilde\psi-\psi}=H(\hat z)-\omega/(2r)\).  To compare \(H(\hat z)\) with \(H(z)\), consider the segment \(z_t=(\gamma,\varphi+t(\tilde\varphi-\varphi),\theta+t(\tilde\theta-\theta),\psi)\), \(0\le t\le1\).  For fixed \(\gamma\), the set of pairs \((\varphi,\theta)\) with \((\gamma,\varphi,\theta)\in\Omega\) is an intersection of four open half-planes, hence convex, so the segment lies in \(\Omega\times\R\).  By \cref{lem:curves}, \(t\mapsto H(z_t)\) is absolutely continuous, and for almost every \(t\), there is an incident gradient \(\nabla\) at \(z_t\) with
\[
\frac{d}{dt}H(z_t)=\nabla_\varphi(\tilde\varphi-\varphi)+\nabla_\theta(\tilde\theta-\theta)\ge\bigl(\nabla_\theta-\abs{\nabla_\varphi}\bigr)(\tilde\theta-\theta)\ge0,
\]
where we used \(\tilde\theta-\theta\ge\abs{\tilde\varphi-\varphi}\) and \Hcond{2}.  Hence \(H(\hat z)\ge H(z)\), so \(H(\tilde z)\ge H(z)-\omega/(2r)\), and the right-hand side of~\eqref{eq:gate-move} is at most \(0\).
\end{proof}

The next corollary handles the second stage of a step: along the pencil of its gate, the potential does not increase.

\begin{corollary}[Sliding the disk along a pencil]\label{cor:pencil-move}
Let \(U\in\R\), and let \(H\) satisfy \Hcond{0} and \Hcond{4} at level \(U\).  Let \(G=[a_G,b_G]\) be a segment with \(y(a_G)>0>y(b_G)\), let \(\delta\in[-\abs G,\abs G]\), and let \(s(X)=(O(X),G,\delta)\), \(X\in\R\), be the gate states along the pencil of \(G\), as in \cref{lem:pencil}\ref{pencil:derivative}.  Then \(W_{H,U}(s(X_2))\le W_{H,U}(s(X_1))\) whenever \(X_1<X_2\).
\end{corollary}

\begin{proof}
By \cref{lem:range}\ref{range:sigma}, the curve \(z(X)\), \(X_1\le X\le X_2\), lies in \(\Sigma\), and it is smooth by \cref{lem:pencil}\ref{pencil:derivative}.  By \cref{lem:curves}, \(H\circ z\) is absolutely continuous, and for almost every \(X\), there is an incident gradient \(\nabla\) at \(z(X)\) with \((H\circ z)'(X)=\nabla\cdot z'(X)\).  At every such \(X\), by \cref{lem:pencil}\ref{pencil:rate} and \Hcond{4}, the function \(X\mapsto W_{H,U}(s(X))\) has derivative \(\cE_U[H;\nabla](z(X))\le0\).  This function is absolutely continuous, since \(r\), \(\theta\) and \(x(o)\) are smooth in \(X\) and \(H\circ z\) is absolutely continuous; so it does not increase.
\end{proof}

\Cref{fig:step-stages} summarizes the two stages and the conditions used in each.

\begin{figure}[!htbp]
\centering
\begin{tikzpicture}[paperfig,font=\normalsize,text=black]
\path[use as bounding box] (0,0) rectangle (15.6,2.2);
\tikzset{stepstate/.style={draw=figInk,line width=1pt,rounded corners=2pt,
  fill=figFill,minimum width=3.4cm,minimum height=1.1cm,inner sep=7pt,
  text=black,font=\normalsize}}
\node[stepstate] (start) at (1.8,1.1) {$s=(O,G,\delta)$};
\node[stepstate] (middle) at (7.8,1.1) {$\tilde s=(O,G',\delta')$};
\node[stepstate] (finish) at (13.8,1.1) {$s'=(O',G',\delta')$};
\draw[draw=figInk,line width=1.1pt,-{Latex[length=2.3mm]},shorten <=3pt,shorten >=3pt]
  (start.east) -- node[above=7pt,text=black] {Gate moves}
  node[below=7pt,text=black] {\Hcond{0}--\Hcond{2}} (middle.west);
\draw[draw=figInk,line width=1.1pt,-{Latex[length=2.3mm]},shorten <=3pt,shorten >=3pt]
  (middle.east) -- node[above=7pt,text=black] {Disk moves}
  node[below=7pt,text=black] {\Hcond{0}, \Hcond{4}} (finish.west);
\end{tikzpicture}
\caption{A step is checked in two stages. In the first stage, the disk stays fixed while the gate moves and the difference is updated; conditions \Hcond{0}--\Hcond{2} give the bound in \cref{lem:gate-move}. In the second stage, the gate and the difference stay fixed while the disk moves along the pencil; conditions \Hcond{0} and \Hcond{4} ensure that the potential does not increase, as stated in \cref{cor:pencil-move}.}
\label{fig:step-stages}
\end{figure}

The next lemma handles the two ends of a chain, where \Hcond{1} and \Hcond{3} pay for the initial pair and for the final step.

\begin{lemma}[Initial pairs and final steps]\label{lem:ends}
Let \(U\in\R\), and let \(H\) satisfy \Hcond{1} and \Hcond{3} at level \(U\).
\begin{enumerate}[label=(\arabic*)]
\item\label{ends:initial} For every initial pair \((O_1,O_2)\), with \(G_1=[a_1,b_1]\) the common chord of \(\partial O_1\) and \(\partial O_2\), \(c_1\) its crossing with the \(x\)-axis, \(u=w_-(O_1)\), and \(\delta_1\) and \(\mu_1\) the difference and the cost of the initial pair (\cref{def:moves}\ref{mv:initial}), the triple \(\tilde s_1=(O_1,G_1,\delta_1)\) is a gate state, and
\begin{equation}\label{eq:ends-initial}
\mu_1-U\bigl(x(c_1)-x(u)\bigr)+W_{H,U}(\tilde s_1)\le0 .
\end{equation}
\item\label{ends:final} For every gate state \(s=(O,G,\delta)\), the cost \(\kappa\) of the final step at \(s\) satisfies
\begin{equation}\label{eq:ends-final}
\kappa-U\bigl(x(w_+(O))-x(c(s))\bigr)-W_{H,U}(s)\le0 .
\end{equation}
\end{enumerate}
\end{lemma}

\begin{proof}
\ref{ends:initial} The triple \(\tilde s_1\) is a gate state, since \(G_1\) is a chord of \(\partial O_1\) with \(y(a_1)>0>y(b_1)\) and \(\abs{\delta_1}\le\abs{G_1}\).  Let \(o_1\) and \(r\) be the center and the radius of \(O_1\), and let \((\gamma,\varphi,\theta,\psi_1)\) be the angular coordinates of \(\tilde s_1\).  Since \(u=w_-(O_1)\), \cref{lem:range}\ref{range:arcs} gives \(\alpha_1=r(\theta+\gamma-\varphi)\), \(\beta_1=r(\theta-\gamma+\varphi)\) and \(x(u)=x(o_1)-r\cos\gamma\).  So \(\alpha_1-\beta_1=2r(\gamma-\varphi)\), and since \(\delta_1=\clip_{\abs{G_1}}(\alpha_1-\beta_1)\), we get \(\abs{\delta_1-2r(\gamma-\varphi)}=\omega_1\), where \(\omega_1=(\abs{\alpha_1-\beta_1}-\abs{G_1})_+\) is the amount \(\omega\) by which the clipping lowers the cost of the initial pair (\cref{def:moves}\ref{mv:initial}); that is, \(\abs{\psi_1-\gamma}=\omega_1/(2r)\).  Using \(\mu_1=\frac12(\alpha_1+\beta_1-\omega_1)=r\theta-\frac{\omega_1}2\) and \cref{def:residual}, we obtain
\[
\mu_1-U\bigl(x(c_1)-x(u)\bigr)+W_{H,U}(\tilde s_1)=r\Bigl[\frac\pi2-U\cos\gamma-H(\gamma,\varphi,\theta,\psi_1)\Bigr]-\frac{\omega_1}2 .
\]
By \Hcond{1}, \(H(\gamma,\varphi,\theta,\psi_1)\ge H(\gamma,\varphi,\theta,\gamma)-\omega_1/(2r)\), and by \Hcond{3}, \(H(\gamma,\varphi,\theta,\gamma)\ge\frac\pi2-U\cos\gamma\).  So the right-hand side is at most \(0\).

\ref{ends:final} Let \(o\) and \(r\) be the center and the radius of \(O\), and let \((\gamma,\varphi,\theta,\psi)\) be the angular coordinates of \(s\).  By \cref{lem:range}\ref{range:arcs}, the lengths in \cref{def:moves}\ref{mv:final} are \(\alpha_{\rm end}=r(\pi+2\gamma)-r(\theta+\gamma-\varphi)\) and \(\beta_{\rm end}=r(\pi-2\gamma)-r(\theta-\gamma+\varphi)\).  Since \(\delta=2r(\psi-\varphi)\), we get \(\frac12\delta+\alpha_{\rm end}=r(\pi-\theta+\psi+\gamma)\) and \(-\frac12\delta+\beta_{\rm end}=r(\pi-\theta-\psi-\gamma)\), so \(\kappa=r(\pi-\theta-\abs{\psi+\gamma})\).  Using \(x(w_+(O))=x(o)+r\cos\gamma\) and \cref{def:residual}, we obtain
\[
\kappa-U\bigl(x(w_+(O))-x(c(s))\bigr)-W_{H,U}(s)=r\Bigl[\frac\pi2-U\cos\gamma+H(\gamma,\varphi,\theta,\psi)-\abs{\psi+\gamma}\Bigr].
\]
By \Hcond{1}, \(H(\gamma,\varphi,\theta,\psi)\le H(\gamma,\varphi,\theta,-\gamma)+\abs{\psi+\gamma}\), and by \Hcond{3}, \(H(\gamma,\varphi,\theta,-\gamma)\le-\frac\pi2+U\cos\gamma\).  So the right-hand side is at most \(0\).
\end{proof}

\begin{proof}[Proof of \cref{thm:sufficient}]
Let \(W=W_{H,U}\).  We verify the three inequalities of \cref{def:feasible}.

(S) Let \(s=(O,G,\delta)\) be a gate state and \(O'\) a next disk for \(s\), and let the step lead to \(s'=(O',G',\delta')\).  By \cref{lem:gate-move}, \(\tilde s=(O,G',\delta')\) is a gate state and \(\Delta\mu-U\,\Delta x\le W(s)-W(\tilde s)\).  The circles \(\partial O\) and \(\partial O'\) both pass through the ends of \(G'\), so by \cref{lem:pencil}\ref{pencil:family}, applied to the segment \(G'\), we have \(O=O(X)\) and \(O'=O(X')\) with \(X'>X\), since \(x(o')>x(o)\).  By \cref{cor:pencil-move}, applied to \(G'\) and \(\delta'\), \(W(s')=W(s(X'))\le W(s(X))=W(\tilde s)\).  Hence \(\Delta\mu-U\,\Delta x\le W(s)-W(s')\).

(I) Let \((O_1,O_2)\) be an initial pair leading to \(s_1=(O_2,G_1,\delta_1)\).  By \cref{lem:ends}\ref{ends:initial}, \(\mu_1-U(x(c_1)-x(u))\le-W(\tilde s_1)\), where \(\tilde s_1=(O_1,G_1,\delta_1)\).  As in the proof of (S), \cref{lem:pencil}\ref{pencil:family} and \cref{cor:pencil-move}, applied to the segment \(G_1\), give \(W(s_1)\le W(\tilde s_1)\), since the center of \(O_2\) has a larger \(x\)-coordinate than the center of \(O_1\).  Hence \(\mu_1-U(x(c_1)-x(u))\le-W(s_1)\).

(F) This is \cref{lem:ends}\ref{ends:final}.

So \(W\) is a feasible potential of level \(U\), and since \(U\ge\pi/2\), \cref{thm:criterion} shows that the chain inequality holds at level \(U\).
\end{proof}

\Cref{thm:sufficient} completes the reduction from nine parameters to four.  The reductions of \cref{sec:accounting} and of this section hold for every level \(U\ge\pi/2\); from the next section on, we fix \(U=U_0=1.59485994\) and construct one explicit function \(H\) of the four angular coordinates for which all five conditions can be verified.  \Cref{sec:family} constructs \(\Hs\) and checks \Hcond{0}--\Hcond{3}.  \Cref{sec:verification} checks \Hcond{4}, the condition with no room to spare: in the limit where the gate lies on the \(x\)-axis, the pencil residual of \(\Hs\) tends to zero, so near this limit, interval arithmetic has no margin to work with.

\section{The Explicit Function \texorpdfstring{\(\Hs\)}{H*}}\label{sec:family}

\Cref{thm:sufficient} in \cref{sec:pencil} reduced the chain inequality at level \(U\) to conditions \Hcond{0}--\Hcond{4} on one function of the four angular coordinates: a regularity condition, a Lipschitz condition in the difference, monotonicity in the position of the gate, bounds at the ends of a chain, and the condition produced by the pencil lemma.  We now fix \(U=U_0=1.59485994\).  This value comes from numerical optimization: among all levels \(U\) for which a potential with rational coefficients in the family of \cref{def:family} could be constructed and verified, \(U_0\) is the smallest that our search produced.  The reductions of \cref{sec:accounting,sec:pencil} hold for every \(U\ge\pi/2\); the choice of \(U_0\) enters only through the verification in this section and \cref{sec:verification}.  As the conditions involve only four real parameters, a computer-assisted check is feasible whenever the conditions hold with room to spare.  In this section, we construct an explicit function \(\Hs\) in a family designed for this, and prove that it satisfies \Hcond{0}--\Hcond{3} at level \(U_0\), by checks in at most three variables.  Condition \Hcond{4} is verified in \cref{sec:verification}: it holds with equality in the limit where the gate lies on the \(x\)-axis, so near this limit it has no room to spare.  Together, the two sections prove the following theorem, the last row of \cref{tab:proofmap}.

\begin{theorem}[The explicit function \(\Hs\)]\label{thm:explicit}
Let \(U_0=1.59485994\).  There exists a function \(\Hs\) in the parametrized family \(\cH\) of \cref{def:family} that satisfies conditions \Hcond{0}--\Hcond{4} at level \(U_0\).  The members of \(\cH\) are piecewise smooth functions of the four angular coordinates \((\gamma,\varphi,\theta,\psi)\) with rational coefficients.  The function \(\Hs\) is specified by \(3{,}150\) rational coefficients, listed in the data file of \cref{subsec:data}; the remaining coefficients are determined by the symmetry~\eqref{eq:sym}.  The conditions are verified with respect to the piecewise \(C^1\) structure of \cref{def:pieces}, which covers the domain by finitely many closed pieces on each of which \(\Hs\) agrees with a continuously differentiable function, as \cref{def:piecewise} requires.
\end{theorem}

\Cref{thm:explicit} provides the input of \cref{thm:sufficient} at level \(U_0\), and \cref{sec:proof-overview} derives \cref{thm:main} from it.  The proof occupies the rest of this section and \cref{sec:verification}: \cref{prop:h0h3} in \cref{subsec:h1h3} verifies \Hcond{0}--\Hcond{3}, and \cref{prop:h4} in \cref{sec:verification} verifies \Hcond{4}, which has no room to spare near the corner curve, by the corner lemmas and a cover by approximately \(1.04\) million boxes.  \Cref{subsec:proofs} assembles the two halves.  In words, \(\Hs\) is one member of a family of functions of the four angles that are polynomial in the position of the gate and in a clipped and rescaled version of \(\psi\), with coefficients that are cubic splines and trigonometric functions of \(\gamma\).  Condition \Hcond{0} holds for every member of \(\cH\) (\cref{lem:regularity}).  The other four conditions are verified with the help of a computer, in exact rational and interval arithmetic: \Hcond{1}--\Hcond{3} reduce to inequalities in at most three variables, and \Hcond{4} reduces to inequalities on finitely many boxes.

\subsection{The Family \texorpdfstring{\(\cH\)}{H}}\label{subsec:family}

In this subsection, we define the family \(\cH\) in which \(\Hs\) is found (\cref{def:family}), in coordinates adapted to the conditions, and we explain its design; we then verify condition \Hcond{0} for every member of the family (\cref{lem:regularity}) and use the two symmetries of the family to restrict the remaining conditions to \(\gamma\ge0\) (\cref{lem:symmetry}).  The design serves the checks of \cref{subsec:h1h3,sec:verification} in three ways.  The members of \(\cH\) are bounded on boxes by their coefficients, so a computer can bound them cheaply; they do not depend on \(\psi\) where \(\abs\psi\) is large, so conditions required for all real \(\psi\) cost nothing there; and they all have the same limit at the corner curve below, where the gate lies on the \(x\)-axis and \Hcond{4} has no room to spare, so that the pencil residual is known to vanish there exactly.

\paragraph{The Shape of the Family.}  We look for \(H\) in a family \(\cH\) of functions with finitely many rational coefficients.  We describe the members of \(\cH\) in the coordinates \((\gamma,a,b,\xi)\), where \(a\) and \(b\) are the positions of \(a_G\) and \(b_G\) on the upper and the lower arc of the circle, measured from \(w_-(O)\) as fractions of the lengths of these arcs, and \(\xi=\psi/k(\gamma)\) with \(k(\gamma)=\frac\pi2-\cos\gamma\), a normalization under which \(\abs\xi\le1\) at the two values \(\psi=\pm\gamma\) of \Hcond{3} and \(\xi\ge1\) where the gate flattens onto the \(x\)-axis (see the corner curve below).  For fixed \(\gamma\), the pair \((a,b)\) ranges over the open unit square.  A member of \(\cH\) is the sum of two parts: \(\cos^2\gamma\) times a polynomial in \((a,b)\), which does not depend on \(\psi\), and an odd polynomial of degree \(5\) in the clipped variable \(\hat\xi=\clip_1(\xi)\), whose coefficients are polynomials in \((a,b)\); the coefficients of all these polynomials are built from \(\gamma\), \(\cos^2\gamma\) and cubic splines in \(\tau=2\gamma/\pi\).  The Bernstein degrees in \((a,b)\) and the number of knots in \cref{def:family} are design choices for the numerical search; the proofs do not depend on them.  With assistance from ChatGPT (OpenAI), we designed this family and searched numerically for the coefficients of \(\Hs\).

\begin{definition}[The family \(\cH\)]\label{def:family}
Let \(N_0,\dots,N_{34}\) be the clamped uniform cubic B-splines~\cite{deboor2001} on \([0,1]\) with knots \((0,0,0,0,\frac1{32},\frac2{32},\dots,\frac{31}{32},1,1,1,1)\); they are \(35\) nonnegative functions that sum to \(1\), twice continuously differentiable on \([0,1]\), and each of them is a cubic polynomial between consecutive knots.  Let \(\cB^d_i(x)=\binom di x^i(1-x)^{d-i}\) be the Bernstein polynomials~\cite{farouki2012}.  Let \(V_{ijl}\) (\(0\le i,j\le11\)), \(R_{ijl}\) (\(0\le i,j\le5\)), \(B_l\) and \(D_l\), where \(0\le l\le34\), be rational numbers.  For \(\gamma\in[0,\frac\pi2]\) and \(a,b\in\R\), let \(\tau=2\gamma/\pi\) and
\[
V(\gamma,a,b)=\sum_{i,j,l}V_{ijl}\,\cB^{11}_i(a)\,\cB^{11}_j(b)\,N_l(\tau),\qquad
 R(\gamma,a,b)=\sum_{i,j,l}R_{ijl}\,\cB^{5}_i(a)\,\cB^{5}_j(b)\,N_l(\tau),
\]
\[
B(\gamma)=\sum_lB_lN_l(\tau),\qquad D(\gamma)=\sum_lD_lN_l(\tau),\qquad \sigma(\gamma)=\gamma-\cos^2\gamma\,V(\gamma,0,1),
\]
and let \(k(\gamma)=\frac\pi2-\cos\gamma\) for \(\gamma\in\R\).  For \((\gamma,\varphi,\theta,\psi)\in\Omega\times\R\), define \(a\) and \(b\) by
\begin{equation}\label{eq:ab}
\theta=\Bigl(\frac\pi2+\gamma\Bigr)a+\Bigl(\frac\pi2-\gamma\Bigr)b,\qquad
\varphi=\gamma-\Bigl(\frac\pi2+\gamma\Bigr)a+\Bigl(\frac\pi2-\gamma\Bigr)b,
\end{equation}
let \(\xi=\psi/k(\gamma)\) and \(\hat\xi=\clip_1(\xi)\).  For \(\gamma\ge0\), set
\begin{equation}\label{eq:H}
H(\gamma,\varphi,\theta,\psi)=\cos^2\gamma\,V(\gamma,a,b)+\sigma(\gamma)\,\hat\xi+\cos^2\gamma\Bigl[\bigl(B(\gamma)+R(\gamma,a,b)\bigr)\bigl(\hat\xi-\hat\xi^3\bigr)+D(\gamma)\bigl(\hat\xi^3-\hat\xi^5\bigr)\Bigr],
\end{equation}
and for \(\gamma<0\), set \(H(\gamma,\varphi,\theta,\psi)=H(-\gamma,-\varphi,\theta,-\psi)\).  The family \(\cH\) consists of the functions \(H\) obtained in this way from coefficients that satisfy, for all indices,
\begin{gather}
V_{11-i,11-j,l}=-V_{ijl},\qquad R_{5-i,5-j,l}=R_{ijl},\qquad R_{ijl}=0\ \text{ for }(i,j)\in\{0,5\}^2,\label{eq:sym}\\
V_{ij0}=V_{ji0},\qquad R_{ij0}=-R_{ji0},\qquad B_0=D_0=0.\label{eq:gamma-zero}
\end{gather}
\end{definition}

Conditions~\eqref{eq:sym} make \(H\) change sign under the reflection \(\rho_2\) of \cref{lem:reflections} below and make \(R\) vanish at the corners of the square of \((a,b)\) (\cref{lem:h3} uses this at the corner \((0,0)\)); conditions~\eqref{eq:gamma-zero} make \(H\) continuous at \(\gamma=0\) (\cref{lem:reflections}).

\paragraph{The Intuition of \cref{def:family}.}  The four coefficient functions \(V\), \(R\), \(B\) and \(D\) play different roles.  On the two regions \(\pm\xi\ge1\), the outer pieces of \cref{def:pieces} below, formula~\eqref{eq:H} reduces to \(H=\cos^2\gamma\,V\pm\sigma\): so \(\cos^2\gamma\,V(\gamma,a,b)\) is the mean of the two values of \(H\) where \(\psi\) no longer matters, and \(2\sigma(\gamma)\) is their difference.  For \(\abs\xi\le1\), the function \(H\) passes from one outer value to the other along an odd polynomial of degree \(5\) in \(\hat\xi\), written in the basis \(\hat\xi\), \(\hat\xi-\hat\xi^3\), \(\hat\xi^3-\hat\xi^5\): the last two vanish at \(\hat\xi=\pm1\), so the coefficient of \(\hat\xi\) is forced to be \(\sigma\) by the outer values, while \(B\) and \(D\) shape the transition and \(R\) lets its shape depend on the position \((a,b)\) of the gate.

\paragraph{Why \(\xi\) Is Clipped, and Why the Powers Are Odd.}  The clipping of \(\xi\) at \(\pm1\) makes \(H\) independent of \(\psi\) for \(\abs\psi\ge k(\gamma)\), so that \Hcond{1} and \Hcond{2}, which are required for all real \(\psi\), cost nothing there; a nonconstant polynomial in \(\psi\) could not be \(1\)-Lipschitz on \(\R\).  The odd powers of \(\hat\xi\), together with~\eqref{eq:sym}, make \(H\) change sign under the reflection \(\rho_2\) of \cref{lem:reflections}.

\paragraph{Why Bernstein Polynomials and B-Splines.}  The Bernstein polynomials in \((a,b)\) and the B-splines in \(\tau\) are nonnegative and sum to \(1\), so a function is bounded by its largest coefficient, and the sign of a derivative follows from the signs of its coefficients after subdivision.  The computer uses these facts, stated in \cref{lem:enclosures}\ref{enc:bernstein}, to bound the coefficient functions in the verification of \Hcond{1}--\Hcond{4}.  The splines also make every function involved a polynomial in \(\tau\) between two consecutive knots, which is why the boxes of \cref{subsec:cover} never straddle a knot.

\paragraph{The Coordinates \((\gamma,a,b)\).} The next lemma makes the coordinates explicit and shows that they describe \(\Omega\) completely; by \cref{lem:range}\ref{range:arcs}, the numbers \(a\) and \(b\) of~\eqref{eq:ab-explicit} are the positions of \(a_G\) and \(b_G\) described above.

\begin{lemma}[The coordinates \((\gamma,a,b)\)]\label{lem:ab}
For \((\gamma,\varphi,\theta)\in\Omega\), equations~\eqref{eq:ab} have the unique solution
\begin{equation}\label{eq:ab-explicit}
a=\frac{\theta-\varphi+\gamma}{\pi+2\gamma},\qquad b=\frac{\theta+\varphi-\gamma}{\pi-2\gamma},
\end{equation}
and the map \((\gamma,\varphi,\theta)\mapsto(\gamma,a,b)\) is a bijection from \(\Omega\) onto \((-\frac\pi2,\frac\pi2)\times(0,1)^2\).  Consequently, every function \(H\) of \cref{def:family} is well defined on \(\Omega\times\R\).
\end{lemma}

\begin{proof}
Every point of \(\Omega\) has \(\abs\gamma<\frac\pi2\), since \(\abs{\gamma-\varphi}+\abs{\gamma+\varphi}\ge2\abs\gamma\); so the denominators in~\eqref{eq:ab-explicit} are positive, and solving~\eqref{eq:ab} gives~\eqref{eq:ab-explicit}.  For \(\abs\gamma<\frac\pi2\), the inequalities \(0<a<1\) and \(0<b<1\) read \(\theta>\varphi-\gamma\), \(\theta<\pi+\gamma+\varphi\), \(\theta>\gamma-\varphi\) and \(\theta<\pi-\gamma-\varphi\), which together are the four inequalities that define \(\Omega\); since~\eqref{eq:ab} determines \((\varphi,\theta)\) from \((\gamma,a,b)\), the map is a bijection onto \((-\frac\pi2,\frac\pi2)\times(0,1)^2\).  Finally, \(k(\gamma)\ge\frac\pi2-1>0\), so \(\xi\) is defined, and \((-\gamma,-\varphi,\theta,-\psi)\in\Omega\times\R\) whenever \((\gamma,\varphi,\theta,\psi)\in\Omega\times\R\), so the definition of \(H\) for \(\gamma<0\) refers to a point with positive first coordinate.
\end{proof}

\paragraph{The Corner Curve and the Choice of \(\sigma\).}  The curve \(\{(\gamma,0,1):0\le\gamma<\frac\pi2\}\) in the coordinates \((\gamma,a,b)\) is the \emph{corner curve} (\cref{def:pieces} below); at its points, the gate degenerates to the chord of the circle on the \(x\)-axis (\cref{fig:corner}).  This is where \Hcond{4} has no room to spare.  There, the gate is horizontal, so the center of the disk moves vertically along the pencil of the gate, and the term \(-U\cos\varphi\) of the pencil residual, the only one that contains \(U\), vanishes: a larger level does not help.  The function \(\sigma\) is chosen so that, for \(\xi\ge1\), the value of \(H\) tends to \(\gamma\) as \((a,b)\) tends to \((0,1)\); the normalization by \(k(\gamma)\) is such that \(\xi\ge1\) at all limit points of \(\Sigma\) with \((a,b)=(0,1)\) (see the beginning of \cref{subsec:corner-identity}).  With the limit value \(\gamma\), \cref{lem:corner-identity} shows that the pencil residual vanishes exactly at these limit points, for every member of \(\cH\), and the corner lemmas of \cref{sec:verification} start from this exact zero.  The zero cannot be avoided by another choice of \(\sigma\): as \cref{rem:corner} explains, no other limit value is compatible with \Hcond{1} and \Hcond{4} together.

\paragraph{Why the Factor \(\cos^2\gamma\).}  As \(\gamma\to\frac\pi2\), the circle becomes tangent to the \(x\)-axis, and the lower arc, of length \(r(\pi-2\gamma)\), shrinks to a point; so the position \(b\) on this arc loses its meaning, and a term that still depended on \(b\) near \(\gamma=\frac\pi2\) would have unbounded derivatives with respect to \(\varphi\) and \(\theta\) (see~\eqref{eq:theta-phi} in \cref{subsec:h1h3}).  The factor \(\cos^2\gamma\) in front of \(V\), \(R\), \(B\) and \(D\) makes these terms tend to \(0\) as \(\gamma\to\frac\pi2\), together with their first derivatives with respect to \((\gamma,\varphi,\theta,\psi)\); so the functions of \cref{def:pieces} below, their gradients and their pencil residuals extend continuously to \(\gamma=\frac\pi2\), and the verification can reach this boundary, as \cref{rem:edge} explains.

\paragraph{Two Reflections.}  The members of \(\cH\) have two symmetries: the first comes from the definition for \(\gamma<0\) together with~\eqref{eq:gamma-zero}, and the second from~\eqref{eq:sym}.  Geometrically, \(\rho_1\) below reflects the disk and the gate in the \(x\)-axis, which exchanges the upper and the lower arc, and \(\rho_2\) reflects them in the vertical line through the center, which measures the positions \(a\) and \(b\) from \(w_+(O)\) instead of \(w_-(O)\); both change the sign of \(\delta\).  The first symmetry reduces \Hcond{1}--\Hcond{4} to the points with \(\gamma\ge0\) (\cref{lem:symmetry}), and the second is used in \cref{sec:verification}.

\begin{lemma}[Reflections]\label{lem:reflections}
Let \(H\in\cH\).  The maps
\begin{equation}\label{eq:reflections}
\rho_1(\gamma,\varphi,\theta,\psi)=(-\gamma,-\varphi,\theta,-\psi),\qquad
\rho_2(\gamma,\varphi,\theta,\psi)=(\gamma,-\varphi,\pi-\theta,-\psi)
\end{equation}
map \(\Omega\times\R\) and \(\Sigma\) onto themselves; in the coordinates \((\gamma,a,b,\xi)\), they read \((\gamma,a,b,\xi)\mapsto(-\gamma,b,a,-\xi)\) and \((\gamma,a,b,\xi)\mapsto(\gamma,1-a,1-b,-\xi)\).  Moreover, \(H\circ\rho_1=H\) and \(H\circ\rho_2=-H\).
\end{lemma}

\begin{proof}
The conditions \(\abs{\gamma-\varphi}<\theta<\pi-\abs{\gamma+\varphi}\) and \(\abs{\psi-\varphi}\le\sin\theta\) are invariant under both maps, since \(\sin(\pi-\theta)=\sin\theta\).  The formulas in the coordinates \((\gamma,a,b,\xi)\) follow from~\eqref{eq:ab-explicit} and from \(k(-\gamma)=k(\gamma)\).

For \(\gamma\ne0\), the identity \(H(\rho_1(z))=H(z)\) is the definition of \(H\) at whichever of the two points has negative \(\gamma\).  For \(\gamma=0\), we have \(N_0(0)=1\) and \(N_l(0)=0\) for \(l\ge1\), so~\eqref{eq:gamma-zero} gives \(V(0,b,a)=V(0,a,b)\), \(R(0,b,a)=-R(0,a,b)\) and \(B(0)=D(0)=0\).  Moreover, \(V(0,1,0)=-V(0,0,1)\) by~\eqref{eq:sym}, since \(\cB^d_i(1-x)=\cB^d_{d-i}(x)\); together with \(V(0,1,0)=V(0,0,1)\), this gives \(\sigma(0)=-V(0,0,1)=0\).  Substituting into~\eqref{eq:H}, we see that \(H\) takes the same value at \((0,a,b,\xi)\) and at \((0,b,a,-\xi)\).

For \(\gamma\ge0\), the identity \(\cB^d_i(1-x)=\cB^d_{d-i}(x)\) and~\eqref{eq:sym} give \(V(\gamma,1-a,1-b)=-V(\gamma,a,b)\) and \(R(\gamma,1-a,1-b)=R(\gamma,a,b)\).  Since \(\rho_2\) also replaces \(\hat\xi\) by \(-\hat\xi\), every term of~\eqref{eq:H} changes sign, so \(H(\rho_2(z))=-H(z)\).  For \(\gamma<0\), the maps \(\rho_1\) and \(\rho_2\) commute, and \(\rho_1(z)\) has positive first coordinate, so \(H(\rho_2(z))=H(\rho_1\rho_2(z))=H(\rho_2\rho_1(z))=-H(\rho_1(z))=-H(z)\).
\end{proof}

\paragraph{Pieces and Condition \Hcond{0}.}  Next, we define a piecewise \(C^1\) structure, in the sense of \cref{def:piecewise}, for the members of \(\cH\), and we show that they satisfy \Hcond{0}.  Roughly speaking, for \(\gamma\ge0\), the function \(H\) is given by one formula on each of the three regions \(\abs\xi\le1\), \(\xi\ge1\) and \(\xi\le-1\), which we call the pieces; the formula of a piece is continuously differentiable beyond the boundary of the piece, and the reflection \(\rho_1\) carries the three pieces to \(\gamma\le0\) (so that several gradients may be incident at a point where pieces meet, as \cref{def:piecewise} allows).

\begin{definition}[Pieces; the corner curve]\label{def:pieces}
Let \(H\in\cH\).  Extend each B-spline \(N_l\) to \(\tau<0\) by the cubic polynomial that it equals on \([0,\frac1{32}]\) (so that the functions below are continuously differentiable on open neighborhoods of the pieces, as \cref{def:piecewise} requires), define \(V\), \(R\), \(B\), \(D\) and \(\sigma\) for \(\gamma\in(-\frac\pi2,0)\) by the formulas of \cref{def:family}, and let \(\widehat H(\gamma,a,b,\xi)\), for \(\abs\gamma<\frac\pi2\) and \(a,b,\xi\in\R\), be the right-hand side of~\eqref{eq:H} with \(\hat\xi\) replaced by \(\xi\).  For \(\abs\gamma<\frac\pi2\) and \(\varphi,\theta,\psi\in\R\), with \(a\) and \(b\) given by~\eqref{eq:ab}, let
\begin{equation}\label{eq:piece-functions}
H^{\rm c}(\gamma,\varphi,\theta,\psi)=\widehat H\bigl(\gamma,a,b,\psi/k(\gamma)\bigr),\qquad
 H^{\pm}(\gamma,\varphi,\theta,\psi)=\widehat H(\gamma,a,b,\pm1)=\cos^2\gamma\,V(\gamma,a,b)\pm\sigma(\gamma).
\end{equation}
The \textbf{\emph{central piece}} is the set of points of \(\Omega\times\R\) with \(\gamma\ge0\) and \(\abs\xi\le1\), and the two \textbf{\emph{outer pieces}} are the sets of points of \(\Omega\times\R\) with \(\gamma\ge0\) and \(\pm\xi\ge1\).  We assign the function \(H^{\rm c}\) to the central piece, the function \(H^\pm\) to the outer piece \(\pm\xi\ge1\), and, to the image under \(\rho_1\) of each of these three pieces, the function of that piece composed with \(\rho_1\).  For \(z\in\Omega\times\R\) with \(\gamma\ge0\), the \emph{gradients of the pieces} at \(z\) are the gradients at \(z\) of the functions of the pieces with \(\gamma\ge0\) that contain \(z\).  The \textbf{\emph{corner curve}} is the curve \(\{(\gamma,0,1):0\le\gamma<\frac\pi2\}\) in the coordinates \((\gamma,a,b)\); at its points, the gate has degenerated to the chord of the circle on the \(x\)-axis.
\end{definition}

\Cref{fig:clipped-variable} shows the clipped coordinate and the three pieces for \(\gamma\ge0\).

\begin{figure}[!hb]
\centering
\begin{tikzpicture}[paperfig,font=\normalsize,text=black]
\path[use as bounding box] (0,-0.15) rectangle (10,4.55);
\draw[draw=figGray,line width=.8pt,dash pattern=on 4pt off 3pt]
  (3.4,1.1) -- (3.4,2.2)
  (6.6,2.2) -- (6.6,3.3);
\draw[draw=figGray,line width=.8pt,dash pattern=on 4pt off 3pt]
  (3.4,1.1) -- (5,1.1)
  (5,3.3) -- (6.6,3.3);
\draw[draw=black,line width=.85pt,-{Latex[length=2.1mm]}]
  (0.5,2.2) -- (9.5,2.2);
\draw[draw=black,line width=.85pt,-{Latex[length=2.1mm]}]
  (5,0.7) -- (5,4.0);
\draw[draw=figBlue,line width=1.6pt]
  (0.9,1.1) -- (3.4,1.1) -- (6.6,3.3) -- (9.1,3.3);
\node[anchor=west] at (9.62,2.2) {$\xi$};
\node[anchor=south] at (5,4.1) {$\hat\xi$};
\node[anchor=south] at (3.4,2.32) {$-1$};
\node[anchor=north] at (6.6,2.08) {$1$};
\node[anchor=west] at (5.16,1.1) {$-1$};
\node[anchor=east] at (4.84,3.3) {$1$};
\node[anchor=north west] at (5.08,2.08) {$0$};
\node[anchor=north] at (1.95,0.4) {Outer piece};
\node[anchor=north] at (5,0.4) {Central piece};
\node[anchor=north] at (8.05,0.4) {Outer piece};
\end{tikzpicture}
\caption{The clipped coordinate \(\hat\xi=\clip_1(\xi)\) equals \(\xi\) for \(-1\le\xi\le1\) and is fixed at \(-1\) or \(1\) outside this interval. For fixed \((\gamma,a,b)\), clipping makes every \(H\in\cH\) independent of \(\psi\) on each outer piece.}
\label{fig:clipped-variable}
\end{figure}

\begin{lemma}[Regularity]\label{lem:regularity}
Let \(H\in\cH\).  Then the six sets and functions of \cref{def:pieces} form a piecewise \(C^1\) structure for \(H\), and \(H\) is continuous and locally Lipschitz; in particular, \(H\) satisfies \Hcond{0}.  Moreover, the following hold, where \(\cE_U\) is given by the expression of \cref{def:residual} at all points of \((-\frac\pi2,\frac\pi2)\times\R^3\).
\begin{enumerate}[label=(\arabic*)]
\item\label{reg:rho2} \(H^{\rm c}\circ\rho_2=-H^{\rm c}\) and \(H^{\mp}\circ\rho_2=-H^{\pm}\) on \((-\frac\pi2,\frac\pi2)\times\R^3\); consequently, for every \(U\in\R\), we have \(\cE_U[H^{\rm c}]\circ\rho_2=\cE_U[H^{\rm c}]\) and \(\cE_U[H^-]\circ\rho_2=\cE_U[H^+]\) on this set.
\item\label{reg:residual} Let \(z\in\Omega\times\R\) lie in the image under \(\rho_1\) of a piece with function \(F\), and let \(\nabla\) be the gradient of \(F\circ\rho_1\) at \(z\).  Then \(\rho_1(z)\) is a point of this piece, \(\nabla'=(-\nabla_\gamma,-\nabla_\varphi,\nabla_\theta,-\nabla_\psi)\) is the gradient of \(F\) at \(\rho_1(z)\), and \(\cE_U[H;\nabla'](\rho_1(z))=\cE_U[H;\nabla](z)\) for every \(U\in\R\).
\end{enumerate}
\end{lemma}

Parts~\ref{reg:rho2} and~\ref{reg:residual} transport the residuals under \(\rho_2\) and \(\rho_1\); they serve \cref{lem:symmetry} and the identities~\eqref{eq:rho2-residuals} of \cref{subsec:corner-lemmas}.

\begin{proof}
The proof has three steps: the six sets and functions form a piecewise \(C^1\) structure; \(H\) is locally Lipschitz, which gives \Hcond{0}; and one identity,~\eqref{eq:reflect-residual} below, transports the residual under both reflections and gives~\ref{reg:rho2} and~\ref{reg:residual}.  The B-splines \(N_l\) are twice continuously differentiable on \([0,1]\) and cubic polynomials on \([0,\frac1{32}]\), so the extended B-splines are twice continuously differentiable on \((-\infty,1]\).  Hence \(\widehat H\) is twice continuously differentiable on \((-\frac\pi2,\frac\pi2)\times\R^3\), since it is a polynomial in \((a,b,\xi)\) whose coefficients are such functions of \(\gamma\).  As \(a\), \(b\) and \(\psi/k(\gamma)\) are smooth functions of \((\gamma,\varphi,\theta,\psi)\) for \(\abs\gamma<\frac\pi2\), the functions \(H^{\rm c}\) and \(H^\pm\) are continuously differentiable on the open set \((-\frac\pi2,\frac\pi2)\times\R^3\), which contains \(\Omega\times\R\) and which \(\rho_1\) maps onto itself; so are their compositions with \(\rho_1\).  For \(\gamma\ge0\), we have \(H=H^{\rm c}\) where \(\abs\xi\le1\), since \(\hat\xi=\xi\) there, and \(H=H^\pm\) where \(\pm\xi\ge1\), since \(\hat\xi=\pm1\) and \(\hat\xi-\hat\xi^3=\hat\xi^3-\hat\xi^5=0\) there.  Since \(H=H\circ\rho_1\) by \cref{lem:reflections} and \(\rho_1\) is an involution, \(H\) equals \(F\circ\rho_1\) on the image under \(\rho_1\) of a piece with function \(F\).  The six sets are closed in \(\Omega\times\R\), and they cover \(\Omega\times\R\), because the three pieces cover the points with \(\gamma\ge0\) and \(\rho_1\) maps the points with \(\gamma\le0\) to points with \(\gamma\ge0\).  So we have a piecewise \(C^1\) structure for \(H\).

For \(\gamma\ge0\), we have \(H(z)=\widehat H(\gamma,a,b,\clip_1(\xi))\), a continuously differentiable function composed with a locally Lipschitz map; so \(H\) is continuous and locally Lipschitz on \(\{\gamma\ge0\}\), and by \cref{lem:reflections}, also on \(\{\gamma\le0\}\).  If two points of a convex neighborhood lie on different sides of \(\{\gamma=0\}\), the segment between them crosses this set; applying the Lipschitz bounds on the two parts of the segment shows that \(H\) is locally Lipschitz on \(\Omega\times\R\).  This proves \Hcond{0}.

It remains to prove~\ref{reg:rho2} and~\ref{reg:residual}.  Let \(\epsilon_\rho=1\) for \(\rho=\rho_1\) and \(\epsilon_\rho=-1\) for \(\rho=\rho_2\).  We first show that, if \(F\) and \(F'\) are continuously differentiable functions on \((-\frac\pi2,\frac\pi2)\times\R^3\) with \(F'\circ\rho=\epsilon_\rho F\), then
\begin{equation}\label{eq:reflect-residual}
\cE_U[F'](\rho(z))=\cE_U[F](z)\qquad\text{for all }U\in\R\text{ and }z\in\bigl(-\tfrac\pi2,\tfrac\pi2\bigr)\times\R^3.
\end{equation}
The derivative of \(\rho\) is \(D\rho_1=\mathrm{diag}(-1,-1,1,-1)\) or \(D\rho_2=\mathrm{diag}(1,-1,-1,-1)\), and \((D\rho)^2\) is the identity; so the gradient \(\nabla'\) of \(F'\) at \(\rho(z)\) is \(\epsilon_\rho\,D\rho\,\nabla\), where \(\nabla\) is the gradient of \(F\) at \(z\).  Explicitly, \(\nabla'=(-\nabla_\gamma,-\nabla_\varphi,\nabla_\theta,-\nabla_\psi)\) for \(\rho_1\) and \(\nabla'=(-\nabla_\gamma,\nabla_\varphi,\nabla_\theta,\nabla_\psi)\) for \(\rho_2\).  Both maps change the signs of \(\cA\) and of \(\psi-\varphi\), leave \(\cos\varphi\), \(\sin\theta\) and \(\cos\theta\,(\frac\pi2-\theta)\) unchanged, and change the signs of \(\cos\theta\) and of the value of the function together (for \(\rho_2\)) or of neither (for \(\rho_1\)).  Substituting into \cref{def:residual}, we obtain~\eqref{eq:reflect-residual}.

\ref{reg:rho2} As in the proof of \cref{lem:reflections}, the identity \(\cB^d_i(1-x)=\cB^d_{d-i}(x)\) and~\eqref{eq:sym} give \(\widehat H(\gamma,1-a,1-b,-\xi)=-\widehat H(\gamma,a,b,\xi)\) for \(\abs\gamma<\frac\pi2\) and \(a,b,\xi\in\R\).  Since \(\rho_2\) acts by \((\gamma,a,b,\xi)\mapsto(\gamma,1-a,1-b,-\xi)\), this gives \(H^{\rm c}\circ\rho_2=-H^{\rm c}\) and, at \(\xi=\mp1\), \(H^{\mp}\circ\rho_2=-H^{\pm}\).  The identities for the residuals follow from~\eqref{eq:reflect-residual}.

\ref{reg:residual} Since \(\rho_1\) is an involution, \(\rho_1(z)\) lies in the piece.  Apply the formula for \(\nabla'\) above and~\eqref{eq:reflect-residual} with \(\rho=\rho_1\), with \(F\circ\rho_1\) in place of \(F\) and with \(F'=F\): the vector \((-\nabla_\gamma,-\nabla_\varphi,\nabla_\theta,-\nabla_\psi)\) is the gradient of \(F\) at \(\rho_1(z)\), and \(\cE_U[F](\rho_1(z))=\cE_U[F\circ\rho_1](z)\).  This is the claimed identity, since \(H=F\circ\rho_1\) at \(z\) and \(H=F\) at \(\rho_1(z)\).
\end{proof}

\paragraph{Reduction to \(\gamma\ge0\).}  The two reflections reduce conditions \Hcond{1}--\Hcond{4} to the points with \(\gamma\ge0\) and to the gradients of the pieces at these points.  From now on, for \(H\in\cH\), incident gradients, and hence conditions \Hcond{2} and \Hcond{4}, are taken with respect to the piecewise \(C^1\) structure of \cref{def:pieces}.

\begin{lemma}[Reduction to \(\gamma\ge0\)]\label{lem:symmetry}
Let \(H\in\cH\) and \(U\in\R\).
\begin{enumerate}[label=(\arabic*)]
\item\label{sym:h1} If \(\abs{H(\gamma,\varphi,\theta,\psi_1)-H(\gamma,\varphi,\theta,\psi_2)}\le\abs{\psi_1-\psi_2}\) for all \((\gamma,\varphi,\theta)\in\Omega\) with \(\gamma\ge0\) and all \(\psi_1,\psi_2\in\R\), then \(H\) satisfies \Hcond{1}.
\item\label{sym:h2} If \(\nabla_\theta\ge\abs{\nabla_\varphi}\) for every \(z\in\Omega\times\R\) with \(\gamma\ge0\) and every gradient \(\nabla\) of the pieces at \(z\), then \(H\) satisfies \Hcond{2}.
\item\label{sym:h3} If \(H(\gamma,\varphi,\theta,\gamma)\ge\frac\pi2-U\cos\gamma\) for all \((\gamma,\varphi,\theta)\in\Omega\) with \(\gamma\ge0\), then \(H\) satisfies \Hcond{3} at level \(U\).
\item\label{sym:h4} If \(\cE_U[H;\nabla](z)\le0\) for every \(z\in\Sigma\) with \(\gamma\ge0\) and every gradient \(\nabla\) of the pieces at \(z\), then \(H\) satisfies \Hcond{4} at level \(U\).
\end{enumerate}
\end{lemma}

\begin{proof}
\ref{sym:h1} For \(\gamma<0\), \cref{lem:reflections} gives \(\abs{H(\gamma,\varphi,\theta,\psi_1)-H(\gamma,\varphi,\theta,\psi_2)}=\abs{H(-\gamma,-\varphi,\theta,-\psi_1)-H(-\gamma,-\varphi,\theta,-\psi_2)}\le\abs{\psi_1-\psi_2}\), since \((-\gamma,-\varphi,\theta)\in\Omega\).

\ref{sym:h2} and~\ref{sym:h4}  Let \(z\in\Omega\times\R\), and let \(\nabla\) be an incident gradient of \(H\) at \(z\).  If \(\nabla\) is the gradient at \(z\) of the function of one of the three pieces with \(\gamma\ge0\), it is a gradient of the pieces at \(z\), and there is nothing to prove.  Otherwise, \(z\) lies in the image under \(\rho_1\) of a piece with function \(F\), and \(\nabla\) is the gradient of \(F\circ\rho_1\) at \(z\).  By \cref{lem:regularity}\ref{reg:residual}, \(\rho_1(z)\) is a point with \(\gamma\ge0\) of this piece, \(\nabla'=(-\nabla_\gamma,-\nabla_\varphi,\nabla_\theta,-\nabla_\psi)\) is the gradient of \(F\) at \(\rho_1(z)\), and \(\cE_U[H;\nabla'](\rho_1(z))=\cE_U[H;\nabla](z)\); so \(\nabla'\) is a gradient of the pieces at \(\rho_1(z)\).  Since \(\nabla_\theta=\nabla'_\theta\) and \(\abs{\nabla_\varphi}=\abs{\nabla'_\varphi}\), and since \(\rho_1(z)\in\Sigma\) whenever \(z\in\Sigma\) (\cref{lem:reflections}), the hypotheses at \(\rho_1(z)\) give the conclusions at \(z\).

\ref{sym:h3} For \(\gamma<0\), \cref{lem:reflections} gives \(H(\gamma,\varphi,\theta,\gamma)=H(-\gamma,-\varphi,\theta,-\gamma)\ge\frac\pi2-U\cos\gamma\), by the hypothesis at \((-\gamma,-\varphi,\theta)\).  For the second inequality of \Hcond{3}, \cref{lem:reflections} gives \(H(\gamma,\varphi,\theta,-\gamma)=-H(\gamma,-\varphi,\pi-\theta,\gamma)\le-\frac\pi2+U\cos\gamma\), by the first inequality at \((\gamma,-\varphi,\pi-\theta)\in\Omega\).
\end{proof}

\subsection{Conditions \Hcond{1}--\Hcond{3}}\label{subsec:h1h3}

With the family in place, we turn to conditions \Hcond{1}--\Hcond{3}; condition \Hcond{4}, which has no room to spare near the corner curve, is treated separately in \cref{sec:verification}.  Conditions \Hcond{1}--\Hcond{3} quantify over all of \(\Omega\times\R\), a four-dimensional set that is unbounded in \(\psi\), and \Hcond{2} over all incident gradients; a computer can only check inequalities on bounded boxes in few variables.  The next three lemmas remove \(\psi\), and for \Hcond{3} also \((a,b)\).  Since \(H\) is a polynomial in the clipped variable \(\hat\xi\), the Lipschitz condition \Hcond{1} becomes a bound on the coefficients of its derivative with respect to \(\psi\), and, given a bound on \(\abs R\), inequalities in \(\gamma\) alone (\cref{lem:h1}); the monotonicity condition \Hcond{2} becomes a sign condition on derivatives of the coefficient functions in \(a\) and \(b\) (\cref{lem:h2}); and since, by \Hcond{2}, the worst case of \Hcond{3} is the corner \((a,b)=(0,0)\), \Hcond{3} becomes one inequality in \(\gamma\) (\cref{lem:h3}).

\begin{lemma}[Reduction of \Hcond{1}]\label{lem:h1}
Let \(H\in\cH\).  For \(\gamma\in[0,\frac\pi2)\) and \((a,b)\in[0,1]^2\), let \(m_1=\cos^2\gamma\,(B(\gamma)+R(\gamma,a,b))\) and \(m_2=\cos^2\gamma\,D(\gamma)\), where \(B\), \(R\), \(D\), \(\sigma\) and \(k\) are the functions of \cref{def:family}; by~\eqref{eq:H}, \(m_1\) and \(m_2\) are the coefficients of \(\hat\xi-\hat\xi^3\) and \(\hat\xi^3-\hat\xi^5\) in \(H\).  Suppose that
\begin{equation}\label{eq:h1-coeffs}
\abs{\sigma(\gamma)+m_1}\le k(\gamma),\qquad
\Bigl\lvert\sigma(\gamma)-\tfrac12m_1+\tfrac32m_2\Bigr\rvert\le k(\gamma),\qquad
\abs{\sigma(\gamma)-2m_1-2m_2}\le k(\gamma)
\end{equation}
for all \(\gamma\in[0,\frac\pi2)\) and \((a,b)\in[0,1]^2\).  Then \(H\) satisfies \Hcond{1}.  In particular, if \(\abs R\le\bar R\) on \([0,\frac\pi2]\times[0,1]^2\) for some \(\bar R\ge0\), then it suffices that~\eqref{eq:h1-coeffs} holds for all \(\gamma\in[0,\frac\pi2)\) with \(R(\gamma,a,b)\) replaced by \(\bar R\) and by \(-\bar R\).
\end{lemma}

The three numbers inside the absolute values in~\eqref{eq:h1-coeffs} are the coefficients of \(k(\gamma)\,\partial_\psi H^{\rm c}\), a polynomial of degree \(2\) in \(\xi^2\), with respect to the Bernstein basis; so~\eqref{eq:h1-coeffs} bounds \(\abs{\partial_\psi H^{\rm c}}\) by \(1\).

\begin{proof}
By \cref{lem:symmetry}\ref{sym:h1}, it suffices to consider \(\gamma\ge0\).  Fix \((\gamma,\varphi,\theta)\in\Omega\) with \(\gamma\ge0\), and let \(f(\psi)=H(\gamma,\varphi,\theta,\psi)\).  The function \(f\) is continuous and constant on \(\{\xi\ge1\}\) and on \(\{\xi\le-1\}\).  For \(\abs\xi\le1\), we have \(k(\gamma)f'(\psi)=\partial_\xi\widehat H=\sigma+m_1(1-3\xi^2)+m_2(3\xi^2-5\xi^4)\).  With \(y=\xi^2\in[0,1]\), this equals
\[
(\sigma+m_1)(1-y)^2+2\Bigl(\sigma-\tfrac12m_1+\tfrac32m_2\Bigr)y(1-y)+(\sigma-2m_1-2m_2)y^2,
\]
and since the weights \((1-y)^2\), \(2y(1-y)\), \(y^2\) are nonnegative with sum \(1\),~\eqref{eq:h1-coeffs} gives \(\abs{f'}\le1\).  Hence \(f\) is \(1\)-Lipschitz.  The last claim holds because each expression inside the absolute values in~\eqref{eq:h1-coeffs} is an affine function of \(R(\gamma,a,b)\), so its absolute value is a convex function of \(R(\gamma,a,b)\) and attains its maximum over \([-\bar R,\bar R]\) at \(\bar R\) or at \(-\bar R\).
\end{proof}

\begin{lemma}[Reduction of \Hcond{2}]\label{lem:h2}
Let \(H\in\cH\).  Suppose that, on \([0,\frac\pi2]\times[0,1]^2\),
\[
\partial_a\Bigl(V\pm\tfrac25R\Bigr)\ge0\qquad\text{and}\qquad\partial_b\Bigl(V\pm\tfrac25R\Bigr)\ge0\qquad\text{for both signs.}
\]
Then \(H\) satisfies \Hcond{2}.
\end{lemma}

The number \(\frac25\) bounds \(\abs{\xi-\xi^3}\) for \(\abs\xi\le1\).

\begin{proof}
By \cref{lem:symmetry}\ref{sym:h2}, it suffices to consider the gradients of the pieces at points with \(\gamma\ge0\).  By~\eqref{eq:ab-explicit}, for a function of \((\gamma,a,b,\psi)\), we have
\begin{equation}\label{eq:theta-phi}
\partial_\theta-\partial_\varphi=\frac{1}{\frac\pi2+\gamma}\,\partial_a,\qquad
\partial_\theta+\partial_\varphi=\frac{1}{\frac\pi2-\gamma}\,\partial_b .
\end{equation}
Each of \(H^{\rm c}\) and \(H^\pm\) is \(\cos^2\gamma\,(V+\chi R)\) plus a function of \(\gamma\) and \(\psi\) only, where \(\chi=\xi-\xi^3\) for \(H^{\rm c}\) and \(\chi=0\) for \(H^\pm\).  So the condition \(\nabla_\theta\ge\abs{\nabla_\varphi}\), that is, \(\nabla_\theta-\nabla_\varphi\ge0\) and \(\nabla_\theta+\nabla_\varphi\ge0\), follows from \(\partial_a(V+\chi R)\ge0\) and \(\partial_b(V+\chi R)\ge0\).  On the central piece, \(\abs\chi\le\max_{0\le\xi\le1}(\xi-\xi^3)=2/(3\sqrt3)<\frac25\), and on the outer pieces, \(\chi=0\).  Since \(\partial_a(V+\chi R)\) and \(\partial_b(V+\chi R)\) are affine in \(\chi\), they are nonnegative for \(\abs\chi\le\frac25\) by the hypothesis.
\end{proof}

\begin{lemma}[Reduction of \Hcond{3}]\label{lem:h3}
Let \(U\in\R\), and let \(H\in\cH\) satisfy \Hcond{2}.  Suppose that, for every \(\gamma\in[0,\frac\pi2)\), at \(\xi=\gamma/k(\gamma)\),
\begin{equation}\label{eq:h3-1d}
\cos^2\gamma\,V(\gamma,0,0)+\sigma(\gamma)\,\xi+\cos^2\gamma\Bigl[B(\gamma)\bigl(\xi-\xi^3\bigr)+D(\gamma)\bigl(\xi^3-\xi^5\bigr)\Bigr]\ge\frac\pi2-U\cos\gamma .
\end{equation}
Then \(H\) satisfies \Hcond{3} at level \(U\).
\end{lemma}

Since \(R(\gamma,0,0)=\sum_lR_{00l}N_l(\tau)=0\) by~\eqref{eq:sym}, the left-hand side of~\eqref{eq:h3-1d} is \(\widehat H(\gamma,0,0,\gamma/k(\gamma))\), the value at \((a,b)=(0,0)\) and \(\psi=\gamma\) of the function \(\widehat H\) of \cref{def:pieces}.  So~\eqref{eq:h3-1d} is the first inequality of \Hcond{3}, extended to the corner \((a,b)=(0,0)\) of the square of \((a,b)\); by \Hcond{2}, this corner is the worst case.

\begin{proof}
By \cref{lem:symmetry}\ref{sym:h3}, it suffices to prove \(H(\gamma,\varphi,\theta,\gamma)\ge\frac\pi2-U\cos\gamma\) for \((\gamma,\varphi,\theta)\in\Omega\) with \(\gamma\ge0\).  The function \(\gamma+\cos\gamma\) is nondecreasing and equals \(\frac\pi2\) at \(\gamma=\frac\pi2\), so \(\xi=\gamma/k(\gamma)\in[0,1]\), and \(H=H^{\rm c}\) at these points.  For fixed \(\gamma\), the function \(p(a,b)=\widehat H(\gamma,a,b,\gamma/k(\gamma))\) is a polynomial in \((a,b)\), and by~\eqref{eq:theta-phi} and \Hcond{2}, \(\partial_ap=(\frac\pi2+\gamma)(\nabla_\theta-\nabla_\varphi)\ge0\) and \(\partial_bp=(\frac\pi2-\gamma)(\nabla_\theta+\nabla_\varphi)\ge0\) on \((0,1)^2\), where \(\nabla=\nabla H^{\rm c}\).  By continuity, \(p(a,b)\ge p(0,0)\) for all \((a,b)\in(0,1)^2\), and \(p(0,0)\) is the left-hand side of~\eqref{eq:h3-1d}, since \(R(\gamma,0,0)=0\) by~\eqref{eq:sym}.
\end{proof}

We now apply these lemmas to \(\Hs\).

\begin{proposition}[\Hcond{0}--\Hcond{3} for \(\Hs\)]\label{prop:h0h3}
The function \(\Hs\) of \cref{thm:explicit} belongs to \(\cH\) and satisfies \Hcond{0}--\Hcond{3} at level \(U_0\).  More precisely, the following four inequalities hold, and they imply \Hcond{1}--\Hcond{3} by \cref{lem:h1,lem:h2,lem:h3}:
\begin{enumerate}[label=(\alph*)]
\item\label{h03:R} \(\abs R\le\frac1{50}\) on \([0,\frac\pi2]\times[0,1]^2\);
\item\label{h03:h1} the inequalities~\eqref{eq:h1-coeffs} with \(R(\gamma,a,b)\) replaced by \(\frac1{50}\) and by \(-\frac1{50}\), for all \(\gamma\in[0,\frac\pi2)\);
\item\label{h03:h2} the hypothesis of \cref{lem:h2};
\item\label{h03:h3} the inequality~\eqref{eq:h3-1d} with \(U=U_0\), for all \(\gamma\in[0,\frac\pi2)\).
\end{enumerate}
\end{proposition}

\begin{proof}
The coefficients of \(\Hs\) satisfy~\eqref{eq:sym} and~\eqref{eq:gamma-zero}, as checked in exact rational arithmetic in \cref{subsec:cert-h0h3}.  So \(\Hs\in\cH\), and \Hcond{0} holds by \cref{lem:regularity}.  Inequalities~\ref{h03:R}--\ref{h03:h3} are verified by exact rational arithmetic or by interval arithmetic with outward rounding, \ref{h03:h1} and~\ref{h03:h3} in the equivalent forms of \cref{subsec:cert-h0h3}.  By \cref{lem:h1}, \ref{h03:R} and~\ref{h03:h1} give \Hcond{1}.  By \cref{lem:h2}, \ref{h03:h2} gives \Hcond{2}.  By \cref{lem:h3}, \Hcond{2} and~\ref{h03:h3} give \Hcond{3} at level \(U_0\).
\end{proof}

Some of these checks also have no room to spare, but only at the boundary \(\gamma=\frac\pi2\), where the circle becomes tangent to the \(x\)-axis: the upper bounds in~\eqref{eq:h1-coeffs} and the inequality~\eqref{eq:h3-1d} become equalities there.  These equalities are harmless, because the vanishing comes from a known factor: \cref{subsec:cert-h0h3} divides the differences of their two sides by \(\cos^2\gamma\) and by \(\cos\gamma\), respectively, and verifies the quotients by interval arithmetic up to \(\gamma=\frac\pi2\).  Condition \Hcond{4} is different: its pencil residual vanishes along the whole corner curve, for every member of \(\cH\) and at every level, by \cref{lem:corner-identity}.  

\paragraph{What Remains for \cref{thm:explicit}.}
\Cref{prop:h0h3} establishes conditions \Hcond{0}--\Hcond{3} for \(\Hs\).  The remaining condition, \Hcond{4} (that the pencil residual is nonpositive on \(\Sigma\)), is the most delicate: it holds with equality along the corner curve (\cref{def:pieces}), so interval arithmetic alone cannot certify it near this curve.  \Cref{sec:verification} introduces two \emph{corner lemmas} that replace the vanishing-margin condition near the corner curve with bounds on first derivatives that hold with a positive margin.  With their help, \cref{prop:h4} verifies \Hcond{4} by covering the domain with approximately \(1.04\) million boxes, and the end of \cref{sec:verification} completes the proof of \cref{thm:explicit}.

\section{Verifying \Hcond{4}: The Corner Lemmas and a Cover by Boxes}\label{sec:verification}

In this section, we verify \Hcond{4} for \(\Hs\) at level \(U_0\) and complete the proof of \cref{thm:explicit}.  Condition \Hcond{4} is the condition produced by the pencil lemma, that the potential does not increase while the disk slides along the pencil of its gate; it is the only one left open by \cref{sec:family}, and the one with no room to spare.  Recall from \cref{sec:intro} that interval arithmetic confirms an inequality on a box only if its overestimate on the box is smaller than the margin by which the inequality holds, so that near a point where the margin is zero, no finite number of boxes suffices.  For \Hcond{4}, this happens along the corner curve of \cref{subsec:family} (\cref{fig:corner}), where the gate becomes the chord on the \(x\)-axis: there, the level \(U\) drops out of the pencil residual, and the residual vanishes exactly, for every member of \(\cH\) (\cref{lem:corner-identity}); this zero cannot be avoided by another choice of \(\sigma\) (\cref{rem:corner}).  So the verification has an exact part and a numerical part.  Near the curve, the corner lemmas (\cref{lem:corner,lem:corner-outer}) start from the exact zero and reduce \Hcond{4} to bounds on first derivatives, which hold with room to spare.  A cover by boxes then completes the verification (\cref{prop:h4}): on each box, interval arithmetic checks either these bounds or \Hcond{4} itself.  The same happens along one edge of the square of \((a,b)\) as \(\gamma\to\frac\pi2\), where the corner lemmas apply as well (\cref{rem:edge}).

\subsection{Reduction of \Hcond{4}}\label{subsec:h4-reduction}

In this subsection, we write \Hcond{4} as inequalities that a computer can check on closed boxes (\cref{lem:h4}).  As it stands, \Hcond{4} is not suited to this: it speaks of incident gradients rather than of explicit functions, and \(\Sigma\) is not closed, so a closed box of parameters contains limit points of \(\Sigma\), at which \Hcond{4} says nothing.  We therefore write the residual of each piece as one smooth function and require it to be nonpositive on the closure of \(\Sigma\).  The closure adds the points above the corner curve, which are not gate states; it is there that \cref{lem:corner-identity} supplies the exact zero on which the corner lemmas rest.

Precisely, throughout this section we write the points of \(\Omega\times\R\) with \(\gamma\ge0\) in the coordinates \((\gamma,a,b,\xi)\) of \cref{def:family}.  Let \(\overline\Sigma\) denote the closure in \(\R^4\), in these coordinates, of the set of points of \(\Sigma\) with \(\gamma\ge0\); the points above the corner curve are its points with \((a,b)=(0,1)\).  We extend the pieces to \(\overline\Sigma\): a point \((\gamma,a,b,\xi)\) of \(\overline\Sigma\) with \(0\le\gamma<\frac\pi2\) lies in the central piece if \(\abs\xi\le1\), and in the outer piece \(\pm\xi\ge1\) if \(\pm\xi\ge1\); for the points of \(\Sigma\), this agrees with \cref{def:pieces}.

Let \(H\in\cH\) and \(U\in\R\), and let \(H^{\rm c}\) and \(H^\pm\) be the functions of \cref{def:pieces}.  For \(\gamma\in[0,\frac\pi2)\) and \((a,b,\xi)\in\R^3\), we define the residuals of the three pieces by
\begin{equation}\label{eq:piece-residuals}
\cE^{\rm c}_U(\gamma,a,b,\xi)=\cE_U[H^{\rm c}](\gamma,\varphi,\theta,\psi),\qquad
\cE^{\pm}_U(\gamma,a,b)=\cE_U[H^{\pm}](\gamma,\varphi,\theta,\psi),
\end{equation}
where \(\varphi\) and \(\theta\) are given by~\eqref{eq:ab} and \(\psi=k(\gamma)\,\xi\).  For fixed \(\gamma\), the residuals are smooth in \((a,b,\xi)\) and in \((a,b)\), respectively, because \(H^{\rm c}\) and \(H^\pm\) are polynomials in \((a,b,\xi)\) whose coefficients are functions of \(\gamma\), and \(\varphi\) and \(\theta\) are affine in \((a,b)\).

\begin{lemma}[Reduction of \Hcond{4}]\label{lem:h4}
Let \(H\in\cH\) and \(U\in\R\).
\begin{enumerate}[label=(\arabic*)]
\item\label{h4:closure} The points of \(\overline\Sigma\) with \(0\le\gamma<\frac\pi2\) are exactly the points \((\gamma,a,b,\xi)\) with \(0\le\gamma<\frac\pi2\), \((a,b)\in[0,1]^2\) and
\begin{equation}\label{eq:sigma-bar}
\abs{k(\gamma)\,\xi-\varphi}\le\sin\theta,
\end{equation}
where \(\varphi\) and \(\theta\) are given by~\eqref{eq:ab}.
\item\label{h4:pieces} If the residuals~\eqref{eq:piece-residuals} satisfy, at every point \((\gamma,a,b,\xi)\) of \(\overline\Sigma\) with \(0\le\gamma<\frac\pi2\),
\begin{equation}\label{eq:h4-pieces}
\cE^{\rm c}_U(\gamma,a,b,\xi)\le0\ \text{ if }\abs\xi\le1,\qquad
\cE^{+}_U(\gamma,a,b)\le0\ \text{ if }\xi\ge1,\qquad
\cE^{-}_U(\gamma,a,b)\le0\ \text{ if }\xi\le-1,
\end{equation}
then \(H\) satisfies \Hcond{4} at level \(U\).
\end{enumerate}
\end{lemma}

\begin{proof}
\ref{h4:closure} For \(\abs\gamma<\frac\pi2\), the coordinates \((\gamma,a,b,\xi)\) and \((\gamma,\varphi,\theta,\psi)\) are continuous functions of each other, and by \cref{lem:ab}, a point lies in \(\Sigma\) if and only if \((a,b)\in(0,1)^2\) and~\eqref{eq:sigma-bar} holds.  These conditions persist, in the stated closed form, under limits with \(0\le\gamma<\frac\pi2\), so every point of \(\overline\Sigma\) with \(0\le\gamma<\frac\pi2\) satisfies them.  Conversely, let \((\gamma,a,b,\xi)\) satisfy them, and let \(\psi=k(\gamma)\xi\).  Choose \((a_n,b_n)\in(0,1)^2\) converging to \((a,b)\), let \(\varphi_n\) and \(\theta_n\) be given by~\eqref{eq:ab} at \((\gamma,a_n,b_n)\), and let \(\psi_n\) be the point of the interval \([\varphi_n-\sin\theta_n,\varphi_n+\sin\theta_n]\) closest to \(\psi\).  Then \((\gamma,\varphi_n,\theta_n,\psi_n)\in\Sigma\), and \(\psi_n\to\psi\), since \(\abs{\psi-\varphi}\le\sin\theta\).  So \((\gamma,a,b,\xi)\) lies in \(\overline\Sigma\).

\ref{h4:pieces} By \cref{lem:symmetry}\ref{sym:h4}, it suffices to show that \(\cE_U[H;\nabla](z)\le0\) for every \(z\in\Sigma\) with \(\gamma\ge0\) and every gradient \(\nabla\) of the pieces at \(z\).  Such a point satisfies \(\gamma<\frac\pi2\).  If \(\nabla=\nabla H^{\rm c}(z)\) or \(\nabla=\nabla H^\pm(z)\), then \(\cE_U[H;\nabla](z)\) is the value of \(\cE^{\rm c}_U\) or \(\cE^\pm_U\) at the coordinates of \(z\), since \(H\) agrees with \(H^{\rm c}\) or \(H^\pm\) at \(z\).  So~\eqref{eq:h4-pieces} gives the claim.
\end{proof}

\subsection{The Corner Identity}\label{subsec:corner-identity}

In this subsection, we show that, for every member of \(\cH\) and every level \(U\), the residuals of the central piece and of the outer piece \(\xi\ge1\) vanish above the corner curve (\cref{lem:corner-identity}).  As explained in \cref{subsec:family}, two facts combine.  Above the corner curve, the gate is horizontal, so the center of the disk moves vertically along the pencil of the gate, and the term \(-U\cos\varphi\), the only one that contains \(U\), vanishes; and every member of \(\cH\) takes the value \(\gamma\) along the whole curve, by the choice of \(\sigma\), which makes the remaining terms cancel.  This is the exact part of the verification: it supplies the zero from which the corner lemmas of \cref{subsec:corner-lemmas} start, and it is also the reason why a direct interval bound cannot succeed near the corner curve, where the margin of \Hcond{4} tends to zero.

\paragraph{The Points Above the Corner Curve.}  By \cref{lem:h4}\ref{h4:closure}, the points of \(\overline\Sigma\) above the corner curve, that is, with \((a,b)=(0,1)\), are the points \((\gamma,0,1,\xi)\) with \(0\le\gamma<\frac\pi2\) and \(1\le\xi\le(\frac\pi2+\cos\gamma)/k(\gamma)\) (limits of gate states, not gate states: the gate has degenerated to the chord on the \(x\)-axis, \cref{fig:corner}).  Indeed, at \((a,b)=(0,1)\), \eqref{eq:ab} gives \(\theta=\frac\pi2-\gamma\) and \(\varphi=\frac\pi2\), and~\eqref{eq:sigma-bar} reads \(\abs{k(\gamma)\,\xi-\frac\pi2}\le\cos\gamma\).  So the central piece contains only the points with \(\xi=1\) among them.

\begin{lemma}[Corner identity]\label{lem:corner-identity}
Let \(H\in\cH\) and \(U\in\R\), and let \(\cE^{\rm c}_U\) and \(\cE^{+}_U\) be the residuals~\eqref{eq:piece-residuals} of the central piece and of the outer piece \(\xi\ge1\).  Then, for every \(\gamma\in[0,\frac\pi2)\),
\begin{equation}\label{eq:corner-zero}
\cE^{\rm c}_U(\gamma,0,1,1)=0\qquad\text{and}\qquad\cE^{+}_U(\gamma,0,1)=0 .
\end{equation}
\end{lemma}

\begin{proof}
Fix \(\gamma\in[0,\frac\pi2)\), and set \(k=k(\gamma)\).  At \((a,b)=(0,1)\), \eqref{eq:ab} gives \(\theta=\frac\pi2-\gamma\) and \(\varphi=\frac\pi2\), so
\begin{equation}\label{eq:corner-values}
\cos\varphi=0,\qquad\cos\theta=\sin\gamma,\qquad\sin\theta=\cos\gamma,\qquad\cA=\frac{\sin\varphi-\sin\gamma\cos\theta}{\cos\gamma}=\cos\gamma .
\end{equation}
Recall that \(a=(\theta-\varphi+\gamma)/(\pi+2\gamma)\) and \(b=(\theta+\varphi-\gamma)/(\pi-2\gamma)\).  Hence
\[
\frac{\partial a}{\partial\theta}=\frac1{\pi+2\gamma},\qquad\frac{\partial a}{\partial\gamma}=\frac{1-2a}{\pi+2\gamma},\qquad
\frac{\partial b}{\partial\theta}=\frac1{\pi-2\gamma},\qquad\frac{\partial b}{\partial\gamma}=\frac{2b-1}{\pi-2\gamma},
\]
and at \((a,b)=(0,1)\), the derivatives with respect to \(\gamma\) and to \(\theta\) coincide.  Therefore, for a continuously differentiable function \(f(\gamma,a,b,\xi)\) and \(F(\gamma,\varphi,\theta,\psi)=f(\gamma,a,b,\psi/k(\gamma))\), the chain rule gives, at the points with \((a,b)=(0,1)\),
\begin{equation}\label{eq:corner-chain}
\partial_\gamma F-\partial_\theta F=\partial_\gamma f-\xi\,\frac{k'}{k}\,\partial_\xi f,\qquad\partial_\psi F=\frac1k\,\partial_\xi f,
\end{equation}
where \(k'=\sin\gamma\) is the derivative of \(k\).

Consider first the central piece, with \(F=H^{\rm c}\) and \(f=\widehat H\), at \(\xi=1\); then \(\psi=k\) and \(\psi-\varphi=-\cos\gamma\).  At \(\xi=1\), the terms of~\eqref{eq:H} with the factors \(\xi-\xi^3\) and \(\xi^3-\xi^5\) vanish, so the definition of \(\sigma\) gives \(\widehat H(\gamma',0,1,1)=\cos^2\gamma'\,V(\gamma',0,1)+\sigma(\gamma')=\gamma'\) for every \(\gamma'\in[0,\frac\pi2)\); hence \(H^{\rm c}=\gamma\) at our point, and \(\partial_\gamma\widehat H(\gamma,0,1,1)=1\).  Let \(\nabla\) be the gradient of \(H^{\rm c}\) at the point, and let \(\zeta=\partial_\xi\widehat H(\gamma,0,1,1)\).  By~\eqref{eq:corner-chain}, \(\nabla_\gamma-\nabla_\theta=1-\zeta\sin\gamma/k\) and \(\nabla_\psi=\zeta/k\).  Substituting these values and~\eqref{eq:corner-values} into \cref{def:residual}, and using \(\cA=\sin\theta\), we obtain
\[
\cE^{\rm c}_U(\gamma,0,1,1)=\gamma\sin\gamma+\cos\gamma-\gamma\sin\gamma-\frac{\zeta}{k}\sin\gamma\cos\gamma-\cos\gamma\Bigl(1-\frac{\zeta}{k}\sin\gamma\Bigr)=0 .
\]

For the piece \(\xi\ge1\), we have \(F=H^+\) and \(f(\gamma,a,b)=\cos^2\gamma\,V(\gamma,a,b)+\sigma(\gamma)\), which does not depend on \(\xi\).  As before, \(f(\gamma',0,1)=\gamma'\) for every \(\gamma'\), so \(H^+=\gamma\) at the point, and by~\eqref{eq:corner-chain}, the gradient \(\nabla\) of \(H^+\) satisfies \(\nabla_\gamma-\nabla_\theta=1\) and \(\nabla_\psi=0\).  Hence \(\cE^{+}_U(\gamma,0,1)=\gamma\sin\gamma+\cos\gamma-\gamma\sin\gamma-\cos\gamma=0\).
\end{proof}

\begin{remark}[The value on the corner curve]\label{rem:corner}
The limit value \(\gamma\) of \(H\) above the corner curve is forced by \Hcond{1} and \Hcond{4} together: any other choice of \(\sigma\) that is compatible with \Hcond{4} gives a larger limit, and no larger limit is compatible with \Hcond{1}.  To see the first claim, let \(H\in\cH\), let \(\tilde\sigma\) be a continuously differentiable function on \([0,\frac\pi2)\) with \(\tilde\sigma(0)=0\), and let \(\tilde H\) be defined as in \cref{def:family}, with the coefficients of \(H\) but with \(\sigma\) replaced by \(\tilde\sigma\), and with the pieces of \cref{def:pieces}; the condition \(\tilde\sigma(0)=0\) makes \(\tilde H\) continuous at \(\gamma=0\), as in the proof of \cref{lem:reflections}.  Let \(\varepsilon=\tilde\sigma-\sigma\).  Then \(\tilde H\) tends to \(\gamma+\varepsilon(\gamma)\) instead of \(\gamma\) at the points above the corner curve, and the computation in the proof of \cref{lem:corner-identity} gives, for the residuals of \(\tilde H\),
\[
\tilde\cE^{\rm c}_U(\gamma,0,1,1)=\tilde\cE^{+}_U(\gamma,0,1)=-\varepsilon(\gamma)\sin\gamma-\varepsilon'(\gamma)\cos\gamma=-\cos^2\gamma\,\frac{d}{d\gamma}\Bigl(\frac{\varepsilon(\gamma)}{\cos\gamma}\Bigr).
\]
The points \((\gamma,0,1,1)\) are limits of points of \(\Sigma\) with \(\xi=1\), which lie in the central piece and in the outer piece \(\xi\ge1\); for example, the points \((\gamma,\frac1n,1-\frac1{n^2},1)\) lie in \(\Sigma\) for all large \(n\), since, by~\eqref{eq:ab}, \(\sin\theta-\abs{k(\gamma)-\varphi}=(\frac\pi2+\gamma)(1+\sin\gamma)/n+O(n^{-2})\) at these points, so that~\eqref{eq:sigma-bar} holds.  So if \(\tilde H\) satisfies \Hcond{4} at some level, then, by continuity, these residuals are nonpositive, \(\varepsilon(\gamma)/\cos\gamma\) is nondecreasing on \([0,\frac\pi2)\), and since \(\varepsilon(0)=0\), we obtain \(\varepsilon\ge0\).  Thus, among the functions \(\tilde H\) that satisfy \Hcond{4} at some level, the members of \(\cH\) take the smallest limit value \(\gamma\) on the corner curve.

If \(\tilde H\) also satisfies \Hcond{1}, then \(\varepsilon=0\).  Indeed, by~\eqref{eq:H} with \(\tilde\sigma\) in place of \(\sigma\), at every point of \(\Omega\) with \(\gamma\ge0\), the function \(\tilde H\) takes the values \(\cos^2\gamma\,V(\gamma,a,b)\mp\tilde\sigma(\gamma)\) at \(\psi=\mp k(\gamma)\); so \Hcond{1} gives \(2\tilde\sigma(\gamma)\le2k(\gamma)\), that is, \(\varepsilon(\gamma)\le k(\gamma)-\sigma(\gamma)=\bigl(\tfrac\pi2-\gamma-\cos\gamma\bigr)+\cos^2\gamma\,V(\gamma,0,1)\).  As \(\gamma\to\frac\pi2\), the right-hand side is \(O((\frac\pi2-\gamma)^2)\), since \(\frac\pi2-\gamma-\cos\gamma=O((\frac\pi2-\gamma)^3)\) and \(V(\gamma,0,1)\) is bounded; so \(\varepsilon(\gamma)/\cos\gamma\to0\), and a nonnegative nondecreasing function on \([0,\frac\pi2)\) with limit \(0\) vanishes identically.  So the limit value \(\gamma\) is the only one compatible with \Hcond{1} and \Hcond{4}.
\end{remark}

\begin{remark}[The edge \(a=0\) as \(\gamma\to\frac\pi2\)]\label{rem:edge}
As along the corner curve, \Hcond{4} has no room to spare at its end where \(\gamma\to\frac\pi2\).  Let \(t=\frac\pi2-\gamma\).  By~\eqref{eq:ab}, \(\theta=(\pi-t)a+tb\) and \(\varphi=\frac\pi2-\theta-t(1-2b)\), so, as \(t\to0\), the point \((\varphi,\theta)\) tends to \((\frac\pi2-\pi a,\pi a)\), whatever \(b\) is.  Hence the edge \(a=0\) of the square of \((a,b)\) shrinks to the end point \((\varphi,\theta)=(\frac\pi2,0)\) of the corner curve.  Correspondingly, the residuals have limits that vanish on this edge: as \(t\to0\), uniformly for \((a,b)\in[0,1]^2\) and \(\abs\xi\le1\),
\begin{equation}\label{eq:edge-limits}
\cE^{\rm c}_U(\gamma,a,b,\xi)\to(1-U)\sin(\pi a),\qquad
\cE^{+}_U(\gamma,a,b)\to(2-2b-U)\sin(\pi a)-\pi a\cos(\pi a).
\end{equation}
Indeed, the terms of~\eqref{eq:H} with the factor \(\cos^2\gamma\) tend to \(0\) together with their first derivatives with respect to \((\gamma,\varphi,\theta,\psi)\), because \(\cos^2\gamma=\sin^2t\) and, by~\eqref{eq:ab-explicit}, the derivatives of \(a\) and \(b\) are \(O(1/t)\); moreover, \(\sigma(\gamma)\to\frac\pi2\), \(\sigma'(\gamma)\to1\), \(k(\gamma)\to\frac\pi2\), and \(\cA=[\cos(\theta+t(1-2b))-\cos t\cos\theta]/\sin t\to(2b-1)\sin(\pi a)\).  Substituting these limits into \cref{def:residual} gives~\eqref{eq:edge-limits}.  The corner lemmas of \cref{subsec:corner-lemmas} apply along this edge as well (see the end of that subsection), and \cref{subsec:cert-h4} describes the boxes that reach \(\gamma=\frac\pi2\).
\end{remark}

\subsection{The Corner Lemmas}\label{subsec:corner-lemmas}

In this subsection, we reduce~\eqref{eq:h4-pieces} near the corner curve, where \Hcond{4} has no room to spare, to inequalities between first derivatives, for the central piece (\cref{lem:corner}) and for the outer piece \(\xi\ge1\) (\cref{lem:corner-outer}).  Near the corner curve, \((a,b)\) is near the corner \((0,1)\) of the unit square, which gives the lemmas their name (\cref{fig:corner}(b)).  The lemmas start from the exact zero at \((a,b,\xi)=(0,1,1)\) (\cref{lem:corner-identity}) and show that the residual cannot become positive nearby.  Moving the ends of the gate away from the corner, which tilts the gate away from the \(x\)-axis, decreases the residual at a guaranteed rate, while moving \(\xi\) down from \(1\) can increase it at a rate bounded by a constant \(C\).  Since \(\abs\delta\le\abs G\), the variable \(\xi\) can drop below \(1\) only in proportion to the distance from the corner, so the guaranteed decrease pays for the possible increase.  For the outer piece, the residual does not depend on \(\xi\), and only the signs of two first derivatives are needed.  So both lemmas ask only for inequalities between first derivatives, which hold with room to spare, and interval arithmetic verifies them on boxes (\cref{subsec:cover}).

\begin{lemma}[Corner lemma for the central piece]\label{lem:corner}
Let \(H\in\cH\), \(U\in\R\) and \(\gamma\in[0,\frac\pi2)\), and set \(k=k(\gamma)\) and \(E(a,b,\xi)=\cE^{\rm c}_U(\gamma,a,b,\xi)\).  Let \(\bar a,\underline b\in[0,1]\) and \(\underline\xi\le1\), let \(K=[0,\bar a]\times[\underline b,1]\times[\underline\xi,1]\), and let
\[
\theta_0=\Bigl(\frac\pi2-\gamma\Bigr)\underline b,\qquad\theta_1=\Bigl(\frac\pi2+\gamma\Bigr)\bar a+\frac\pi2-\gamma .
\]
Let \(C\ge0\) be such that \(\partial_\xi E\ge-C\) on \(K\), and suppose that, at every point of \(K\),
\begin{equation}\label{eq:corner-hyp}
\partial_aE\le-\frac Ck\Bigl(\frac\pi2+\gamma\Bigr)(1+\cos\theta_0)\qquad\text{and}\qquad\partial_bE\ge\frac Ck\Bigl(\frac\pi2-\gamma\Bigr)(1-\cos\theta_1).
\end{equation}
Then \(E(a,b,\xi)\le0\) at every point \((a,b,\xi)\in K\) with \(k\xi\ge\varphi-\sin\theta\), where \(\varphi\) and \(\theta\) are given by~\eqref{eq:ab}.
\end{lemma}

\begin{figure}[!htbp]
\centering
\begin{tikzpicture}[paperfig]
\path[use as bounding box] (-0.1,0.2) rectangle (13.4,5.25);
\draw[fdisk] (2.6,2.681) circle[radius=2];
\draw[faxis,-{Latex[length=1.8mm]}] (0.159,1.9) -- (5.441,1.9);
\draw[fold] (0.759,1.9) arc[start angle=203,end angle=186.05,radius=2];
\draw[fnew] (0.759,1.9) arc[start angle=203,end angle=327.62,radius=2];
\draw[fdashgate] (0.759,1.9) -- (4.441,1.9);
\draw[fgate] (0.611,2.471) -- (4.289,1.61);
\node[fpt] at (0.759,1.9) {};
\node[fpt] at (4.441,1.9) {};
\node[fpt] at (0.611,2.471) {};
\node[fpt] at (4.289,1.61) {};
\node[fpt] at (2.6,2.681) {};
\node[text=figBlueText] at (0.33,2.09) {$a$};
\node[text=figGreenText] at (2.41,0.40) {$b$};
\node[anchor=south] at (2.6,2.781) {$o$};
\node[anchor=north east] at (0.70,1.86) {$w_-(O)$};
\node[anchor=south west] at (4.501,1.96) {$w_+(O)$};
\node[anchor=south east] at (0.531,2.531) {$a_G$};
\node[anchor=north west] at (4.369,1.55) {$b_G$};
\node[anchor=south,text=figOrangeText] at (2.0,2.22) {$G$};
\node[anchor=south,text=black] at (2.85,4.801) {$O$};
\draw[fthick] (8.8,0.75) rectangle (12,3.95);
\draw[fthick,fill=figFill] (8.8,2.95) rectangle (9.85,3.95);
\node[fpt] at (8.8,3.95) {};
\node[anchor=south east] at (8.72,4.03) {$(0,1)$};
\node[anchor=north] at (10.4,0.65) {$a$};
\node[anchor=east] at (8.68,2.35) {$b$};
\node[anchor=north] at (8.8,0.7) {$0$};
\node[anchor=north] at (12,0.7) {$1$};
\node[anchor=east] at (8.75,0.75) {$0$};
\node[anchor=south] at (9.85,4) {$\bar a$};
\node[anchor=west] at (12.05,2.95) {$\underline b$};
\draw[fthick] (9.85,3.95) -- (9.85,4);
\draw[fguide] (9.85,2.95) -- (12,2.95);
\node[anchor=west] at (0,4.85) {(a)};
\node[anchor=west] at (7.6,4.85) {(b)};
\end{tikzpicture}
\caption{The gate becomes horizontal as its coordinates approach the corner curve.  (a)~The blue arc is the fraction \(a\) of the upper arc from \(w_-(O)\), and the green arc is the fraction \(b\) of the lower arc from the same point.  As \((a,b)\) tends to \((0,1)\), the gate tends to the dashed chord on the \(x\)-axis.  (b)~For fixed \(\gamma\), gate states have \((a,b)\) in the interior of the unit square, and the corner curve corresponds to the boundary point \((0,1)\).  The shaded rectangle is \([0,\bar a]\times[\underline b,1]\), the projection of the box \(K\) in the central corner lemma onto this square.}
\label{fig:corner}
\end{figure}

Note that \(\theta_0\) and \(\theta_1\) are the smallest and the largest value of \(\theta=(\frac\pi2+\gamma)a+(\frac\pi2-\gamma)b\) on \([0,\bar a]\times[\underline b,1]\), the shaded rectangle of \cref{fig:corner}(b).  The factors \((\frac\pi2+\gamma)(1+\cos\theta)\) and \((\frac\pi2-\gamma)(1-\cos\theta)\) in~\eqref{eq:corner-hyp} are the rates at which \(\varphi-\sin\theta\), the lower bound for \(k\xi\) in~\eqref{eq:sigma-bar}, decreases in \(a\) and increases in \(b\); the proof integrates these rates along a path from the corner.

\begin{proof}
Fix \((a,b,\xi)\in K\) with \(k\xi\ge\varphi-\sin\theta\).  The segments from \((0,1,1)\) to \((a,1,1)\), from \((a,1,1)\) to \((a,b,1)\), and from \((a,b,1)\) to \((a,b,\xi)\) lie in \(K\).  Since \(E(0,1,1)=0\) by \cref{lem:corner-identity}, integrating the partial derivatives of \(E\) along these segments gives
\[
E(a,b,\xi)=\int_0^a\partial_aE(a',1,1)\,da'-\int_b^1\partial_bE(a,b',1)\,db'-\int_\xi^1\partial_\xi E(a,b,\xi')\,d\xi' .
\]
By~\eqref{eq:corner-hyp} and the choice of \(C\), we obtain
\[
E(a,b,\xi)\le-\frac Ck\,\Lambda+C(1-\xi),\qquad\text{where}\quad\Lambda=\Bigl(\frac\pi2+\gamma\Bigr)(1+\cos\theta_0)\,a+\Bigl(\frac\pi2-\gamma\Bigr)(1-\cos\theta_1)(1-b).
\]
Since \(C\ge0\), it suffices to show that \(k(1-\xi)\le\Lambda\).  Let \(f(a',b')=\varphi-\sin\theta\), where \(\varphi\) and \(\theta\) are given by~\eqref{eq:ab} as functions of \((a',b')\).  By~\eqref{eq:ab},
\[
\partial_{a'}f=-\Bigl(\frac\pi2+\gamma\Bigr)(1+\cos\theta),\qquad\partial_{b'}f=\Bigl(\frac\pi2-\gamma\Bigr)(1-\cos\theta).
\]
On \([0,\bar a]\times[\underline b,1]\), we have \(\theta_0\le\theta\le\theta_1\), and \([\theta_0,\theta_1]\subset[0,\pi]\) because \(\bar a\le1\); hence \(\cos\theta_1\le\cos\theta\le\cos\theta_0\).  Integrating along the segments from \((0,1)\) to \((a,1)\) and from \((a,1)\) to \((a,b)\), we obtain \(f(0,1)-f(a,b)\le\Lambda\).  Finally, \(f(0,1)=\frac\pi2-\cos\gamma=k\) by~\eqref{eq:corner-values}, and \(k\xi\ge f(a,b)\) by assumption, so \(k(1-\xi)\le f(0,1)-f(a,b)\le\Lambda\).
\end{proof}

For the outer piece, the counterpart is simpler.

\begin{lemma}[Corner lemma for the outer piece]\label{lem:corner-outer}
Let \(H\in\cH\), \(U\in\R\) and \(\gamma\in[0,\frac\pi2)\), and set \(F(a,b)=\cE^{+}_U(\gamma,a,b)\).  Let \(\bar a,\underline b\in[0,1]\).  If \(\partial_aF\le0\) on \([0,\bar a]\times[\underline b,1]\) and \(\partial_bF(0,b)\ge0\) for all \(b\in[\underline b,1]\), then \(F\le0\) on \([0,\bar a]\times[\underline b,1]\).
\end{lemma}

\begin{proof}
For \((a,b)\in[0,\bar a]\times[\underline b,1]\), the two hypotheses and \cref{lem:corner-identity} give \(F(a,b)\le F(0,b)\le F(0,1)=0\).
\end{proof}

The reflection \(\rho_2\) carries both lemmas to the image \(\{(\gamma,1,0):0\le\gamma<\frac\pi2\}\) of the corner curve.  Indeed, by \cref{lem:reflections}, \(\rho_2\) acts by \((\gamma,a,b,\xi)\mapsto(\gamma,1-a,1-b,-\xi)\), it maps \(\overline\Sigma\) onto itself, and \cref{lem:regularity}\ref{reg:rho2} gives
\begin{equation}\label{eq:rho2-residuals}
\cE^{\rm c}_U(\gamma,1-a,1-b,-\xi)=\cE^{\rm c}_U(\gamma,a,b,\xi),\qquad\cE^{-}_U(\gamma,1-a,1-b)=\cE^{+}_U(\gamma,a,b).
\end{equation}
\Cref{lem:corner,lem:corner-outer} also apply near the edge \(a=0\) as \(\gamma\to\frac\pi2\) (\cref{rem:edge}), since they allow every \(\underline b\in[0,1]\) and every \(\gamma<\frac\pi2\); by~\eqref{eq:rho2-residuals}, they apply near the edge \(a=1\) as well.

\subsection{The Cover}\label{subsec:cover}

We now verify~\eqref{eq:h4-pieces} for \(\Hs\) at level \(U_0\) by covering the domain with boxes; a box is accepted in one of three ways: no point of the domain lies over it, a direct interval bound shows that the residual is nonpositive over it, or, near the corner curve, one of the corner lemmas of \cref{subsec:corner-lemmas} applies.  This is the numerical part of the verification, and the last step in the proof of \cref{thm:explicit}.

\begin{proposition}[\Hcond{4} for \(\Hs\)]\label{prop:h4}
The function \(\Hs\) of \cref{thm:explicit} satisfies \Hcond{4} at level \(U_0\).
\end{proposition}

Before proving \cref{prop:h4}, we describe the boxes and when a box is accepted.

\paragraph{Boxes.}  Let \(H=\Hs\) and \(U=U_0\) in~\eqref{eq:piece-residuals}.  We use \(\tau=2\gamma/\pi\in[0,1]\) in place of \(\gamma\).  A \emph{box} is a set \(Q=[\tau^-,\tau^+]\times[a^-,a^+]\times[b^-,b^+]\subset[0,1]^3\) with rational endpoints such that \([\tau^-,\tau^+]\) lies in one of the spline intervals \([\frac i{32},\frac{i+1}{32}]\), \(0\le i\le31\); on such a box, the B-splines \(N_l\) are polynomials in \(\tau\).  A point \((\gamma,a,b,\xi)\) \emph{lies over} \(Q\) if \(\gamma<\frac\pi2\) and \((2\gamma/\pi,a,b)\in Q\), whatever its coordinate \(\xi\).  \paragraph{Accepted Boxes.}  For one of the three pieces, a box \(Q\) is \emph{accepted} if one of the following conditions is verified by exact rational arithmetic or by interval arithmetic with outward rounding.
\begin{enumerate}[label=(Q\arabic*),leftmargin=2.7em]
\item\label{acc:empty} No point of \(\overline\Sigma\) in the piece lies over \(Q\).
\item\label{acc:direct} An upper bound for the residual of the piece at the points of \(\overline\Sigma\) in the piece that lie over \(Q\) is at most \(0\).
\item\label{acc:corner} The piece is the central piece, and there are a set \(\widehat Q=[\hat\tau^-,\hat\tau^+]\times[0,\bar a]\times[\underline b,1]\subset[0,1]^3\) that contains \(Q\) and numbers \(\underline\xi\le1\) and \(C\ge0\) with the following two properties: \(\xi\ge\underline\xi\) at every point of \(\overline\Sigma\) in the central piece over \(Q\); and, for every \(\gamma\in[0,\frac\pi2)\) with \(2\gamma/\pi\in[\hat\tau^-,\hat\tau^+]\), the hypotheses of \cref{lem:corner} hold with these numbers \(\bar a\), \(\underline b\), \(\underline\xi\) and \(C\).
\item\label{acc:corner-outer} The piece is the outer piece \(\xi\ge1\), and \(Q\) is contained in a set \(\widehat Q=[\hat\tau^-,\hat\tau^+]\times[0,\bar a]\times[\underline b,1]\subset[0,1]^3\) such that, for every \(\gamma\in[0,\frac\pi2)\) with \(2\gamma/\pi\in[\hat\tau^-,\hat\tau^+]\), the hypotheses of \cref{lem:corner-outer} hold with these numbers \(\bar a\) and \(\underline b\).
\end{enumerate}
Moreover, a box \(Q\) is accepted for the central piece or for the piece \(\xi\le-1\) if its reflection \(\rho_2(Q)=[\tau^-,\tau^+]\times[1-a^+,1-a^-]\times[1-b^+,1-b^-]\) satisfies~\ref{acc:corner} or~\ref{acc:corner-outer}, respectively.  The set \(\widehat Q\) reaches the corner \((a,b)=(0,1)\) because the corner lemmas integrate from there.

\Cref{fig:box-acceptance} summarizes these acceptance tests.

\begin{figure}[!htbp]
\centering
\begin{tikzpicture}[paperfig,font=\normalsize,text=black,
  testbox/.style={draw=figInk,line width=1pt,fill=figFill,
    rounded corners=2pt,align=center,minimum height=.85cm,inner sep=6pt},
  testarrow/.style={draw=figInk,line width=1pt,-{Latex[length=2mm]}}]
\path[use as bounding box] (-.15,0) rectangle (14.65,4);
\node[testbox,minimum width=3.6cm] (input) at (1.8,2)
  {Box and piece};
\node[testbox,minimum width=4.5cm] (empty) at (7.4,3.35)
  {No domain point};
\node[testbox,minimum width=4.5cm] (direct) at (7.4,2)
  {Direct residual bound};
\node[testbox,minimum width=4.5cm] (corner) at (7.4,.65)
  {Corner lemma};
\node[testbox,minimum width=3cm,fill=figGreen!8] (accepted) at (13,2)
  {Accepted};
\draw[testarrow] (input.east) -- (empty.west);
\draw[testarrow] (input.east) -- (direct.west);
\draw[testarrow] (input.east) -- (corner.west);
\draw[testarrow] (empty.east) -- (accepted.west);
\draw[testarrow] (direct.east) -- (accepted.west);
\draw[testarrow] (corner.east) -- (accepted.west);
\end{tikzpicture}
\caption{A box is accepted for a fixed piece in one of three ways. Either no point of the domain in that piece lies over the box, a direct interval bound proves that the residual is nonpositive, or the appropriate corner lemma applies on a containing box that reaches the corner. The last route uses bounds on first derivatives. The corner tests also apply to reflected boxes.}
\label{fig:box-acceptance}
\end{figure}

\begin{lemma}[Accepted boxes]\label{lem:accepted}
If a box \(Q\) is accepted for a piece, then the residual of the piece is nonpositive at every point of \(\overline\Sigma\) in the piece over \(Q\).
\end{lemma}

\begin{proof}
Under~\ref{acc:empty}, there is no such point, and under~\ref{acc:direct}, this is the condition itself.  Under~\ref{acc:corner}, every such point \((\gamma,a,b,\xi)\) has \((a,b)\in[0,\bar a]\times[\underline b,1]\) and \(\underline\xi\le\xi\le1\), and it satisfies \(k(\gamma)\xi\ge\varphi-\sin\theta\) by~\eqref{eq:sigma-bar}; so \cref{lem:corner} applies.  Under~\ref{acc:corner-outer}, \cref{lem:corner-outer} applies, since \((a,b)\in[0,\bar a]\times[\underline b,1]\).  For reflected boxes, the claim follows from~\eqref{eq:rho2-residuals}, since \(\rho_2\) maps \(\overline\Sigma\) onto itself, the central piece onto itself, and the piece \(\xi\le-1\) onto the piece \(\xi\ge1\).
\end{proof}

How the bound in~\ref{acc:direct} and the hypotheses in~\ref{acc:corner} and~\ref{acc:corner-outer} are verified on a box is described in \cref{subsec:cert-h4}.

\begin{table}[!htbp]
\centering\small
\begin{tabular}{@{}llr@{}}
\toprule
 Cover & Boxes & Number\\
\midrule
 Main & all boxes, for the three pieces and the \(32\) spline intervals & \(1{,}037{,}865\)\\
 & accepted for their piece & \(945{,}810\)\\
 & not accepted (all for the central piece) & \(92{,}055\)\\
 & \quad contained in a box~\eqref{eq:corner-boxes} & \(44{,}134\)\\
 & \quad with reflection \(\rho_2(Q)\) contained in a box~\eqref{eq:corner-boxes} & \(47{,}920\)\\
 & \quad contained in \(Q_0\) & \(1\)\\
\midrule
 Supplementary & boxes~\eqref{eq:corner-boxes}, accepted by~\ref{acc:corner} & \(1{,}499\)\\
 & halves of \(Q_0\), accepted by~\ref{acc:direct} & \(2\)\\
\bottomrule
\end{tabular}
\caption{The two covers in the proof of \cref{prop:h4}.  The main cover accepts all but \(92{,}055\) of its boxes; all exceptions belong to the central piece, and all but one lie within \(2^{-24}\) of the corner curve or of its image under \(\rho_2\), with \(\gamma\) between about \(0.4032\) and \(0.4038\), where the supplementary cover accepts them with sharper enclosures; the remaining one is covered by the two halves of a box \(Q_0\).  The counts of the main cover are totals over the three pieces; the supplementary cover concerns the central piece only, and it may overlap accepted boxes of the main cover, so the counts of the two covers are not added.  The acceptance conditions are those of \cref{subsec:cover}: under~\ref{acc:empty}, no point of \(\overline\Sigma\) in the piece lies over the box; under~\ref{acc:direct}, a direct interval bound shows that the residual is nonpositive over it; under~\ref{acc:corner} and~\ref{acc:corner-outer}, the corner lemma for the central piece or for the outer piece \(\xi\ge1\) applies.  The boxes~\eqref{eq:corner-boxes} and \(Q_0\) are defined in the proof below.}
\label{tab:cover}
\end{table}

\begin{proof}[Proof of \cref{prop:h4}]
By \cref{lem:h4}\ref{h4:pieces} and \cref{lem:accepted}, it suffices to show that, for each piece, every point of \([0,1]^3\) lies in a box that is accepted for this piece.  The certificates of this proposition, that is, the files of \cref{subsec:data} that record the boxes and the bounds computed on them, describe two families of boxes; \cref{tab:cover} gives their sizes.

The \emph{main cover} partitions, for each of the three pieces and each \(0\le i\le31\), the box \([\frac i{32},\frac{i+1}{32}]\times[0,1]^2\) into boxes obtained by repeated bisection, that is, by halving one of the three coordinate intervals at a time.  All but \(92{,}055\) of its about \(1.04\) million boxes are accepted for their piece, and all exceptions belong to the central piece.  With one exception, they lie within \(2^{-24}\) of the corner curve or of its image under \(\rho_2\), with \(\gamma\) between about \(0.4032\) and \(0.4038\): each such box \(Q\), or its reflection \(\rho_2(Q)\), is contained in one of the \(1{,}499\) boxes
\begin{equation}\label{eq:corner-boxes}
\Bigl[\frac{j}{2^{22}},\frac{j+1}{2^{22}}\Bigr]\times[0,\bar a_j]\times\bigl[1-2^{-24},1\bigr],\qquad1{,}076{,}681\le j\le1{,}078{,}179,
\end{equation}
where \(\bar a_j\in\{2^{-25},2^{-24}\}\) is recorded in the certificate.  The remaining exception is a box of the main cover on which the enclosures give a positive upper bound for the residual, of the order \(10^{-10}\), so that it is not accepted; it is contained in a box \(Q_0\) whose coordinates are given in \cref{subsec:cert-h4}, and bisecting \(Q_0\) once in \(\tau\) gives negative upper bounds on both halves.

The \emph{supplementary cover} consists of the \(1{,}499\) boxes~\eqref{eq:corner-boxes} and of the two halves of \(Q_0\).  Each box~\eqref{eq:corner-boxes} satisfies~\ref{acc:corner}, with sharper enclosures than those of the main cover, and both halves of \(Q_0\) satisfy~\ref{acc:direct} for the central piece.  \Cref{subsec:cert-h4} describes these checks.  All containments above are checked in exact rational arithmetic.

Consequently, every box \(Q\) of the main cover that is not accepted is covered by boxes accepted for the central piece.  Indeed, if \(Q\) lies in the box~\eqref{eq:corner-boxes} with index \(j\), then \(Q\) satisfies~\ref{acc:corner} with this box as \(\widehat Q\) and with the numbers \(\underline\xi\) and \(C\) of this box; if \(\rho_2(Q)\) lies in such a box, then \(\rho_2(Q)\) satisfies~\ref{acc:corner}, so \(Q\) is accepted through its reflection; and if \(Q\) lies in \(Q_0\), it is covered by the two halves of \(Q_0\).  Together with the accepted boxes of the main cover, this shows that, for each piece, every point of \([0,1]^3\) lies in a box accepted for this piece.
\end{proof}

\begin{remark}[The smallest level for \(\Hs\)]\label{rem:margin}
Numerically, \(\Hs\) satisfies \Hcond{4} also at levels slightly below \(U_0\): a search for the points at which \Hcond{4} is tightest suggests that it does so at every level above about \(1.59485806\), about \(1.9\times10^{-6}\) below \(U_0\).  Pushing \(U_0\) closer to this threshold would require subdividing the cover much more finely near the corner curve, where the margin of \Hcond{4} approaches zero; in our experiments, the number of boxes grows rapidly and the verification ceases to terminate within a reasonable time.
\end{remark}

\subsection{Proof of \texorpdfstring{\cref{thm:explicit}}{Theorem 5.1}}\label{subsec:proofs}

We can now complete the proof of \cref{thm:explicit}, assembling the results of this section with those of \cref{sec:family}.

\begin{proof}[Proof of \cref{thm:explicit}]
By \cref{prop:h0h3}, \(\Hs\) belongs to \(\cH\) and satisfies \Hcond{0}--\Hcond{3} at level \(U_0\), and by \cref{prop:h4}, it satisfies \Hcond{4} at level \(U_0\).
\end{proof}

\section{Proof of the Main Theorem}\label{sec:proof-overview}

We now assemble the proof of \cref{thm:main}.  The argument is a chain of reductions, each replacing a statement by one that suffices for it and involves fewer or more tractable parameters.  \Cref{tab:chain} summarizes the chain; the rest of this section traces its logic.

\begin{table}[!htb]
\centering\small
\begin{tabular}{@{}l@{\hspace{0.7em}}P{13.6cm}@{\hspace{0.7em}}c@{}}
\toprule
& Statement, and the result that reduces it to the next one & Parameters\\
\midrule
\multicolumn{3}{@{}l}{\emph{From the whole chain to one step (\cref{sec:chains,sec:accounting})}}\\
(1) & \cref{thm:main} with a constant \(U\ge\pi/2\) in place of \(1.5949\) & \\
\(\Downarrow\) & by the reduction to disk chains (\cref{thm:reduction}, \cref{sec:chains}), adapted from Xia~\cite{xia2013} & \\
(2) & the chain inequality \(P(u,v)\le U\norm{u-v}\) for every disk chain & \(3n\)\\
\(\Downarrow\) & by the accounting identity~\eqref{eq:accounting} (\cref{subsec:recursion}), which splits \(P(u,v)-U\norm{u-v}\) into one term per disk; each term is local by \cref{lem:update} & \\
(3) & for every disk chain with at least two disks (one disk: \cref{lem:one-disk}), the sum over its disks of \(\text{cost}-U\cdot\text{progress}\) is at most \(0\) & \(3n\)\\
\(\Downarrow\) & by the potential criterion (\cref{thm:criterion}, \cref{subsec:potentials}), which is exact (\cref{prop:exact}) & \\
(4) & a feasible potential of level \(U\) exists & \(9\)\\
\midrule
\multicolumn{3}{@{}l}{\emph{From one step to a computer check (\crefrange{sec:pencil}{sec:verification})}}\\
\(\Downarrow\) & by the symmetries of \cref{subsec:angles} and the form \(W_{H,U}(s)\) of \cref{def:residual}, which are tight & \\
(5) & a feasible potential of level \(U\) of the form \(s\mapsto W_{H,U}(s)\) exists & \(7\)\\
\(\Downarrow\) & by the two stages of a step and the pencil lemma (\cref{thm:sufficient}, \cref{sec:pencil}) & \\
(6) & one function of four angles satisfies the five conditions \Hcond{0}--\Hcond{4} at level \(U\) & \(4\)\\
\(\Downarrow\) & by taking \(U=U_0\) and the explicit function \(\Hs\) of \cref{sec:family} & \\
(7) & \(\Hs\) satisfies \Hcond{0}--\Hcond{4} at level \(U_0\) (\cref{thm:explicit}): \Hcond{0}--\Hcond{3} in at most three variables (\cref{prop:h0h3}); \Hcond{4}, which has no room to spare near the corner curve, by the corner lemmas and a cover by boxes (\cref{sec:verification}, \cref{app:computation}) & \(\le4\)\\
\bottomrule
\end{tabular}
\caption{The proof of \cref{thm:main} as a chain of reductions, for a constant \(U\ge\pi/2\) (below \(\pi/2\), the chain inequality fails already for one disk, \cref{lem:one-disk}).  Each arrow leads from a statement to one that suffices for it, and names the result that performs the reduction.  The last column gives the number of real parameters of the inequalities to be proved, with the line through \(u\) and \(v\) as the \(x\)-axis: \(3n\) for a disk chain with \(n\) disks, where \(n\) is arbitrary, and from~(4) on, the largest number over the inequalities, attained by~(S).  The first five reductions hold for every \(U\ge\pi/2\); the last step fixes \(U=U_0\).}
\label{tab:chain}
\end{table}

\paragraph{From a Triangulation to a Disk Chain.}
The starting point is the reformulation introduced by Xia~\cite{xia2013}: the circumdisks of the Delaunay triangles crossed by a segment \([u,v]\) form a disk chain, and the shortest path along the circles of this chain is at least as long as the shortest path in the triangulation.  By \cref{thm:reduction}, the stretch bound for every Delaunay triangulation therefore reduces to the chain inequality \(P(u,v)\le U\norm{u-v}\) for every disk chain.  This reduction involves \(3n\) real parameters, where \(n\) is the number of disks.

\paragraph{From the Whole Chain to One Step.}
Progress accounting (\cref{sec:accounting}) splits the chain inequality into one term per disk, and the potential criterion (\cref{thm:criterion} and \cref{prop:exact}) shows that proving it is equivalent to finding a feasible potential, a function on gate states whose decrease at each step pays for the excess of cost over \(U\) times progress.  Each step involves nine real parameters, regardless of \(n\).

\paragraph{From Nine Parameters to Four.}
Translation and scaling symmetries remove two parameters, and the pencil lemma (\cref{lem:pencil}) removes three more by decomposing each step into a gate slide along the current circle and a disk slide along a pencil.  The sufficient conditions \Hcond{0}--\Hcond{4} of \cref{thm:sufficient} reduce the problem to finding one function \(H\) of the four angular coordinates \((\gamma,\varphi,\theta,\psi)\) of a gate state.

\paragraph{The Explicit Verification.}
At the fixed level \(U_0=1.59485994\), \cref{prop:h0h3} shows that the function \(\Hs\) of \cref{sec:family} satisfies \Hcond{0}--\Hcond{3}, by checks in at most three variables.  \Cref{prop:h4} in \cref{sec:verification} verifies condition \Hcond{4}, which has no margin near the corner curve, by the corner lemmas and a cover of approximately \(1.04\) million boxes.  Together, these give \cref{thm:explicit}.

\begin{proof}[Proof of \cref{thm:main}]
By \cref{thm:explicit}, \(\Hs\) satisfies \Hcond{0}--\Hcond{4} at level \(U_0\).  Since \(U_0>1.59>\frac\pi2\), \cref{thm:sufficient} shows that the chain inequality holds at level \(U_0\), and \cref{thm:reduction} gives \(d_{\cD}(u,v)\le U_0\norm{u-v}\) for every finite set \(S\subset\R^2\), every Delaunay triangulation \(\cD\) of \(S\), and all \(u,v\in S\).  Finally, \(U_0=1.59485994<1.5949\).
\end{proof}

\section{Concluding Remarks}\label{sec:conclusion}

We have shown that the stretch factor of every planar Delaunay triangulation is less than \(1.5949\), narrowing the gap to the lower bound \(1.5932\) of Xia and Zhang~\cite{xiazhang2011} to less than \(0.0017\).  The upper and lower bounds now agree to within \(0.11\%\).

The main open question is the \emph{disk-chain gap}: whether the passage from triangulations to the larger class of disk chains loses anything.  More precisely, does the smallest level at which the chain inequality holds equal the worst-case stretch factor of Delaunay triangulations?  If so, our framework provides a path toward the exact value.  If this smallest level is larger, reaching the exact Delaunay value will require geometric information discarded by the disk-chain reduction.

Numerical experiments reported in \cref{rem:margin} suggest that the fixed function \(\Hs\) satisfies \Hcond{4} above about \(1.59485806\), roughly \(1.9\times10^{-6}\) below \(U_0\). This leaves little room for improving the bound by verifying the
same function at a lower level. More substantial improvements
may come from finding a different function, within \(\cH\) or
in a richer family, and from relaxing the sufficient conditions
by treating the two stages of a step jointly. By \cref{thm:criterion,prop:exact}, the local potential criterion is exact.

The computer-assisted verification is the computational bottleneck: the cover for \Hcond{4} already uses approximately \(1.04\) million boxes, and the box count grows rapidly as the level approaches the threshold \(1.59485806\) of \(\Hs\).  This bottleneck is tied to the particular function \(\Hs\) and to the current verification strategy.

\section*{Acknowledgments}
\addcontentsline{toc}{section}{Acknowledgments}

The mathematical framework, proof strategy, and all proofs are the work of the authors.  We used ChatGPT (OpenAI) as a computational tool in the design of the parametrized family \(\cH\) of \cref{def:family} and the numerical search for the coefficients of \(\Hs\), and in preparing this manuscript.  The authors take full responsibility for the correctness and originality of all content, including references.  The correctness of the computer-assisted verification relies on the stored coefficients, the certificates and the reproducible checks described in \cref{app:computation}.

\appendix
\section{From Triangulations to Disk Chains}\label{app:reduction}

In this appendix, we prove \cref{thm:reduction}, the only part of the proof in which the triangulation itself appears.  The proof adapts Xia's argument~\cite{xia2013} to our definition of disk chains (\cref{def:chain}), following the sketch given after \cref{thm:reduction} step by step, with full details.  The idea is that the circumdisks of the triangles that \([u,v]\) crosses form a disk chain whose gates are edges of \(\cD\), and that each rail can be replaced by edges of \(\cD\) that are chords of its circle, and hence not longer.  The care goes into cocircular points: if four or more points of \(S\) lie on the circumcircle of a triangle, consecutive triangles along \([u,v]\) may have the same circumcircle, which a disk chain does not allow, so we work with the faces of the Delaunay subdivision, the convex hulls of the points of \(S\) on these circles.  Points of \(S\) on \([u,v]\) are handled by splitting the segment at them.

The proof uses four lemmas.  \Cref{lem:faces} collects the properties of the faces of the Delaunay subdivision that we need; \cref{lem:segment} shows that a segment \([u,v]\) crosses a sequence of faces through common edges; \cref{lem:generated} shows that the circumdisks of these faces form a disk chain whose gates are these edges; and \cref{lem:rails-to-paths} replaces every rail by a path in \(\cD\) that is not longer.

Throughout, \(S\subset\R^2\) is a finite set, and \(\cD\) is a Delaunay triangulation of \(S\).  We use the following formal definition.  A Delaunay triangulation of \(S\) is a set \(\cD\) of triangles with vertices in \(S\) that cover the convex hull of \(S\), such that every point of \(S\) is a vertex of some triangle, any two triangles are disjoint or meet in a common vertex or a common edge, and the open circumdisk of every triangle contains no point of \(S\).

For a closed disk \(O\) with center \(o\) and radius \(r\), the \emph{power} of a point \(x\in\R^2\) is \(\pow_O(x)=\norm{x-o}^2-r^2\); it is negative in the interior of \(O\), zero on \(\partial O\), and positive outside \(O\).  For two closed disks \(O\) and \(O'\) with centers \(o,o'\) and radii \(r,r'\),
\begin{equation}\label{eq:power-difference}
\pow_O(x)-\pow_{O'}(x)=2\langle o'-o,\,x\rangle+\norm{o}^2-r^2-\norm{o'}^2+r'^2
\end{equation}
is an affine function of \(x\).  If the two circles are distinct and meet, this function is not constant, so its zero set is a line, the radical axis, which contains the common points of the two circles.

For a triangle \(T\in\cD\), let \(O_T\) be its closed circumdisk, and let \(F_T\) be the convex hull of \(S\cap\partial O_T\).  We call the sets \(F_T\) the \emph{faces}; they are the faces of the Delaunay subdivision of \(S\), but we use only the properties proved below.  A face \(F\) has at least three vertices on a circle, so it determines this circle; we denote by \(O_F\) the closed disk bounded by it; its open disk contains no point of \(S\).

\begin{lemma}[Faces]\label{lem:faces}
\leavevmode
\begin{enumerate}[label=(\arabic*)]
\item\label{faces:tri} For every \(T\in\cD\), the face \(F_T\) is a convex polygon whose vertex set is \(S\cap\partial O_T\).  The triangles of \(\cD\) contained in \(F_T\) are exactly the triangles \(T'\in\cD\) with \(\partial O_{T'}=\partial O_T\), and they cover \(F_T\).  In particular, the faces cover the convex hull of \(S\).
\item\label{faces:edges} Every edge of a face is an edge of \(\cD\).
\item\label{faces:meet} Any two distinct faces are disjoint, or meet in a common vertex, or meet in a common edge.
\end{enumerate}
\end{lemma}

\begin{proof}
Parts~\ref{faces:tri} and~\ref{faces:meet} use that the difference of the powers of two circles is an affine function~\eqref{eq:power-difference}, and part~\ref{faces:edges} follows from part~\ref{faces:tri}.

\ref{faces:tri} The points of \(S\cap\partial O_T\) lie on a circle, so each of them is a vertex of their convex hull; thus \(F_T\) is a convex polygon with vertex set \(S\cap\partial O_T\).  If a triangle \(T'\in\cD\) is contained in \(F_T\), then its vertices are points of \(S\) in \(O_T\), and since the open disk of \(O_T\) contains no point of \(S\), they lie on \(\partial O_T\); hence \(\partial O_{T'}=\partial O_T\).  Conversely, if \(\partial O_{T'}=\partial O_T\), then the vertices of \(T'\) lie in \(S\cap\partial O_T\), so \(T'\subset F_T\).

It remains to show that these triangles cover \(F_T\).  Let \(x\) be an interior point of \(F_T\) that lies on no edge of \(\cD\).  Since the triangles of \(\cD\) cover the convex hull of \(S\), some \(T'\in\cD\) contains \(x\), and \(x\) is an interior point of \(T'\).  Let \(\Psi=\pow_{O_T}-\pow_{O_{T'}}\).  Every vertex \(z\) of \(F_T\) satisfies \(\pow_{O_T}(z)=0\) and \(\pow_{O_{T'}}(z)\ge0\), since \(z\in S\) and the open disk of \(O_{T'}\) contains no point of \(S\).  So \(\Psi\le0\) at the vertices of \(F_T\), and since \(\Psi\) is affine by~\eqref{eq:power-difference}, \(\Psi\le0\) on \(F_T\).  In the same way, \(\Psi\ge0\) on \(T'\).  Hence \(\Psi\) vanishes on \(F_T\cap T'\), which contains a neighborhood of \(x\).  An affine function that vanishes on an open set vanishes everywhere, so \(\pow_{O_T}=\pow_{O_{T'}}\), and the circles \(\partial O_T\) and \(\partial O_{T'}\), which are the zero sets of these two functions, coincide.  Thus \(T'\subset F_T\).  The points \(x\) considered are dense in \(F_T\), because the edges of \(\cD\) are finitely many segments, and the union of the triangles of \(\cD\) contained in \(F_T\) is closed.  Therefore this union is \(F_T\).  Finally, every triangle \(T\in\cD\) is contained in the face \(F_T\); since the triangles of \(\cD\) cover the convex hull of \(S\), so do the faces.

\ref{faces:edges} Let \(F\) be a face, and let \([x,y]\) be an edge of \(F\).  The open segment \((x,y)\) lies in the open disk of \(O_F\), so it contains no point of \(S\).  Let \(z\in(x,y)\).  By~\ref{faces:tri}, some triangle \(T'\in\cD\) with \(T'\subset F\) contains \(z\).  Since \(z\) is a boundary point of \(F\), it is a boundary point of \(T'\), so it lies on an edge \(E\) of \(T'\), and it is not an endpoint of \(E\), because the endpoints of \(E\) are in \(S\).  The face \(F\) lies on one side of the line through \(x\) and \(y\), and the segment \(E\subset F\) has a point of its relative interior on this line; hence \(E\) lies on this line, and therefore in the intersection of \(F\) with this line, which is the edge \([x,y]\).  The endpoints of \(E\) are points of \(S\) in \([x,y]\), so they are \(x\) and \(y\).  Hence \([x,y]=E\) is an edge of \(\cD\).

\ref{faces:meet} Let \(F\ne F'\) be faces that meet, let \(O=O_F\) and \(O'=O_{F'}\), and let \(\Psi=\pow_O-\pow_{O'}\).  Distinct faces have distinct circles, since \(F=\operatorname{conv}(S\cap\partial O_F)\).  As in the proof of~\ref{faces:tri}, \(\Psi\le0\) on \(F\) and \(\Psi\ge0\) on \(F'\), so \(\Psi\) vanishes on \(F\cap F'\ne\emptyset\).  In particular, \(\Psi\) is not a nonzero constant, and it is not identically zero, because \(\partial O\ne\partial O'\); so its zero set is a line \(\ell\), and \(F\) and \(F'\) lie in the two closed half-planes bounded by \(\ell\).  Hence \(F\cap\ell\) is a vertex or an edge of \(F\), and it is the convex hull of the vertices of \(F\) on \(\ell\).  A vertex \(z\) of \(F\) on \(\ell\) satisfies \(\pow_O(z)=0\) and \(\Psi(z)=0\), hence \(\pow_{O'}(z)=0\); thus \(z\in S\cap\partial O'\), and \(z\) is a vertex of \(F'\).  So \(F\cap\ell\) is the convex hull of \(S\cap\partial O\cap\partial O'\), and by symmetry so is \(F'\cap\ell\).  Therefore
\[
F\cap F'=(F\cap\ell)\cap(F'\cap\ell)=\operatorname{conv}\bigl(S\cap\partial O\cap\partial O'\bigr).
\]
Two distinct circles have at most two common points.  If \(S\cap\partial O\cap\partial O'\) is a single point, it is a common vertex of \(F\) and \(F'\); if it consists of two points, then \(F\cap F'=F\cap\ell=F'\cap\ell\) is a segment and hence an edge of both faces.
\end{proof}

We now follow a segment \([u,v]\) through the faces.

\begin{lemma}[Segment decomposition]\label{lem:segment}
Let \(u\ne v\) be points of \(S\) such that the open segment \((u,v)\) contains no point of \(S\), and suppose that \([u,v]\) is not an edge of \(\cD\).  Then there are pairwise distinct faces \(F_1,\dots,F_m\), where \(m\ge1\), and pairwise distinct points \(z_0=u,z_1,\dots,z_m=v\), in this order along \([u,v]\), such that the following hold:
\begin{enumerate}[label=(\arabic*)]
\item\label{seg:pieces} \([u,v]\cap F_i=[z_{i-1},z_i]\) for \(1\le i\le m\);
\item\label{seg:edges} for \(1\le i\le m-1\), the faces \(F_i\) and \(F_{i+1}\) have a common edge \(E_i\), the point \(z_i\) lies in the relative interior of \(E_i\), and \(E_i\) does not lie on the line through \(u\) and \(v\);
\item\label{seg:ends} \(u\) is a vertex of \(F_1\), \(v\) is a vertex of \(F_m\), and \(E_{i-1}\ne E_i\) for \(2\le i\le m-1\).
\end{enumerate}
\end{lemma}

\begin{proof}
We first show that every edge \(E\) of \(\cD\) meets \([u,v]\) in at most one point.  Otherwise \(E\) and \([u,v]\) lie on a common line and overlap in a segment of positive length.  Both segments have their endpoints in \(S\) and no point of \(S\) in their relative interiors.  For \([u,v]\), this is the assumption.  For \(E\), a point \(s\) of \(S\) in the relative interior of \(E\) would be a vertex of a triangle \(T\) of \(\cD\) other than the triangle \(T_E\) of \(E\), and \(T\cap T_E\) would contain \(s\), which is neither a vertex of \(T_E\) nor a point of an edge of \(T_E\) other than \(E\); so \(T\) and \(T_E\) would share the edge \(E\), and the vertex \(s\) of \(T\) would lie in the relative interior of the edge \(E\) of \(T\), which is impossible.  So neither segment contains an endpoint of the other in its relative interior.  Two collinear segments that overlap in positive length and satisfy this condition coincide, so \([u,v]=E\) would be an edge of \(\cD\), a contradiction.

By \cref{lem:faces}\ref{faces:tri}, the faces cover \([u,v]\), and since faces are convex, \([u,v]\) is the union of the segments \([u,v]\cap F\) over all faces \(F\).  Let \(F_1,\dots,F_m\) be the faces \(F\) for which \([u,v]\cap F\) has positive length.  For \(i\ne i'\), the segments \([u,v]\cap F_i\) and \([u,v]\cap F_{i'}\) have at most one common point: by \cref{lem:faces}\ref{faces:meet}, \(F_i\cap F_{i'}\) is empty, a vertex, or an edge, and an edge of a face is an edge of \(\cD\) by \cref{lem:faces}\ref{faces:edges}.  The union of these segments is closed and contains all but finitely many points of \([u,v]\), so it is \([u,v]\).  Numbering the faces in the order of their segments along \([u,v]\), we obtain the points \(z_0=u,z_1,\dots,z_m=v\) and~\ref{seg:pieces}.

For~\ref{seg:edges}, let \(1\le i\le m-1\).  The point \(z_i\) lies in \(F_i\cap F_{i+1}\) and in the open segment \((u,v)\), so it is not a point of \(S\), and in particular not a vertex of a face.  By \cref{lem:faces}\ref{faces:meet}, \(F_i\cap F_{i+1}\) is a common edge \(E_i\), and \(z_i\in E_i\) is not an endpoint of \(E_i\), because the endpoints are in \(S\).  By \cref{lem:faces}\ref{faces:edges}, \(E_i\) is an edge of \(\cD\), so it meets \([u,v]\) only in \(z_i\); since \(z_i\) lies in the relative interiors of both segments, \(E_i\) does not lie on the line through \(u\) and \(v\).

For~\ref{seg:ends}, the point \(u\in S\) lies in \(F_1\subset O_{F_1}\), and the open disk of \(O_{F_1}\) contains no point of \(S\); so \(u\in\partial O_{F_1}\), and \(u\) is a vertex of \(F_1\) by \cref{lem:faces}\ref{faces:tri}.  The same argument applies to \(v\).  Finally, for \(2\le i\le m-1\), the edges \(E_{i-1}\) and \(E_i\) contain the distinct points \(z_{i-1}\) and \(z_i\) of \([u,v]\), and each edge of \(\cD\) meets \([u,v]\) in at most one point; hence \(E_{i-1}\ne E_i\).
\end{proof}

The next lemma shows that the circumdisks of the faces \(F_1,\dots,F_m\) form a disk chain whose gates are the common edges \(E_i\), as in \cref{fig:delaunay}. Conditions~\ref{c:gate}--\ref{c:order} come from \cref{lem:segment}, and condition~\ref{c:rails} from the cyclic order of the vertices of each face on its circle.

\begin{lemma}[Generated chains]\label{lem:generated}
In the setting of \cref{lem:segment}, suppose that \(u\) and \(v\) lie on the \(x\)-axis and \(x(u)<x(v)\), and let \(O_i=O_{F_i}\) for \(1\le i\le m\).  Then \((O_1,\dots,O_m;u,v)\) is a disk chain, its gates are the edges \(E_1,\dots,E_{m-1}\), and each of its rails is an arc of a circle \(\partial O_i\) whose ends are vertices of \(F_i\).
\end{lemma}

\begin{proof}
Each \(O_i\) has positive radius, since \(F_i\) has at least three vertices on \(\partial O_i\).  By \cref{lem:segment}\ref{seg:ends}, \(u\in\partial O_1\) and \(v\in\partial O_m\), and \(x(u)<x(v)\) by assumption.  For \(m=1\), this proves that \((O_1;u,v)\) is a disk chain, and its two rails are the arcs of \(\partial O_1\) between the vertices \(u\) and \(v\) of \(F_1\).  Suppose now that \(m\ge2\).

\emph{Condition~\ref{c:gate}.}  Let \(1\le i\le m-1\).  The faces \(F_i\) and \(F_{i+1}\) are distinct, so their circles are distinct, and both circles pass through the two endpoints of the common edge \(E_i\).  Two distinct circles have at most two common points, so \(\partial O_i\) and \(\partial O_{i+1}\) meet exactly in the endpoints of \(E_i\).  By \cref{lem:segment}\ref{seg:edges}, the \(x\)-axis meets \(E_i\) in the point \(z_i\) of its relative interior and does not contain \(E_i\), so the endpoints of \(E_i\) lie strictly on opposite sides of the \(x\)-axis.  Hence~\ref{c:gate} holds, with \(G_i=E_i\) and \(c_i=z_i\).

\emph{Condition~\ref{c:forward}.}  Let \(1\le i\le m-1\) and \(\Psi=\pow_{O_i}-\pow_{O_{i+1}}\).  As in the proof of \cref{lem:faces}\ref{faces:tri}, \(\Psi\le0\) on \(F_i\) and \(\Psi\ge0\) on \(F_{i+1}\).  Since the two circles are distinct and meet, \(\Psi\) is a nonconstant affine function, and its zero set is the line through \(E_i\), which meets the \(x\)-axis only at \(c_i\).  The points of the \(x\)-axis slightly to the left of \(c_i\) lie in \([z_{i-1},z_i]\subset F_i\) but not on this line, so \(\Psi<0\) there; similarly, \(\Psi>0\) at the points slightly to the right of \(c_i\), which lie in \(F_{i+1}\).  By~\eqref{eq:power-difference}, the restriction of \(\Psi\) to the \(x\)-axis is affine with slope \(2\bigl(x(o_{i+1})-x(o_i)\bigr)\), where \(o_i\) is the center of \(O_i\).  This slope is positive, which proves~\ref{c:forward}.

\emph{Condition~\ref{c:order}.}  The points \(z_0=u,z_1,\dots,z_m=v\) are pairwise distinct and lie in this order on the \(x\)-axis, so \(x(u)<x(c_1)<\dots<x(c_{m-1})<x(v)\).

\emph{Condition~\ref{c:rails}.}  What remains is condition~\ref{c:rails}, which needs most of the work; \cref{fig:face} shows the configuration.  Let \(2\le i\le m-1\), and set \(O=O_i\), \(F=F_i\) and \(w_\pm=w_\pm(O)\).  For an edge \(E\) of \(F\), the line through \(E\) bounds a closed half-plane whose interior is disjoint from \(F\); we call the arc of \(\partial O\) that joins the endpoints of \(E\) and lies in this half-plane the \emph{outer arc} of \(E\).  Since \(F\) is inscribed in \(\partial O\), the outer arcs of the edges of \(F\) are the arcs between cyclically consecutive vertices of \(F\).  So an outer arc contains no vertex of \(F\) other than its ends, and the outer arcs of two distinct edges meet at most in a common end.

\begin{figure}[!htbp]
\centering
\begin{tikzpicture}[paperfig]
\path[use as bounding box] (-5.556,-2.15) rectangle (5.156,3.2);
\fill[figFill] (-1.951,1.669) -- (-0.08,2.749) -- (1.654,2.048) -- (1.992,-0.7) -- (0.24,-1.837) -- (-1.762,-1.028) -- cycle;
\draw[fdisk] (0,0.45) circle[radius=2.3];
\draw[faxis,-{Latex[length=1.8mm]}] (-2.956,0) -- (3.156,0);
\draw[ftri] (-1.951,1.669) -- (-0.08,2.749);
\draw[ftri] (-0.08,2.749) -- (1.654,2.048);
\draw[ftri] (1.992,-0.7) -- (0.24,-1.837);
\draw[ftri] (0.24,-1.837) -- (-1.762,-1.028);
\draw[fgate] (-1.951,1.669) -- (-1.762,-1.028);
\draw[fgate] (1.654,2.048) -- (1.992,-0.7);
\draw[frail] (-1.951,1.669) arc[start angle=148,end angle=220,radius=2.3];
\draw[frail] (1.992,-0.7) arc[start angle=-30,end angle=44,radius=2.3];
\node[fpt] at (-1.951,1.669) {};
\node[fpt] at (-0.08,2.749) {};
\node[fpt] at (1.654,2.048) {};
\node[fpt] at (1.992,-0.7) {};
\node[fpt] at (0.24,-1.837) {};
\node[fpt] at (-1.762,-1.028) {};
\node[fpt] at (-2.256,0) {};
\node[fpt] at (2.256,0) {};
\node[fpto] at (-1.834,0) {};
\node[fpto] at (1.906,0) {};
\node[anchor=south east] at (-2.031,1.729) {$a_{i-1}$};
\node[anchor=south west] at (1.734,2.108) {$a_i$};
\node[anchor=north west] at (2.072,-0.76) {$b_i$};
\node[anchor=north east] at (-1.842,-1.088) {$b_{i-1}$};
\node[anchor=south east] at (-2.376,0.03) {$w_-$};
\node[anchor=south west] at (2.376,0.03) {$w_+$};
\node[anchor=south west] at (-1.70,0.16) {$c_{i-1}$};
\node[anchor=south east] at (1.76,0.1) {$c_i$};
\node[anchor=west,text=figOrangeText] at (-1.60,0.80) {$E_{i-1}$};
\node[anchor=east,text=figOrangeText] at (1.683,1.024) {$E_i$};
\node[text=black] at (-0.1,0.6) {$F$};
\node[anchor=east,text=figBlueText] at (-2.45,-0.72) {outer arc of $E_{i-1}$};
\node[anchor=west,text=figBlueText] at (2.306,-0.62) {outer arc of $E_i$};
\node[anchor=south west,text=black] at (1.035,2.631) {$\partial O$};
\end{tikzpicture}\caption{The face \(F=F_i\) is inscribed in the circle \(\partial O\) of \(O=O_i\).  The gates \(E_{i-1}\) and \(E_i\) are edges of \(F\).  Their outer arcs are shown in blue and lie on the side of each edge away from \(F\).  These arcs contain \(w_-\) and \(w_+\), respectively, and no vertex of \(F\) other than their endpoints.  Thus \(a_{i-1}\) and \(a_i\) are the first and last vertices of \(F\) on the upper arc of \(O\), and \(b_{i-1}\) and \(b_i\) are the first and last vertices on the lower arc.}
\label{fig:face}
\end{figure}

We first show that \(w_-\) lies in the outer arc of \(E_{i-1}\) and is not one of its ends.  The line through \(E_{i-1}\) meets the \(x\)-axis only at \(c_{i-1}\).  The points of the \(x\)-axis slightly to the right of \(c_{i-1}\) lie in \(F\) by \cref{lem:segment}\ref{seg:pieces}, and not on this line; hence all points of the \(x\)-axis to the left of \(c_{i-1}\) lie strictly on the side of this line that does not contain \(F\).  Since \(c_{i-1}\) lies in the relative interior of the chord \(E_{i-1}\) of \(\partial O\), it is an interior point of \(O\), so \(x(w_-)<x(c_{i-1})\).  This proves the claim.  In the same way, \(w_+\) lies in the outer arc of \(E_i\) and is not one of its ends.  Since \(E_{i-1}\ne E_i\) by \cref{lem:segment}\ref{seg:ends}, the outer arcs of \(E_{i-1}\) and \(E_i\) meet at most in a common end, and \(w_+\) is a point of the outer arc of \(E_i\) other than its ends; hence \(w_+\) does not lie in the outer arc of \(E_{i-1}\).  Likewise, \(w_-\) does not lie in the outer arc of \(E_i\).

Now remove \(w_-\) from the outer arc of \(E_{i-1}\).  The remaining part that ends at \(a_{i-1}\) does not meet the \(x\)-axis, because \(\partial O\) meets the \(x\)-axis only at \(w_-\) and \(w_+\), and \(w_+\) is not in the outer arc of \(E_{i-1}\); since \(y(a_{i-1})>0\), the remaining part that ends at \(a_{i-1}\) lies in the upper half-plane.  Hence the arc from \(w_-\) to \(a_{i-1}\) is an initial part of the upper arc of \(O\), and it contains no vertex of \(F\) other than \(a_{i-1}\).  So \(a_{i-1}\) is the first vertex of \(F\) on the upper arc (\cref{fig:face}).  The same argument with the outer arc of \(E_i\) shows that \(a_i\) is the last vertex of \(F\) on the upper arc.  Therefore \(a_{i-1}\) does not come after \(a_i\).  In the same way, \(b_{i-1}\) and \(b_i\) are the first and the last vertex of \(F\) on the lower arc, so \(b_{i-1}\) does not come after \(b_i\).

Finally, for \(1\le i\le m\), the ends of the rails of \(O_i\) are \(a_{i-1},b_{i-1},a_i,b_i\), where \(a_0=b_0=u\) and \(a_m=b_m=v\).  Each of them is \(u\) (for \(i=1\)), \(v\) (for \(i=m\)), or an endpoint of \(E_{i-1}\) or \(E_i\), and hence a vertex of \(F_i\).
\end{proof}

It remains to compare the ladder graph of this chain with \(\cD\): every rail can be replaced by a path in \(\cD\) that is not longer.

\begin{lemma}[From rails to paths in \(\cD\)]\label{lem:rails-to-paths}
For the disk chain of \cref{lem:generated}, \(d_\cD(u,v)\le P(u,v)\).
\end{lemma}

\begin{proof}
By \cref{lem:generated} and \cref{lem:faces}\ref{faces:edges}, every gate is an edge of \(\cD\), and its weight in the ladder graph is its length.  Consider a rail of \(O_i\), and let \(r_i\) be the radius of \(O_i\).  By \cref{lem:generated}, the rail is an arc of \(\partial O_i\), with some central angle \(\eta\ge0\), whose ends are vertices of \(F_i\).  List the vertices of \(F_i\) on this arc in order as \(x_0,x_1,\dots,x_t\), where \(t\ge0\) and \(x_0,x_t\) are the ends of the rail.  For \(1\le l\le t\), there is no vertex of \(F_i\) on the arc strictly between \(x_{l-1}\) and \(x_l\), so \([x_{l-1},x_l]\) is an edge of \(F_i\), and hence of \(\cD\) by \cref{lem:faces}\ref{faces:edges}.  If \(\eta_l\) is the central angle of the arc from \(x_{l-1}\) to \(x_l\), then \(\eta_1+\dots+\eta_t=\eta\) and
\[
\norm{x_{l-1}-x_l}=2r_i\sin\frac{\eta_l}{2}\le r_i\eta_l .
\]
Summing over \(l\), the path \(x_0x_1\cdots x_t\) in \(\cD\) is not longer than the rail.  Replacing every rail of a shortest path from \(u\) to \(v\) in the ladder graph by such a path, and keeping its gates, we obtain a path from \(u\) to \(v\) in \(\cD\) of length at most \(P(u,v)\).
\end{proof}

\begin{proof}[Proof of \cref{thm:reduction}]
Suppose that the chain inequality holds at level \(U\).  By \cref{lem:one-disk}, \(U\ge\pi/2>1\).  Let \(S\), \(\cD\) and \(u,v\in S\) be as in the theorem, and assume \(u\ne v\).  Let \(u=p_0,p_1,\dots,p_N=v\) be the points of \(S\) on \([u,v]\), in this order.  Since \(d_\cD\) satisfies the triangle inequality and \(\norm{u-v}=\sum_{j=1}^N\norm{p_{j-1}-p_j}\), it suffices to prove the claim for each pair \(p_{j-1},p_j\).  So we may assume that the open segment \((u,v)\) contains no point of \(S\).  If \([u,v]\) is an edge of \(\cD\), then \(d_\cD(u,v)\le\norm{u-v}\le U\norm{u-v}\).  Otherwise, since rigid motions preserve lengths and map Delaunay triangulations to Delaunay triangulations, we may assume that \(u\) and \(v\) lie on the \(x\)-axis and \(x(u)<x(v)\).  By \cref{lem:segment,lem:generated}, the circumdisks of the faces that \([u,v]\) crosses form a disk chain \((O_1,\dots,O_m;u,v)\).  By \cref{lem:rails-to-paths} and the chain inequality, \(d_\cD(u,v)\le P(u,v)\le U\norm{u-v}\).
\end{proof}

\section{The Potential Criterion Is Exact}\label{app:exactness}

In this appendix, we prove \cref{prop:exact}, the converse of the potential criterion (\cref{thm:criterion}): if the chain inequality holds at a level \(U\ge\pi/2\), then a feasible potential of level \(U\) exists, and the value function \(W^*_U\) defined below is the smallest one (\cref{thm:exactness}).  The proposition is not used in the proof of \cref{thm:main}; it shows that the potential criterion loses nothing, so that any gap between \(U_0\) and the worst-case stretch factor comes from the other steps, as discussed in \cref{sec:conclusion}.

Roughly speaking, a continuation from a gate state \(s\) is the rest of a chain after \(s\): a sequence of next disks, taken one step at a time, followed by the final step.  Its payoff is the sum of \(\text{cost}-U\cdot\text{progress}\) over these moves, and the value function at \(s\) is the largest payoff of a continuation from \(s\).  Adding (S) and (F) along the moves of a continuation from \(s\) shows that a feasible potential \(W(s)\) is an upper bound for its payoff; so the value function is the natural candidate for the smallest feasible potential, and \cref{thm:exactness} shows that it is one whenever the chain inequality holds.  The value function satisfies (S) and (F) because a best continuation either ends at once or takes a step and then continues in the best possible way, and it satisfies (I) because every chain obeys the chain inequality.

The only delicate point is that the value function is finite.  To bound it, we put a first disk in front of a continuation, so that the chain inequality applies; this is possible up to a change of the difference \(\delta\), which changes the payoff by at most half as much.  We first define the value function and then present two lemmas.  Specifically, \cref{lem:payoff} shows that the payoff of a continuation changes by at most half the change of the difference \(\delta\), and that it splits after the first step; \cref{lem:past} shows that every gate state can be reached by an initial pair, up to its third component, and that an initial pair followed by a continuation forms a disk chain.

\begin{definition}[Continuations, payoff and value function]\label{def:continuation}
Let \(s\) be a gate state.  A \emph{continuation} from \(s\) is a finite sequence \(\Gamma=(O^{(1)},\dots,O^{(m)})\) of closed disks, where \(m\ge0\), for which there are gate states \(s_0=s,s_1,\dots,s_m\) such that, for \(1\le l\le m\), the disk \(O^{(l)}\) is a next disk for \(s_{l-1}\) and the step from \(s_{l-1}\) to \(O^{(l)}\) leads to \(s_l\); we call \(s_0,\dots,s_m\) the gate states of \(\Gamma\) from \(s\) (they are determined by \(s\) and \(\Gamma\)).  For \(U\in\R\), the \emph{payoff} \(J_U(\Gamma;s)\) is the sum of \(\text{cost}-U\cdot\text{progress}\) over the \(m\) steps of \(\Gamma\) and the final step at \(s_m\).  The \emph{value function of level \(U\)} is
\[
W^*_U(s)=\sup_\Gamma J_U(\Gamma;s)\in(-\infty,+\infty],
\]
where the supremum is over all continuations \(\Gamma\) from \(s\); the empty continuation, with \(m=0\), shows that it is over a nonempty set, so \(W^*_U(s)>-\infty\).
\end{definition}

\begin{lemma}[Two properties of the payoff]\label{lem:payoff}
Let \(U\in\R\).
\begin{enumerate}[label=(\arabic*)]
\item\label{payoff:lip} Let \(s=(O,G,\delta)\) and \(\tilde s=(O,G,\tilde\delta)\) be gate states.  Then every continuation \(\Gamma\) from \(s\) is a continuation from \(\tilde s\), and \(\abs{J_U(\Gamma;s)-J_U(\Gamma;\tilde s)}\le\frac12\abs{\delta-\tilde\delta}\).
\item\label{payoff:concat} Let \(O'\) be a next disk for a gate state \(s\), and let the step from \(s\) to \(O'\) lead to \(s'\), with cost \(\Delta\mu\) and progress \(\Delta x\).  If \(\Gamma'=(O^{(1)},\dots,O^{(m)})\) is a continuation from \(s'\), then \(\Gamma=(O',O^{(1)},\dots,O^{(m)})\) is a continuation from \(s\), and \(J_U(\Gamma;s)=\Delta\mu-U\,\Delta x+J_U(\Gamma';s')\).  Every continuation from \(s\) with at least one disk arises in this way.
\end{enumerate}
\end{lemma}

\begin{proof}
\ref{payoff:lip} By \cref{def:moves}\ref{mv:step}, whether a disk is a next disk for a gate state depends only on the disk and the gate of that gate state, not on its third component, and so do the disk and the gate of the gate state to which the step leads.  By induction on \(l\), the disks \(O^{(1)},\dots,O^{(m)}\) of \(\Gamma\) form a continuation from \(\tilde s\), and the gate states \(s_l\) and \(\tilde s_l\) of \(\Gamma\) from \(s\) and from \(\tilde s\) differ only in their third components \(\delta_l\) and \(\tilde\delta_l\).  In particular, the corresponding moves have the same progress and the same lengths \(\alpha,\beta,\alpha_{\rm end},\beta_{\rm end}\), so it suffices to compare the sums of the costs.

In this proof, let \(\mu_l\) be the sum of the costs of the first \(l\) steps of \(\Gamma\) from \(s\), with \(\mu_0=0\), and let \(d_l^\pm=\mu_l\pm\frac12\delta_l\); define \(\tilde\mu_l\) and \(\tilde d_l^\pm\) in the same way for \(\tilde s\).  Unlike the numbers of \cref{def:prefix}, these numbers are counted from the gate state \(s\).  Consider the step from \(s_l\) to \(O^{(l+1)}\), with lengths \(\alpha,\beta\) and new gate \(G'\) as in \cref{def:moves}\ref{mv:step}.  By the two identities~\eqref{eq:minclip}, applied with \(A=d_l^++\alpha\), \(B=d_l^-+\beta\) and \(h=\abs{G'}\), the numbers \(\min(A,B+h)\) and \(\min(B,A+h)\) have difference \(\clip_{\abs{G'}}(\delta_l+\alpha-\beta)=\delta_{l+1}\) and sum \(2\mu_l+\alpha+\beta-\omega=2\mu_{l+1}\), by the formulas for \(\delta'\) and \(\Delta\mu\) in \cref{def:moves}\ref{mv:step}; so these two minima are \(\mu_{l+1}\pm\frac12\delta_{l+1}=d_{l+1}^\pm\), that is,
\[
d_{l+1}^+=\min\bigl(d_l^++\alpha,\ d_l^-+\beta+\abs{G'}\bigr),\qquad
 d_{l+1}^-=\min\bigl(d_l^-+\beta,\ d_l^++\alpha+\abs{G'}\bigr),
\]
and the same formulas hold for \(\tilde d^\pm\).  Moreover, the sum of the costs of all moves of \(\Gamma\) is \(\mu_m+\kappa=\min(d_m^++\alpha_{\rm end},\,d_m^-+\beta_{\rm end})\), where \(\kappa\) is the cost of the final step at \(s_m\), and similarly for \(\tilde s\).  For all real numbers \(A,B,\tilde A,\tilde B,t_1,t_2\), we have \(\abs{\min(A+t_1,B+t_2)-\min(\tilde A+t_1,\tilde B+t_2)}\le\max(\abs{A-\tilde A},\abs{B-\tilde B})\), since the minimum of two numbers changes by at most the larger of the changes of the two.  Hence \(\max(\abs{d_l^+-\tilde d_l^+},\abs{d_l^--\tilde d_l^-})\) does not increase with \(l\), and the two sums of costs differ by at most its initial value \(\frac12\abs{\delta-\tilde\delta}\), which is the value for \(l=0\), where \(d_0^\pm=\pm\frac12\delta\) and \(\tilde d_0^\pm=\pm\frac12\tilde\delta\).

\ref{payoff:concat} The steps of \(\Gamma\) are the step from \(s\) to \(O'\), followed by the steps of \(\Gamma'\) from \(s'\); so \(\Gamma\) is a continuation from \(s\), and its moves are the step from \(s\) and the moves of \(\Gamma'\).  Conversely, if \(\Gamma=(O^{(1)},\dots,O^{(m)})\) is a continuation from \(s\) with \(m\ge1\), then \(O^{(1)}\) is a next disk for \(s\), and \((O^{(2)},\dots,O^{(m)})\) is a continuation from \(s_1\).
\end{proof}

The next lemma connects continuations with disk chains.

\begin{lemma}[Chains through a gate state]\label{lem:past}
\leavevmode
\begin{enumerate}[label=(\arabic*)]
\item\label{past:exists} For every gate state \(s=(O,G,\delta)\), there is a closed disk \(O'\) such that \((O',O)\) is an initial pair with gate \(G\); it leads to a gate state of the form \((O,G,\tilde\delta)\).
\item\label{past:chain} Let \((O_1,O_2)\) be an initial pair leading to the gate state \(s_1\), with \(u=w_-(O_1)\), \(c_1\) and \(\mu_1\) as in \cref{def:moves}\ref{mv:initial}, and let \(\Gamma=(O^{(1)},\dots,O^{(m)})\) be a continuation from \(s_1\).  Let \(v=w_+(O^{(m)})\) if \(m\ge1\), and \(v=w_+(O_2)\) if \(m=0\).  Then \((O_1,O_2,O^{(1)},\dots,O^{(m)};u,v)\) is a disk chain, and for every \(U\in\R\),
\[
P(u,v)-U\norm{u-v}=\mu_1-U\bigl(x(c_1)-x(u)\bigr)+J_U(\Gamma;s_1).
\]
\end{enumerate}
\end{lemma}

\begin{proof}
\ref{past:exists} Let \(G=[a_G,b_G]\).  By \cref{lem:pencil}\ref{pencil:family}, there is a closed disk \(O'\) whose boundary meets \(\partial O\) exactly in \(a_G\) and \(b_G\) and whose center has a smaller \(x\)-coordinate than the center of \(O\).  Since \(y(a_G)>0>y(b_G)\) by the definition of a gate state, \((O',O)\) is an initial pair with gate \(G_1=G\), and by \cref{def:moves}\ref{mv:initial} it leads to the gate state \((O,G,\delta_1)\).

\ref{past:chain} Let \(n=m+2\), and set \(O_{l+2}=O^{(l)}\) for \(1\le l\le m\).  Let \(s_1=(O_2,G_1,\delta_1)\), and for \(1\le l\le m\), let \(s_{l+1}=(O_{l+2},G_{l+1},\delta_{l+1})\) be the gate state to which the step from \(s_l\) to \(O^{(l)}\) leads; so \(s_1,\dots,s_{m+1}\) are the gate states of \(\Gamma\) from \(s_1\) of \cref{def:continuation}, where the gate state numbered \(l\) there is numbered \(l+1\) here.  We check \cref{def:chain}.  Conditions~\ref{c:gate} and~\ref{c:forward} for \(i=1\) are conditions of an initial pair, and for \(i=l+1\) with \(1\le l\le m\), they are conditions~\ref{n:gate} and~\ref{n:forward} of a next disk for \(s_l\).  For~\ref{c:order}, the point \(c_1\) lies in the relative interior of the chord \(G_1\) of \(\partial O_1\), so it is an interior point of \(O_1\) and \(x(u)<x(c_1)\); condition~\ref{n:order} of a next disk for \(s_l\) gives \(x(c_l)<x(c_{l+1})\) for \(1\le l\le m\); and \(c_{n-1}\) is an interior point of \(O_n\), so \(x(c_{n-1})<x(w_+(O_n))=x(v)\).  In particular, \(u\in\partial O_1\), \(v\in\partial O_n\) and \(x(u)<x(v)\).  Condition~\ref{c:rails} for the disk \(O_{l+1}\), where \(1\le l\le m\), is condition~\ref{n:rails} of a next disk for \(s_l\).  Hence we have a disk chain.

By \cref{lem:moves}\ref{acc:moves}, the moves of this chain are the initial pair \((O_1,O_2)\), the steps from the gate states \(s_1,\dots,s_m\) of \(\Gamma\) to the disks \(O^{(1)},\dots,O^{(m)}\), and the final step at \(s_{m+1}\); indeed, a move is determined by the current gate state and the next disk, so the gate states of \cref{lem:moves} are those of \(\Gamma\).  By \cref{lem:moves}\ref{acc:sums}, \(P(u,v)-U\norm{u-v}\) is the sum of \(\text{cost}-U\cdot\text{progress}\) over these moves, which is \(\mu_1-U(x(c_1)-x(u))+J_U(\Gamma;s_1)\).
\end{proof}

\begin{theorem}[Exactness]\label{thm:exactness}
Let \(U\ge\pi/2\).  Then the chain inequality (\cref{def:ladder}) holds at level \(U\) if and only if there is a feasible potential of level \(U\) (\cref{def:feasible}).  In this case, the value function \(W^*_U\) is finite at every gate state and satisfies
\begin{equation}\label{eq:value}
W^*_U(s)=\max\Bigl(\kappa-U\bigl(x(w_+(O))-x(c(s))\bigr),\ \sup_{O'}\bigl[\Delta\mu-U\,\Delta x+W^*_U(s')\bigr]\Bigr)
\end{equation}
for every gate state \(s=(O,G,\delta)\), where \(\kappa\) is the cost of the final step at \(s\), the supremum is over all next disks \(O'\) for \(s\), and \(\Delta\mu\), \(\Delta x\) and \(s'\) are the cost, the progress and the gate state of the step from \(s\) to \(O'\).  Moreover, \(W^*_U\) is a feasible potential of level \(U\), and \(W^*_U\le W\) for every feasible potential \(W\) of level \(U\).
\end{theorem}

The statements about \(W^*_U\) in \cref{thm:exactness} are \cref{prop:exact}.  Equation~\eqref{eq:value} is the dynamic programming equation of the problem: a best continuation from \(s\) either ends with the final step at once, or takes a step and then continues in the best possible way.

\begin{proof}
If there is a feasible potential of level \(U\), then the chain inequality holds at level \(U\) by \cref{thm:criterion}.  Conversely, suppose that the chain inequality holds at level \(U\).  We show the remaining claims in four parts.

\emph{Finiteness.}  Let \(s=(O,G,\delta)\) be a gate state.  By \cref{lem:past}\ref{past:exists}, there is an initial pair \((O',O)\) leading to a gate state \(\tilde s=(O,G,\tilde\delta)\); let \(\mu'\) be its cost and \(L'\) its progress.  Let \(\Gamma\) be a continuation from \(s\).  By \cref{lem:payoff}\ref{payoff:lip}, \(\Gamma\) is a continuation from \(\tilde s\), and by \cref{lem:past}\ref{past:chain} and the chain inequality, \(\mu'-UL'+J_U(\Gamma;\tilde s)\le0\).  Hence, by \cref{lem:payoff}\ref{payoff:lip},
\[
J_U(\Gamma;s)\le J_U(\Gamma;\tilde s)+\tfrac12\abs{\delta-\tilde\delta}\le UL'-\mu'+\tfrac12\abs{\delta-\tilde\delta}.
\]
The right-hand side does not depend on \(\Gamma\), so \(W^*_U(s)<\infty\).

\emph{The value equation.}  By \cref{lem:payoff}\ref{payoff:concat}, a continuation from \(s\) is either empty, with payoff \(\kappa-U(x(w_+(O))-x(c(s)))\), or consists of a next disk \(O'\) for \(s\) followed by a continuation \(\Gamma'\) from the gate state \(s'\) to which the step from \(s\) to \(O'\) leads, with payoff \(\Delta\mu-U\,\Delta x+J_U(\Gamma';s')\); and every such combination is a continuation from \(s\).  Taking the supremum over \(\Gamma'\) for each \(O'\), then over \(O'\), and comparing with the first kind gives~\eqref{eq:value}.

\emph{Feasibility.}  Inequalities (S) and (F) of \cref{def:feasible} follow from~\eqref{eq:value}, since \(W^*_U(s)\) is at least each of the terms in the maximum.  For (I), let \((O_1,O_2)\) be an initial pair leading to \(s_1\), with cost \(\mu_1\) and progress \(x(c_1)-x(u)\).  For every continuation \(\Gamma\) from \(s_1\), \cref{lem:past}\ref{past:chain} and the chain inequality give \(\mu_1-U(x(c_1)-x(u))+J_U(\Gamma;s_1)\le0\).  Taking the supremum over \(\Gamma\) gives \(\mu_1-U(x(c_1)-x(u))\le-W^*_U(s_1)\), which is (I).

\emph{Minimality.}  Let \(W\) be a feasible potential of level \(U\), and let \(\Gamma\) be a continuation from \(s\), with gate states \(s_0=s,s_1,\dots,s_m\).  Adding (S) for the \(m\) steps of \(\Gamma\) and (F) for the final step at \(s_m\), we obtain
\[
J_U(\Gamma;s)\le\sum_{l=1}^{m}\bigl(W(s_{l-1})-W(s_l)\bigr)+W(s_m)=W(s).
\]
Taking the supremum over \(\Gamma\) gives \(W^*_U(s)\le W(s)\).
\end{proof}

\section{The Computer-Assisted Verification}\label{app:computation}

In this appendix, we describe the computations behind \cref{prop:h0h3,prop:h4}: the data, the arithmetic, the enclosures, the checks, and how they can be reproduced.  The proofs in \cref{sec:family,sec:verification} use only the list of inequalities in \cref{prop:h0h3} and the acceptance conditions~\ref{acc:empty}--\ref{acc:corner-outer} of \cref{subsec:cover}; everything that turns these into a verified claim is collected here.  Every inequality in these propositions is verified in exact rational arithmetic or in interval arithmetic with outward rounding, using the software named in \cref{subsec:arithmetic}.  The data, the programs and the versions of the software used are publicly archived~\cite{mcd2026package}; \cref{subsec:reproduction} explains how to run them.

\subsection{Data}\label{subsec:data}

The data of the verification consist of the coefficients of \(\Hs\) and the boxes of the two covers of \cref{subsec:cover}.  From these data, the programs compute the \emph{certificates} of \cref{prop:h0h3,prop:h4}, that is, the records of what was checked for each proposition and of the bounds obtained (\cref{tab:computations}).  The coefficient file stores \(U_0\) and \(3{,}150\) coefficients as exact rational numbers (\cref{tab:coefficients}); the knots, the degrees and the formulas of \cref{def:family} are implemented in the programs and restated in the documentation of the archive.  The main cover is stored as its \(1{,}037{,}769\) bisections, each with the coordinate that it halves, and its \(1{,}037{,}865\) boxes, each with its piece, whether it is accepted, and the check to try first; the supplementary cover is stored as its boxes with their checks.  Every run of the programs computes the certificates afresh; the archive also contains those of a complete run.

The coefficients that are not stored follow from~\eqref{eq:sym}, and conditions~\eqref{eq:gamma-zero} are checked (\cref{subsec:cert-h0h3}).  The file stores \(V_{ijl}\) and \(R_{ijl}\) for one index pair from each orbit of the maps \((i,j)\mapsto(11-i,11-j)\) and \((i,j)\mapsto(5-i,5-j)\) of~\eqref{eq:sym}, except for the orbits on which \(R_{ijl}=0\), and all \(B_l\) and \(D_l\); \cref{tab:coefficients} gives the counts.  Conditions~\eqref{eq:gamma-zero} concern disjoint sets of stored coefficients and are therefore independent, so the coefficients of the members of \(\cH\) form a linear space of dimension \(3{,}150-48=3{,}102\).

\begin{table}[!htbp]
\centering\small
\begin{tabular}{@{}lrrrrr@{}}
\toprule
 & Orbits stored & Coefficients & Zero conditions & Pair conditions & Conditions\\
\midrule
\(V_{ijl}\), \(0\le i,j\le11\) & \(72\) & \(2{,}520\) & \(6\) (antidiagonal) & \(30\) & \(36\)\\
\(R_{ijl}\), \(0\le i,j\le5\) & \(16\) & \(560\) & \(4\) (diagonal, antidiagonal) & \(6\) & \(10\)\\
\(B_l\) & \(1\) & \(35\) & \(1\) & \(0\) & \(1\)\\
\(D_l\) & \(1\) & \(35\) & \(1\) & \(0\) & \(1\)\\
\midrule
 Total & \(90\) & \(3{,}150\) & \(12\) & \(36\) & \(48\)\\
\bottomrule
\end{tabular}
\caption{The stored coefficients and the conditions~\eqref{eq:gamma-zero}.  An orbit is a pair \(\{(i,j),(11-i,11-j)\}\) of indices of the coefficients \(V_{ijl}\), or a pair \(\{(i,j),(5-i,5-j)\}\) of indices of the coefficients \(R_{ijl}\), whose coefficients are related by~\eqref{eq:sym}; each orbit contributes \(35\) coefficients, one for each B-spline \(N_l\).  A ``zero'' condition sets the coefficient with \(l=0\) of an orbit to \(0\); a ``pair'' condition relates the coefficients with \(l=0\) of two orbits exchanged by \((i,j)\mapsto(j,i)\).}
\label{tab:coefficients}
\end{table}

\subsection{Arithmetic and Computations}\label{subsec:arithmetic}

In this subsection, we name the software, describe the arithmetic, and list the computations that the proof uses.

\paragraph{Arithmetic.}  Every inequality in \cref{prop:h0h3,prop:h4} is verified in exact rational arithmetic or in interval arithmetic with outward rounding~\cite{moore2009}.  Interval computations use the ball arithmetic of Arb~\cite{arb} through python-flint~\cite{pythonflint}, except for the one-dimensional checks of \cref{prop:h0h3}\ref{h03:h1} and~\ref{h03:h3}, which use the interval arithmetic of mpmath~\cite{mpmath}; the inequalities that the proof uses need only its arithmetic operations, including integer powers, and its enclosure of \(\pi\) (\cref{subsec:cert-h0h3}).  Rational arithmetic uses Python's exact fractions and the rational numbers of FLINT; the library numpy is used only for arrays of exact rational numbers and integers.  Floating-point arithmetic is used only for heuristic choices, such as the multipliers \(\lambda\) of \cref{subsec:cert-h4}, which are rounded to rational numbers before they are used.  As with any computer-assisted proof, the verification relies on the correctness of this software and of our programs, which implement the checks of \cref{subsec:enclosures,subsec:cert-h0h3,subsec:cert-h4} and are publicly archived~\cite{mcd2026package}.

\paragraph{The Computations.}  \Cref{tab:computations} lists the computations that the proof uses; they are the four commands of \cref{subsec:reproduction}, and each of them computes its bounds afresh from the data.  The boxes of the main cover were produced by an earlier computation, whose own acceptances and bounds are not used.  The data mark as accepted the boxes whose acceptance a later verification with \(90\) digits confirmed, that is, all boxes accepted by the earlier computation except the box in \(Q_0\), and the second command verifies all of them again with \(90\) digits.  The other boxes of the main cover, including the box in \(Q_0\), are covered by the supplementary cover: the third command checks the containments of \cref{subsec:cover} in exact rational arithmetic and verifies the boxes of the supplementary cover with \(150\) digits.  The first command verifies the inequalities of \cref{prop:h0h3}, with \(70\) digits where it uses interval arithmetic, and repeats these checks with \(90\) digits.

\begin{table}[!htbp]
\centering\small
\begin{tabularx}{\textwidth}{@{}P{2.7cm}YP{2.5cm}P{2.6cm}@{}}
\toprule
 Command & What it checks & Arithmetic & Output\\
\midrule
(1) \Cref{prop:h0h3} & the stored coefficients against~\eqref{eq:sym} and~\eqref{eq:gamma-zero}; inequalities \ref{h03:R} and \ref{h03:h2} in exact arithmetic, \ref{h03:h1} and \ref{h03:h3} in interval arithmetic (\cref{subsec:cert-h0h3}) & exact; \(70\) digits, repeated with \(90\) & certificate of \cref{prop:h0h3}\\
(2) Main cover & rebuilds the main cover from its bisections and verifies the \(945{,}810\) boxes marked as accepted & \(90\) digits & certificate of the main cover\\
(3) Supplementary cover & the containments of the other \(92{,}055\) boxes (\cref{subsec:cover}); the \(1{,}499\) boxes~\eqref{eq:corner-boxes} and the two halves of \(Q_0\); the residual at \(64\) points, as a check of the implementation only & containments exact; \(150\) digits & certificate of the supplementary cover\\
(4) Summary & that the three certificates belong to the same data, programs and level, that the two covers leave no box unaccepted, and that \(U_0>\frac\pi2\) & exact; interval arithmetic for \(\pi\) & summary of the verification\\
\bottomrule
\end{tabularx}
\caption{The computations behind \cref{prop:h0h3,prop:h4}, that is, the four commands of \cref{subsec:reproduction}.  Precisions are working precisions in decimal digits; ``exact'' means exact rational arithmetic.  The box \(Q_0\) and the boxes~\eqref{eq:corner-boxes} are those of \cref{tab:cover}.}
\label{tab:computations}
\end{table}

\subsection{Enclosures}\label{subsec:enclosures}

Besides direct interval evaluation, the programs use the following elementary facts: parts~\ref{enc:taylor} and~\ref{enc:series} give enclosures of functions and their derivatives, part~\ref{enc:bernstein} bounds polynomials by their Bernstein coefficients, and parts~\ref{enc:endpoint} and~\ref{enc:quadratic} reduce a maximum over a set to finitely many values.  For an interval \(I\) and a real number \(w\ge0\), we denote the products of intervals by \(I\,[-w,w]\) and \(I\,[0,w]\).

\begin{lemma}[Enclosures]\label{lem:enclosures}
\leavevmode
\begin{enumerate}[label=(\alph*)]
\item\label{enc:taylor} Let \(f\) be twice continuously differentiable on an open set containing a box \(X=\prod_{i=1}^n[m_i-w_i,m_i+w_i]\), and let \(I_i\) and \(I_{ij}=I_{ji}\) be intervals containing \(\partial_if(X)\) and \(\partial_i\partial_jf(X)\), respectively.  Then, for every \(q\in X\), with \(h=q-m\), there are numbers \(g_i\in I_i\) and \(g_{ij}=g_{ji}\in I_{ij}\) such that
\[
f(q)=f(m)+\sum_ig_ih_i=f(m)+\sum_i\partial_if(m)\,h_i+\frac12\sum_ig_{ii}h_i^2+\sum_{i<j}g_{ij}h_ih_j .
\]
In particular,
\begin{gather*}
f(q)\in f(m)+\sum_iI_i\,[-w_i,w_i],\\
f(q)\in f(m)+\sum_i\partial_if(m)\,[-w_i,w_i]+\frac12\sum_iI_{ii}\,[0,w_i^2]+\sum_{i<j}I_{ij}\,[-w_iw_j,w_iw_j].
\end{gather*}
\item\label{enc:series} Let \(\sinc(x)=(\sin x)/x\) and \(\operatorname{sr}(x)=(x-\sin x)/x^3\), extended continuously to \(x=0\), and let \(\sinc_n\) and \(\operatorname{sr}_n\) be their Taylor polynomials of degree \(n\) at \(0\).  Then, for every real \(x\),
\[
\abs{\sinc(x)-\sinc_{36}(x)}\le\frac{x^{38}}{39!},\qquad
\abs{\operatorname{sr}(x)-\operatorname{sr}_{36}(x)}\le\frac{x^{38}}{41!},
\]
and, for \(\abs x\le2\),
\[
\abs{\sinc'-\sinc_{40}'}\le\frac{2^{41}}{43\cdot41!},\qquad
\abs{\sinc''-\sinc_{40}''}\le\frac{2^{40}}{43\cdot40!},
\]
\[
\abs{\operatorname{sr}-\operatorname{sr}_{40}}\le\frac{2^{42}}{45!},\qquad
\abs{\operatorname{sr}'-\operatorname{sr}_{40}'}\le\frac{42\cdot2^{41}}{45!},\qquad
\abs{\operatorname{sr}''-\operatorname{sr}_{40}''}\le\frac{42\cdot41\cdot2^{40}}{45!}.
\]
\item\label{enc:bernstein} Let \(p(y)=\sum_{j=0}^nc_j\,\cB^n_j(y)\).  Then \(\min_jc_j\le p(y)\le\max_jc_j\) for \(y\in[0,1]\).  Moreover, for \(0\le x_1\le x_2\le1\), we have \(p(x_1+(x_2-x_1)y)=\sum_{i=0}^nc'_i\,\cB^n_i(y)\) with \(c'_i=\sum_jM_{ij}c_j\), where \(M_{ij}\) is the coefficient of \(z^j\) in \(((1-x_1)+x_1z)^{n-i}((1-x_2)+x_2z)^i\); the numbers \(M_{ij}\) are nonnegative, and \(\sum_jM_{ij}=1\) for every \(i\).  The same holds, variable by variable, for polynomials in several variables written in a product basis such as \(\cB^{11}_i(a)\,\cB^{11}_j(b)\).
\item\label{enc:endpoint} Let \(f\) be a continuously differentiable function on a compact interval \(J=[x_1,x_2]\), and let \(y_-\le y_+\) be such that \(J\cap[y_-,y_+]\ne\emptyset\).  If \(f'\ge0\) on \(J\), then the maximum of \(f\) on \(J\cap[y_-,y_+]\) is \(f(\min(x_2,y_+))\); if \(f'\le0\) on \(J\), it is \(f(\max(x_1,y_-))\).
\item\label{enc:quadratic} Let \(p\) be a polynomial of degree at most \(2\) on \(\R^n\), and let \(X\subset\R^n\) be a box.  Then the maximum of \(p\) on \(X\) is the largest of the values of \(p\) at the vertices of \(X\) and at the points \(x\) with the following property: \(x\) lies in the relative interior of a face \(F\) of \(X\) of positive dimension, and the restriction of \(p\) to the affine hull of \(F\) has a vanishing gradient at \(x\) and a nonsingular Hessian.
\end{enumerate}
\end{lemma}

\begin{proof}
\ref{enc:taylor} The segment from \(m\) to \(q\) lies in \(X\), and Taylor's formula with integral remainder gives
\begin{gather*}
f(q)=f(m)+\sum_ih_i\int_0^1\partial_if(m+uh)\,du,\\
f(q)=f(m)+\sum_i\partial_if(m)\,h_i+\sum_{i,j}h_ih_j\int_0^1(1-u)\,\partial_i\partial_jf(m+uh)\,du .
\end{gather*}
Let \(g_i\) be the integral in the first formula, and let \(g_{ij}\) be twice the integral in the second.  These numbers are averages of \(\partial_if\) and of \(\partial_i\partial_jf\) over the segment, with the densities \(1\) and \(2(1-u)\) on \([0,1]\); so \(g_i\in I_i\) and \(g_{ij}\in I_{ij}\), and \(g_{ij}=g_{ji}\) since \(\partial_i\partial_jf=\partial_j\partial_if\).  Collecting the terms with \(i<j\) and with \(i>j\) in the second formula gives the claimed expressions for \(f(q)\), and the inclusions follow from \(h_i\in[-w_i,w_i]\), \(h_i^2\in[0,w_i^2]\) and \(h_ih_j\in[-w_iw_j,w_iw_j]\).

\ref{enc:series} The Taylor polynomials of \(\sin\) of degrees \(38\) and \(40\) at \(0\) are \(x\sinc_{36}(x)\) and \(x-x^3\operatorname{sr}_{36}(x)\).  Since all derivatives of \(\sin\) are bounded by \(1\), Lagrange's form of the remainder gives \(\abs{\sin x-x\sinc_{36}(x)}\le\abs x^{39}/39!\) and \(\abs{\sin x-x+x^3\operatorname{sr}_{36}(x)}\le\abs x^{41}/41!\); dividing by \(\abs x\) and by \(\abs x^3\) gives the first two bounds.  For the others, we have \(\sinc(x)=\int_0^1\cos(ux)\,du\) and \(\operatorname{sr}(x)=\frac12\int_0^1(1-u)^2\cos(ux)\,du\).  Differentiating under the integral sign and replacing \(\sin y\) and \(\cos y\) by their Taylor polynomials, whose remainders after the terms of degree \(m\) are at most \(\abs y^{m+1}/(m+1)!\), gives the Taylor polynomials of \(\sinc',\sinc'',\operatorname{sr},\operatorname{sr}',\operatorname{sr}''\) with errors at most \(\int_0^1u\,(u\abs x)^{41}du/41!\), \(\int_0^1u^2(u\abs x)^{40}du/40!\), \(\frac12\int_0^1(1-u)^2(u\abs x)^{42}du/42!\), \(\frac12\int_0^1(1-u)^2u\,(u\abs x)^{41}du/41!\) and \(\frac12\int_0^1(1-u)^2u^2(u\abs x)^{40}du/40!\), respectively.  Evaluating these integrals with \(\int_0^1(1-u)^2u^{42}du=2\cdot42!/45!\) gives the stated bounds.

\ref{enc:bernstein} The polynomials \(\cB^n_j\) are nonnegative on \([0,1]\) and sum to \(1\), which gives the first claim.  For the second claim, which is the subdivision formula for Bernstein coefficients~\cite{farouki2012}, let \(x=x_1+(x_2-x_1)y\).  Then \((1-x)+xz=(1-y)\bigl((1-x_1)+x_1z\bigr)+y\bigl((1-x_2)+x_2z\bigr)\), and expanding the \(n\)-th powers of both sides by the binomial theorem gives \(\sum_j\cB^n_j(x)\,z^j=\sum_i\cB^n_i(y)\sum_jM_{ij}z^j\).  Comparing the coefficients of \(z^j\), we obtain \(\cB^n_j(x)=\sum_iM_{ij}\,\cB^n_i(y)\), and hence \(p(x)=\sum_ic'_i\,\cB^n_i(y)\).  The numbers \(M_{ij}\) are nonnegative since \(x_1,x_2\in[0,1]\), and their sum over \(j\) is the value of the defining polynomial at \(z=1\), which is \(1\).  The statement for several variables follows by applying these facts in each variable.

\ref{enc:endpoint} The set \(J\cap[y_-,y_+]\) is the interval \([\max(x_1,y_-),\min(x_2,y_+)]\), and a nondecreasing (nonincreasing) function attains its maximum on an interval at the right (left) end.

\ref{enc:quadratic} Let \(x^*\) be a point of \(X\) at which \(p\) attains its maximum, and let \(F\) be the face of \(X\) that contains \(x^*\) in its relative interior.  If \(F\) is a vertex, there is nothing to prove.  Otherwise, the restriction of \(p\) to the affine hull of \(F\) has a vanishing gradient at \(x^*\).  If its Hessian is singular, let \(v\ne0\) be a vector in its kernel, parallel to \(F\); then \(p(x^*+tv)=p(x^*)\) for all \(t\in\R\), and moving along \(v\) to the relative boundary of \(F\) gives a maximum point in a face of smaller dimension.  Repeating this argument proves the claim.
\end{proof}

In the computations with Arb, the values of \(\sin\), \(\cos\), \(\sinc\) and \(\pi\) are enclosed by the functions of Arb, and the derivatives of \(\sinc\) and the function \(\operatorname{sr}\) and its derivatives by the Taylor polynomials of degree \(40\) with the bounds of \cref{lem:enclosures}\ref{enc:series}; their arguments lie in \([-2,2]\) (\cref{subsec:cert-h4}).  Part~\ref{enc:quadratic} is used together with the expression for \(f(q)\) in part~\ref{enc:taylor}, as follows.  On a sub-box of \(X\) on which the coordinates of \(h=q-m\) have fixed signs, the monomials \(h_i\), \(h_i^2\) and \(h_ih_j\) have fixed signs.  So \(f(q)\) is at most the quadratic polynomial in \(h\) obtained from the second expression for \(f(q)\) by replacing \(f(m)\), \(\partial_if(m)\), \(g_{ii}\) and \(g_{ij}\) by the upper or the lower ends of intervals that contain them, according to the signs of the corresponding monomials.  This polynomial has rational coefficients, and its maximum on the sub-box is computed exactly by part~\ref{enc:quadratic}, by solving the linear systems for the critical points of all faces in rational arithmetic.

\subsection{The Certificates of \texorpdfstring{\cref{prop:h0h3}}{Proposition 5.11}}\label{subsec:cert-h0h3}

In this subsection, we describe the checks behind the four inequalities in the proof of \cref{prop:h0h3}; \cref{tab:recorded} lists the numbers that they recorded.  The program reconstructs all coefficients from the stored ones by~\eqref{eq:sym} and checks that they satisfy~\eqref{eq:sym} and~\eqref{eq:gamma-zero} in exact rational arithmetic.  The program then verifies the four inequalities of \cref{prop:h0h3} as follows.  We first describe \ref{h03:R} and \ref{h03:h2}, which are checked in exact arithmetic, and then \ref{h03:h1} and \ref{h03:h3} together, as they involve only \(\gamma\) and need the same rewriting near \(\gamma=\frac\pi2\).

\ref{h03:R} The B-splines \(N_l\) are nonnegative and sum to \(1\) on \([0,1]\), and so are the Bernstein polynomials.  Hence, by \cref{lem:enclosures}\ref{enc:bernstein}, \(\abs R\) is at most the largest absolute value of its coefficients with respect to the products \(\cB^5_i(a)\,\cB^5_j(b)\,N_l(\tau)\), after one exact subdivision in \(a\) and one in \(b\).  This value is less than \(\frac1{50}\) (\cref{tab:recorded}).

\begin{table}[!htbp]
\centering\small
\begin{tabularx}{\textwidth}{@{}P{3.3cm}YP{3.6cm}@{}}
\toprule
 Check & Partition and recorded value & Meaning\\
\midrule
\ref{h03:R} & largest absolute Bernstein coefficient of \(R\) after one subdivision in \(a\) and one in \(b\): about \(0.016059\) & \(\abs R<\frac1{50}=0.02\)\\
\ref{h03:h2} & per sign, \(2\cdot64\cdot132\cdot32=540{,}672\) cubic polynomials in \(\tau\) (two derivatives, \(64\) squares of \((a,b)\), \(132\) Bernstein coefficients, \(32\) spline intervals), partitioned into \(549{,}301\) intervals; smallest Bernstein coefficient about \(5.35\times10^{-10}\) & all coefficients positive\\
\ref{h03:h1}, \ref{h03:h3} & \(2{,}152\) intervals of \(\tau\), from the \(32\) spline intervals bisected to depth at most \(8\); \(70\) digits, repeated with \(90\) & verified in the rewritten forms of \cref{subsec:cert-h0h3}\\
box of the main cover in \(Q_0\) & upper bound for \(E\) (with \(150\) digits): \(6.4\times10^{-11}\) & positive; not accepted\\
two halves of \(Q_0\) & upper bounds for \(E\): \(-1.47\times10^{-6}\) and \(-1.15\times10^{-6}\) & accepted by~\ref{acc:direct}\\
the \(1{,}499\) boxes~\eqref{eq:corner-boxes} & \(\underline\xi\ge1-1.17\times10^{-6}\); second inequality of~\eqref{eq:corner-hyp}, after division by \(t\): margin at least \(1.089\times10^{-7}\); first inequality: margin at least \(2.8\); \(\partial_\xi E\in[-1.0713,-1.0698]\) & accepted by~\ref{acc:corner}; \(C\) about \(1.07\)\\
\bottomrule
\end{tabularx}
\caption{Numbers recorded in the certificates.  The first three rows belong to \cref{prop:h0h3} (\cref{subsec:cert-h0h3}), the others to \cref{prop:h4} (\cref{subsec:cert-h4}); \(E=\cE^{\rm c}_{U_0}\).  In the last three rows, \(Q_0\) and the boxes~\eqref{eq:corner-boxes} are those of \cref{tab:cover}, \ref{acc:direct} and~\ref{acc:corner} are the acceptance conditions of \cref{subsec:cover} (a direct interval bound, and the corner lemma for the central piece), and \(t=\frac\pi2-\gamma\) as in \cref{subsec:cert-h4}.}
\label{tab:recorded}
\end{table}

\ref{h03:h2} On each spline interval, \(\partial_a(V\pm\frac25R)\) and \(\partial_b(V\pm\frac25R)\) are polynomials in \((\tau,a,b)\).  For each sign and each of the two derivatives, the program subdivides the square of \((a,b)\) exactly into \(8\times8\) squares and computes on each of them the \(132\) Bernstein coefficients in \((a,b)\) (\cref{lem:enclosures}\ref{enc:bernstein}).  On each spline interval, each of these coefficients is a cubic polynomial in \(\tau\), and the program bisects the spline interval exactly until the Bernstein coefficients in \(\tau\) of this polynomial are nonnegative on each part.  By \cref{lem:enclosures}\ref{enc:bernstein}, applied in \(\tau\) and in \((a,b)\), the two derivatives are then nonnegative (\cref{tab:recorded}).

\ref{h03:h1}, \ref{h03:h3} These inequalities involve only \(\gamma\), and some of them become equalities in the limit \(\gamma\to\frac\pi2\): there, \(\sigma\) and \(k\) tend to \(\frac\pi2\) and \(m_1\) and \(m_2\) tend to \(0\), so the upper bounds in~\eqref{eq:h1-coeffs} become equalities, and both sides of~\eqref{eq:h3-1d} tend to \(\frac\pi2\).  On intervals of \(\gamma\) that contain \(\frac\pi2\), interval arithmetic cannot verify such inequalities as they stand.  We therefore divide the differences of their two sides by \(\cos^2\gamma\) and by \(\cos\gamma\), respectively, which are positive for \(\gamma<\frac\pi2\), and write the quotients without division by \(t=\frac\pi2-\gamma\), using \(\cos\gamma=t\sinc(t)\) and \(k(\gamma)-\gamma=t-\sin t=t^3\operatorname{sr}(t)\).  Let \(m\) be one of \(m_1\), \(-\frac12m_1+\frac32m_2\) and \(-2m_1-2m_2\), so that \(\sigma+m\) is one of the three numbers inside the absolute values in~\eqref{eq:h1-coeffs}, and let \(\bar m=m/\cos^2\gamma\), with \(R(\gamma,a,b)\) replaced by \(\frac1{50}\) or by \(-\frac1{50}\); so \(\bar m\) is a linear combination of \(B(\gamma)\pm\frac1{50}\) and \(D(\gamma)\).  Since \(\sigma=\gamma-\cos^2\gamma\,V(\gamma,0,1)\), the inequality \(\sigma+m\le k(\gamma)\) is equivalent, for \(\gamma<\frac\pi2\), to
\[
V(\gamma,0,1)-\bar m+\frac{t\operatorname{sr}(t)}{\sinc(t)^2}\ge0,
\]
whose left-hand side is the difference of the two sides divided by \(\cos^2\gamma\).  The inequality \(\sigma+m\ge-k(\gamma)\) does not become an equality in the limit, and it is checked in the form \(k(\gamma)+\gamma-\cos^2\gamma\,(V(\gamma,0,1)-\bar m)\ge0\).  Similarly, dividing the difference of the two sides of~\eqref{eq:h3-1d} by \(\cos\gamma\) and using \(\frac\pi2k(\gamma)-\gamma^2=t(\frac\pi2+\gamma)-\frac\pi2\cos\gamma\), we obtain the equivalent inequality
\begin{multline*}
U_0-\frac1{k(\gamma)}\Bigl(\frac{\frac\pi2+\gamma}{\sinc(t)}-\frac\pi2\Bigr)\\
+\cos\gamma\Bigl[V(\gamma,0,0)-\frac{\gamma\,V(\gamma,0,1)}{k(\gamma)}+B(\gamma)\bigl(\xi-\xi^3\bigr)+D(\gamma)\bigl(\xi^3-\xi^5\bigr)\Bigr]\ge0,
\end{multline*}
where \(\xi=\gamma/k(\gamma)\); its left-hand side tends to \(U_0-1\) as \(\gamma\to\frac\pi2\).  These expressions contain no division by \(t\), so interval arithmetic can evaluate them on intervals of \(\gamma\) that contain \(\frac\pi2\).  They are verified on a partition of \(\tau\in[0,1]\) obtained by bisecting each spline interval recursively (\cref{tab:recorded}), with the precisions of \cref{tab:computations}.  On each interval, \(V\), \(B\) and \(D\) are enclosed by the Bernstein coefficients of their polynomials in \(\tau\) (\cref{lem:enclosures}\ref{enc:bernstein}), and \(\sinc\) and \(\operatorname{sr}\) by the Taylor polynomials of degree \(36\) with the bounds of \cref{lem:enclosures}\ref{enc:series}.

\subsection{The Cover of \texorpdfstring{\cref{prop:h4}}{Proposition 6.7}}\label{subsec:cert-h4}

In this subsection, we describe how conditions~\ref{acc:empty}--\ref{acc:corner-outer} are verified on a box, and we give the numbers recorded in the two covers.

For a box \(Q\), the programs compute interval enclosures of all functions of \((\gamma,a,b)\) involved: the B-spline polynomials of the spline interval of \(Q\) and their derivatives in \(\tau\), the Bernstein coefficients of \(V\), \(R\) and their derivatives restricted to the \((a,b)\)-ranges of \(Q\) (\cref{lem:enclosures}\ref{enc:bernstein}), and \(\varphi\), \(\theta\), \(k\) and their trigonometric functions.  Derivatives are computed by forward-mode automatic differentiation in interval arithmetic.  The expressions that contain \(\cos\gamma\) or \(\pi-2\gamma\) in a denominator are rewritten in terms of \(t=\frac\pi2-\gamma\) and the functions \(\sinc\) and \(\operatorname{sr}\) at multiples of \(t\), with arguments in \([-2,2]\), so that they extend continuously to the boxes with \(\tau^+=1\); for instance, \(\cos^2\gamma/(\pi-2\gamma)=\frac12t\sinc(t)^2\), and, since \(\varphi=\frac\pi2-\theta-t(1-2b)\) by~\eqref{eq:ab},
\[
\cA=\frac{2tb(1-b)\cos\theta\,\sinc(tb)\,\sinc(t(1-b))-(1-2b)\sin\theta\,\sinc(t(1-2b))}{\sinc(t)} .
\]
On a box, each B-spline is replaced by its polynomial on the spline interval of the box; so all functions involved are infinitely differentiable on a neighborhood of the box, and \cref{lem:enclosures}\ref{enc:taylor} applies to them.

\paragraph{Conditions~\ref{acc:empty} and~\ref{acc:direct} for the Outer Pieces.}  By~\eqref{eq:sigma-bar}, there is a point of \(\overline\Sigma\) in the piece \(\xi\ge1\) over \((\gamma,a,b)\) if and only if \(\varphi+\sin\theta-k(\gamma)\ge0\), and similarly \(-\varphi+\sin\theta-k(\gamma)\ge0\) for the piece \(\xi\le-1\).  So~\ref{acc:empty} holds if an upper bound for this expression on \(Q\) is negative.  For~\ref{acc:direct}, the residual of the piece, which does not depend on \(\xi\), is bounded above on \(Q\) by direct evaluation or by the inclusions of \cref{lem:enclosures}\ref{enc:taylor}.  Since \(\pm\varphi+\sin\theta-k\) is nonnegative at the points of \(\overline\Sigma\) in the piece, an upper bound on \(Q\) for the residual plus \(\lambda(\pm\varphi+\sin\theta-k)\), for a rational \(\lambda\ge0\), is also an upper bound for the residual at these points; such multipliers \(\lambda\) sharpen the bound.

\paragraph{Conditions~\ref{acc:empty} and~\ref{acc:direct} for the Central Piece.}  The residual \(\cE^{\rm c}_{U_0}\) is a polynomial of degree \(5\) in \(\xi\) whose coefficients are functions of \((\gamma,a,b)\).  By~\eqref{eq:sigma-bar}, the values of \(\xi\) at the points of \(\overline\Sigma\) in the central piece over \((\gamma,a,b)\) form the interval \([\max(-1,\xi_-),\min(1,\xi_+)]\), where \(\xi_\pm=(\varphi\pm\sin\theta)/k(\gamma)\).  From enclosures of \(\xi_\pm\) on \(Q\), the program computes numbers \(\xi^-\ge-1\) and \(\xi^+\le1\) such that all these values, for all \((\gamma,a,b)\) over \(Q\), lie in \([\xi^-,\xi^+]\); if \(\xi^->\xi^+\), there are no such values, and~\ref{acc:empty} holds.  Otherwise, \ref{acc:direct} is verified on a subdivision of \([\xi^-,\xi^+]\) into intervals \(J=[\xi_1,\xi_2]\): for each \(J\), one of the following upper bounds for the residual at the points of \(\overline\Sigma\) in the central piece over \(Q\) with \(\xi\in J\) is at most \(0\).
\begin{enumerate}[label=(\roman*)]
\item The largest of the Bernstein coefficients in \(\xi\) of \(\cE^{\rm c}_{U_0}\) on \(J\), enclosed on \(Q\) by direct evaluation or by the inclusions of \cref{lem:enclosures}\ref{enc:taylor}.
\item If \(\partial_\xi\cE^{\rm c}_{U_0}\ge0\) on \(Q\times J\), an upper bound on \(Q\) for \(\cE^{\rm c}_{U_0}(\gamma,a,b,\min(\xi_2,\xi_+))\), and if \(\partial_\xi\cE^{\rm c}_{U_0}\le0\) on \(Q\times J\), one for \(\cE^{\rm c}_{U_0}(\gamma,a,b,\max(\xi_1,\xi_-))\); these are upper bounds by \cref{lem:enclosures}\ref{enc:endpoint}.  The cases of the minimum or of the maximum may be treated separately, and in each case a nonnegative multiple of an expression that is nonnegative in that case may be added, as for the outer pieces.
\item The maximum on \(Q\times J\) of the quadratic upper bound described after \cref{lem:enclosures}, computed exactly.
\end{enumerate}

\paragraph{Conditions~\ref{acc:corner} and~\ref{acc:corner-outer}.}  In the main cover, these conditions are used for boxes \(Q\) with \(a^+\le\frac18\) and either \(b^-\ge\frac78\), near the corner curve, or \(\tau^-\ge\frac{31}{32}\), near the edge \(a=0\) as \(\gamma\to\frac\pi2\) (\cref{rem:edge}), with \(\widehat Q=[\tau^-,\tau^+]\times[0,a^+]\times[b^-,1]\); for boxes with \(a^-\ge\frac78\), they are used in the same way for the reflection \(\rho_2(Q)\).  In the supplementary cover, each box~\eqref{eq:corner-boxes} is its own \(\widehat Q\).

For the central piece, \(\underline\xi\) is the larger of \(-1\) and a lower bound for \(\xi_-\) on \(\widehat Q\); it is at most the value of \(\xi\) at every point of \(\overline\Sigma\) in the central piece over \(Q\).  The number \(C\) is the larger of \(0\) and an upper bound for \(-\partial_\xi\cE^{\rm c}_{U_0}\) on \(\widehat Q\times[\underline\xi,1]\), and \(\overline c\) and \(\underline c\) are an upper and a lower bound for \(\cos\theta\) on \(\widehat Q\).  The program checks~\eqref{eq:corner-hyp} on \(\widehat Q\times[\underline\xi,1]\), with interval enclosures of \(\partial_a\cE^{\rm c}_{U_0}\) and \(\partial_b\cE^{\rm c}_{U_0}\) and with \(\cos\theta_0\) and \(\cos\theta_1\) replaced by \(\overline c\) and \(\underline c\).  For every \(\gamma\) with \(2\gamma/\pi\in[\hat\tau^-,\hat\tau^+]\), the angles \(\theta_0\) and \(\theta_1\) of \cref{lem:corner} are values of \(\theta\) on \(\widehat Q\), so \(\overline c\ge\cos\theta_0\) and \(\underline c\le\cos\theta_1\), and the checked inequalities imply~\eqref{eq:corner-hyp}.  The second inequality is checked after division by \(t=\frac\pi2-\gamma\), since both of its sides tend to \(0\) as \(\gamma\to\frac\pi2\).

\paragraph{The Quotient \(\partial_b\cE^{\rm c}_{U_0}/t\) Without Division.}  On boxes of the main cover that reach \(\gamma=\frac\pi2\), the quotient cannot be formed in interval arithmetic, since \(t\) may be \(0\); we therefore expand \(\cE^{\rm c}_U\) in powers of \(t\) and read off the quotient from the expansion.  On the central piece, \(\widehat H=\gamma\xi+\cos^2\gamma\,Y\), where
\[
Y=V(\gamma,a,b)-V(\gamma,0,1)\,\xi+\bigl(B(\gamma)+R(\gamma,a,b)\bigr)\bigl(\xi-\xi^3\bigr)+D(\gamma)\bigl(\xi^3-\xi^5\bigr),
\]
and \cref{def:residual} and the chain rule, as in the proof of \cref{lem:corner-identity}, give
\begin{align*}
\cE^{\rm c}_U&=(1-U)\cos\varphi+t\,q_1+t^2q_2+t^3q_3,\qquad\text{where}\\
q_1&=(2b-1)\Bigl[\cos\Bigl(\Bigl(\frac\pi2+\gamma\Bigr)a+\frac t2\Bigr)\sinc\Bigl(\frac{(2b-1)t}2\Bigr)-\cos\theta\Bigr]\\
&\qquad+2\sinc(t)\sin\gamma\,\cA\,Y+\frac{\sinc(t)^2}2\bigl(\sin\theta-(2b-1)\cA\bigr)\,\partial_bY,\\
q_2&=-\frac12\sinc\Bigl(\frac t2\Bigr)^2\cA\,\xi+\sinc(t)^2\Bigl[-Y\cos\theta-\cA\,\partial_\gamma Y+\frac{\sin\theta-(1-2a)\cA}{\pi+2\gamma}\,\partial_aY\\
&\qquad+\Bigl(\Bigl(\xi-\frac\varphi k\Bigr)\cos\theta+\frac{\xi\sin\gamma}k\,\cA\Bigr)\partial_\xi Y\Bigr],\\
q_3&=\frac{\operatorname{sr}(t)}k\bigl(\varphi\cos\theta-\xi\sin\gamma\,\cA\bigr),
\end{align*}
with \(k=k(\gamma)\) and the derivatives of \(Y\) taken in the coordinates \((\gamma,a,b,\xi)\).  The terms without \(Y\) are those of \(\cE_U[\gamma\xi]=\cos\theta\,(\frac\pi2-\theta-\gamma\varphi/k)+\sin\theta-U\cos\varphi-\cA\,\xi\,(1-\gamma\sin\gamma/k)\); they give \((1-U)\cos\varphi\), the first terms of \(q_1\) and \(q_2\), and \(q_3\), since \(\frac\pi2-\theta-\gamma\varphi/k=t(1-2b)+t^3\operatorname{sr}(t)\,\varphi/k\), \(1-\gamma\sin\gamma/k=\frac12t^2\sinc(\frac t2)^2+t^3\operatorname{sr}(t)\sin\gamma/k\) and \(\sin\theta-\cos\varphi=t(2b-1)\cos((\frac\pi2+\gamma)a+\frac t2)\sinc(\frac{(2b-1)t}2)\), by~\eqref{eq:ab}, \(k-\gamma=t^3\operatorname{sr}(t)\) and \(1-\sin\gamma=\frac12t^2\sinc(\frac t2)^2\).  In the terms with \(Y\), the factor \(\cos^2\gamma\) of \(Y\) in \(\widehat H\) gives \(2\cos\gamma\sin\gamma=2t\sinc(t)\sin\gamma\) through its derivative, \(\cos^2\gamma/(\pi-2\gamma)=\frac12t\sinc(t)^2\) in the terms with \(\partial_bY\), and \(\cos^2\gamma=t^2\sinc(t)^2\) in the others.  Since \(\cos\varphi=\sin((\frac\pi2+\gamma)a+t(1-b))\) by~\eqref{eq:ab}, we obtain
\[
\frac{\partial_b\cE^{\rm c}_U}t=(U-1)\cos\Bigl(\Bigl(\frac\pi2+\gamma\Bigr)a+t(1-b)\Bigr)+\partial_bq_1+t\,\partial_bq_2+t^2\partial_bq_3,
\]
which extends continuously to \(\gamma=\frac\pi2\).  In the supplementary cover, \(t>1.16\) on all boxes, and the division is carried out in interval arithmetic; there, each enclosure of a derivative is also intersected with the first inclusion of \cref{lem:enclosures}\ref{enc:taylor}, applied to this derivative as a function of the four variables \((\tau,a,b,\xi)\), which is sharper on small boxes.

\paragraph{Condition~\ref{acc:corner-outer}.}  For the piece \(\xi\ge1\), the program checks that an upper bound for \(\partial_a\cE^{+}_{U_0}\) on \(\widehat Q\) is at most \(0\), and that a lower bound for the derivative in \(b\) of \(\cE^{+}_{U_0}(\gamma,0,b)/t\) on the edge \(a=0\) of \(\widehat Q\) is at least \(0\); for \(\gamma<\frac\pi2\), these are the hypotheses of \cref{lem:corner-outer}.  On this edge, \(\theta=tb\), \(\varphi=\gamma+\theta\) and \(\cA=\sin\theta\) by~\eqref{eq:ab}, and a computation as above gives, with \(V_0(\gamma)=V(\gamma,0,1)\) and with \(V\), \(\partial_bV\) and \(\partial_\gamma V\) evaluated at \((\gamma,0,b)\),
\begin{multline*}
\frac{\cE^{+}_U(\gamma,0,b)}t=(1-b)\bigl(\cos(tb)-U\sinc(t(1-b))\bigr)+t\Bigl[\bigl(2b\sinc(t)\sin\gamma\,\sinc(tb)-\sinc(t)^2\cos(tb)\bigr)(V-V_0)\\
+b(1-b)\sinc(t)^2\sinc(tb)\,\partial_bV\Bigr]+t^2\,b\sinc(t)^2\sinc(tb)\,(V_0'-\partial_\gamma V),
\end{multline*}
which extends continuously to \(\gamma=\frac\pi2\).

\paragraph{Recorded Bounds.}  \Cref{tab:recorded} lists the numbers recorded in the certificates; they are the ones quoted in the proof of \cref{prop:h4}.  The box \(Q_0\) of that proof is \([\frac{205}{1024},\frac{103}{512}]\times[\frac{367}{512},\frac{23}{32}]\times[\frac{253}{256},\frac{507}{512}]\), and its two halves are obtained by bisecting at \(\tau=\frac{411}{2048}\).

\subsection{Reproduction}\label{subsec:reproduction}

The data and programs~\cite{mcd2026package} are archived at \doi{10.5281/zenodo.23835635}.  The archive also contains a description of the file formats and of the checks, the versions of the software used, and the outputs of a complete run; it refers to \cref{prop:h0h3,prop:h4} by the conditions \Hcond{0}--\Hcond{3} and \Hcond{4} that they establish.  A script runs the four commands below in order and exits with status zero only if every check succeeds.

The programs reproduce the verification with four commands, run in this order.
\begin{enumerate}[label=(\arabic*)]
\item The first command recomputes the certificate of \cref{prop:h0h3} from the coefficient file.
\item The second rebuilds the main cover from the recorded bisections and repeats the verification of its accepted boxes with \(90\) digits (\cref{subsec:arithmetic}); it reports whether it confirms all these acceptances.
\item The third checks the containments of the boxes of the main cover that are not accepted (\cref{subsec:cover}), recomputes the supplementary cover with \(150\) digits, and, as a check of the implementation, evaluates the residual at \(64\) points near the corner curve and its image under \(\rho_2\) with a separate program that evaluates \(\Hs\) directly from \cref{def:family}.
\item The fourth checks that the outputs of the first three commands belong to the same data, programs and level, that the two covers leave no box unaccepted, and that \(U_0>\frac\pi2\); it then writes a summary of the verification.
\end{enumerate}
The second, third and fourth commands load the main cover into memory, which takes less than \(1\) gigabyte.

\paragraph{AI methodology.}
We used ChatGPT (OpenAI) as a computational tool in the design of the
parametrized function family of \cref{sec:family} and the numerical
search for the coefficients of \(\Hs\), and in preparing this
manuscript.  The function construction and verification are described
in \cref{sec:family,sec:verification}, with reproduction details in
\cref{app:computation}.
The authors developed and verified the mathematical framework and
proofs and take full responsibility for the correctness and originality
of all content, including references.

\end{document}